\documentclass[11pt]{article}
\usepackage{amssymb}
\usepackage{amsmath}
\usepackage{amsfonts}
\usepackage{amsthm}
\usepackage{geometry}
\usepackage{setspace}
\usepackage[font={small}]{caption}
\usepackage{chbibref}
\usepackage{float}
\usepackage{color}
\usepackage{appendix}
\usepackage{natbib}
\usepackage{hyperref}
\usepackage{enumerate}
\usepackage{bbm}
\usepackage{comment}
\usepackage{mathtools}

\usepackage{booktabs}
\usepackage{multirow}
\usepackage{graphicx}
\usepackage{diagbox} 

\usepackage{subcaption}
\usepackage{pdflscape}

\usepackage{pstricks}
\newrgbcolor{dgreen}{0 0.5 0}

\hypersetup{
    colorlinks,
    citecolor=blue,
    filecolor=black,
    linkcolor=blue,
    urlcolor=black
}
\newtheorem{theorem}{Theorem}

\newtheorem{algorithm}{Algorithm}[section]
\newtheorem{assumption}{Assumption}

\newtheorem{definition}{Definition}[section]

\newtheorem{lemma}[theorem]{Lemma}

\theoremstyle{definition}

\newtheorem{remark}{Remark}[section]
\setbibref{References}{\sffamily}
\allowdisplaybreaks[1]

\newcommand{\plim}{\operatorname*{plim}}
\newcommand{\argmin}{\operatorname*{argmin}}
\DeclareMathOperator{\bias}{bias}
\DeclareMathOperator{\se}{se}

\DeclareMathOperator{\var}{var}

\newcommand{\maxbias}{\overline{\operatorname{bias}}}
\newcommand{\biasweight}{{b}}  %
\newcommand{\bTW}{{b^*}}  %

\newcommand*\norm[1]{\left\Vert #1\right\Vert}
\newcommand{\abs}[1]{\left\vert #1\right\vert}
\newcommand{\pr}{\mathbb{P}}
\newcommand{\rank}[1]{\text{rank}\left(#1\right)}

\newcommand{\maxSNT}{\max \{\sqrt{N},\sqrt{T}\}}
\newcommand{\minSNT}{\min \{\sqrt{N},\sqrt{T}\}}
\newcommand{\e}{\mathbb{E}}
\newcommand{\ex}[1]{\e\left[#1\right]}
\newcommand{\Op}{\mathcal O_{\Theta,\mathcal{P}}}
\newcommand{\op}{o_{\Theta,\mathcal{P}}}
\newcommand{\asy}{\asymp_{\Theta,\mathcal{P}}}
\newcommand{\dist}{\underset{\Theta,\mathcal{P}}{\overset{d}{\to}}}
\newcommand{\uto}{\underset{\Theta,\mathcal{P}}{\to}}
\newcommand{\uleq}{\underset{\Theta,\mathcal{P}}{\leqslant}}

\begin{document}

\title{\bf Robust Estimation and Inference\\in Panels with Interactive Fixed
  Effects\thanks{
    We thank the participants of the numerous seminars and conferences for helpful comments
and suggestions.  
    We also thank Riccardo D'Adamo and Chen-Wei Hsiang for their excellent research assistance.
    Any remaining errors are our own.
    Armstrong gratefully acknowledges support by the National
    Science Foundation Grant SES-2049765.
    Weidner gratefully acknowledges support through the European Research Council grant ERC-2018-CoG-819086-PANEDA.
    Zeleneev gratefully acknowledges the generous funding from the UK
Research and Innovation (UKRI) under the UK government’s Horizon Europe funding guarantee (Grant Ref:
EP/X02931X/1).
    }}
\author{\setcounter{footnote}{2}Timothy B. Armstrong\thanks{%
University of Southern California. Email: \texttt{timothy.armstrong@usc.edu} } 
\and 
Martin Weidner%
\thanks{%
University of Oxford. Email: \texttt{martin.weidner@economics.ox.ac.uk} } 
\and 
Andrei Zeleneev%
\thanks{%
University College London. Email: \texttt{a.zeleneev@ucl.ac.uk} } 
}
\date{May 2025}

\maketitle
\thispagestyle{empty}
\setcounter{page}{0}

\bigskip
\begin{abstract}
\noindent
We consider estimation and inference for a regression coefficient in panels with interactive fixed effects (i.e., with a factor structure). We demonstrate that existing estimators and confidence intervals (CIs) can be heavily biased and size-distorted when some of the factors are weak. We propose estimators with improved rates of convergence and bias-aware CIs that remain valid uniformly, regardless of factor strength. Our approach applies the theory of minimax linear estimation to form a debiased estimate, using a nuclear norm bound on the error of an initial estimate of the interactive fixed effects. Our resulting bias-aware CIs take into account the remaining bias caused by weak factors. Monte Carlo experiments show substantial improvements over conventional methods when factors are weak, with minimal costs to estimation accuracy when factors are strong.
\end{abstract}

\vskip 3cm

\newpage

\section{Introduction}

In this paper, we consider a linear panel regression model of the form
 \begin{align}
  Y_{it} = X_{it} \beta + \sum_{k=1}^K Z_{k,it}\delta_k +  \Gamma_{it} + U_{it} ,
  \label{model0}
\end{align}
where $Y_{it}, X_{it}, Z_{1,it}, \ldots, Z_{K,it} \in \mathbb{R}$  are the observed outcome variable and covariates for units $i=1,\ldots,N$ and time periods $t=1,\ldots,T$.
The error components  $\Gamma_{it}  \in \mathbb{R}$ and $U_{it}  \in \mathbb{R}$ are unobserved, and the regression coefficients   $\beta, \delta_1, \ldots, \delta_K \in \mathbb{R}$ are unknown.
The  parameter of interest is $\beta \in \mathbb{R}$, the  coefficient on $X_{it}$. We are interested in ``large panels'', where both $N$ and $T$ are relatively large.

The error component $U_{it}$ is modelled as a mean-zero random shock that is uncorrelated with the regressors $X_{it}$ and $Z_{k,it}$
and that is at most weakly autocorrelated across $i$ and over $t$. By contrast, the error component $\Gamma_{it}$ can be correlated with $X_{it}$ and $Z_{k,it}$
and can also be strongly autocorrelated across $i$ and over $t$. Of course, further restrictions on $\Gamma_{it}$ are required to allow estimation and inference on $\beta$.
For example, the additive fixed effect model imposes that $\Gamma_{it} = \alpha_i + \gamma_t$, where 
$\alpha_i$ accounts for any omitted variable that is constant over time, and $\gamma_t$ for any omitted variable that is constant across units.
Instead of this additive fixed effect model we consider the so-called interactive fixed effect model, where
 \begin{align}
     \Gamma_{it} =  \sum_{r=1}^R \,  \lambda_{ir}\, f_{tr} \; .
     \label{FactorModel}
\end{align}
Here, the $\lambda_{ir}$ and $f_{tr} $ can either be interpreted as unknown parameters or as unobserved shocks. This model for $\Gamma_{it} $ is also known as a factor model, with factor loadings $\lambda_{ir}$ and factors $f_{tr}$. We will use the terms factor and interactive fixed effect interchangeably. The number of factors $R$
is unknown, but is assumed to be small relative to $N$ and $T$. The interactive fixed effect model is attractive because it introduces enough restrictions to allow
estimation and inference on $\beta$ while still incorporating or approximating a large class of data generating processes (DGPs) for $ \Gamma_{it} $.

The existing econometrics literature on panel regressions with interactive fixed
effects is  quite large. Since the seminal work of \cite{Pesaran2006estimation} and
\cite{bai2009panel}, developing tools for estimation and inference on $\beta$ in
model~\eqref{model0}-\eqref{FactorModel} under large $N$ and large $T$
asymptotics has been a primary focus of this literature. Specifically, \citet{Pesaran2006estimation} introduces the common correlated effects (CCE) estimator, which uses cross-sectional averages of the observed variables as proxies for the unobserved factors. \citet{bai2009panel} derives the large $N$, $T$ properties of the least-squares (LS) estimator that jointly minimizes the sum of squared residuals over the regression coefficients, factors, and factor loadings.\footnote{This estimator was first introduced by \cite{Kiefer1980}.}

\cite{bai2009panel} shows that, under appropriate assumptions, the LS estimator for the regression coefficients is $\sqrt{NT}$-consistent and asymptotically normally distributed 
as both $N$ and $T$ grow to infinity. One of the key assumptions imposed for this result is the so-called ``strong factor assumption'', which requires all the factor loadings
$ \lambda_{ir}$ and factors $f_{tr}$ to have sufficient variation across $i$ and over $t$, respectively. If the strong factor assumption is violated, then the LS estimator 
for $ \lambda_{ir}$ and  $f_{tr}$  may be unable to pick up the true loadings and factors correctly, because the ``weak factors''\footnote{
See, for example, \cite{Onatski2010,Onatski2012} for a discussion and formalization of the notion of weak factors.
}  in $ \Gamma_{it}$ cannot be 
distinguished from the noise in $U_{it}$.
This can lead to substantial bias and misleading inference, due to omitted
variables bias from $\Gamma_{it}$ that is not picked up by the estimator.

\begin{figure}[tb]
    \begin{center}
      \begin{subfigure}{0.49\textwidth}
        \includegraphics[width=1.0\linewidth]{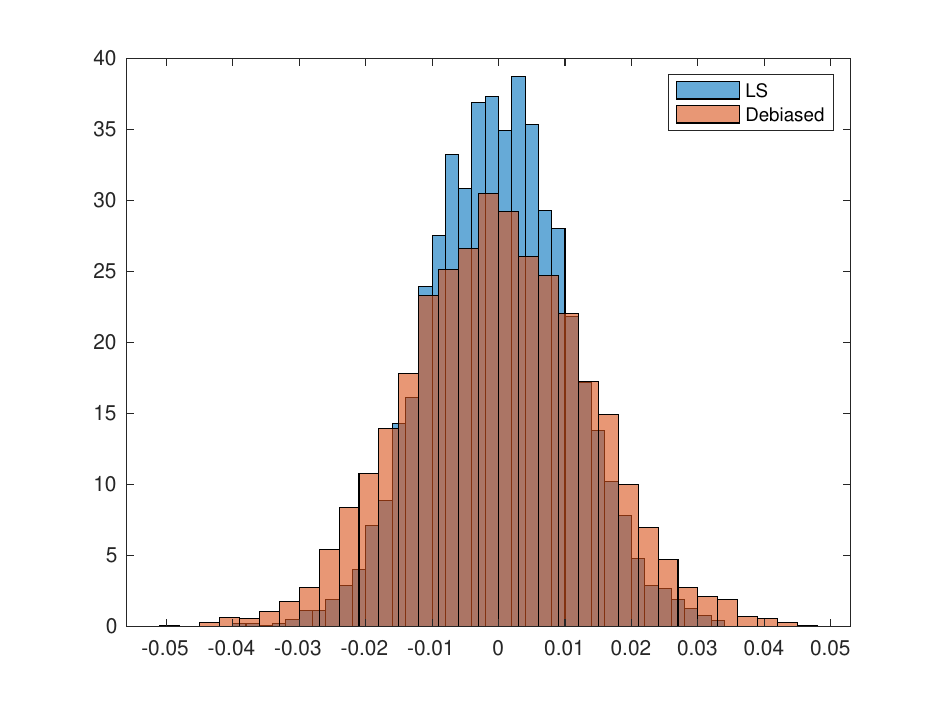}
        \caption{No ID ($\kappa = 0.00$)}
      \end{subfigure}
      \begin{subfigure}{0.49\textwidth}
        \includegraphics[width=1.0\linewidth]{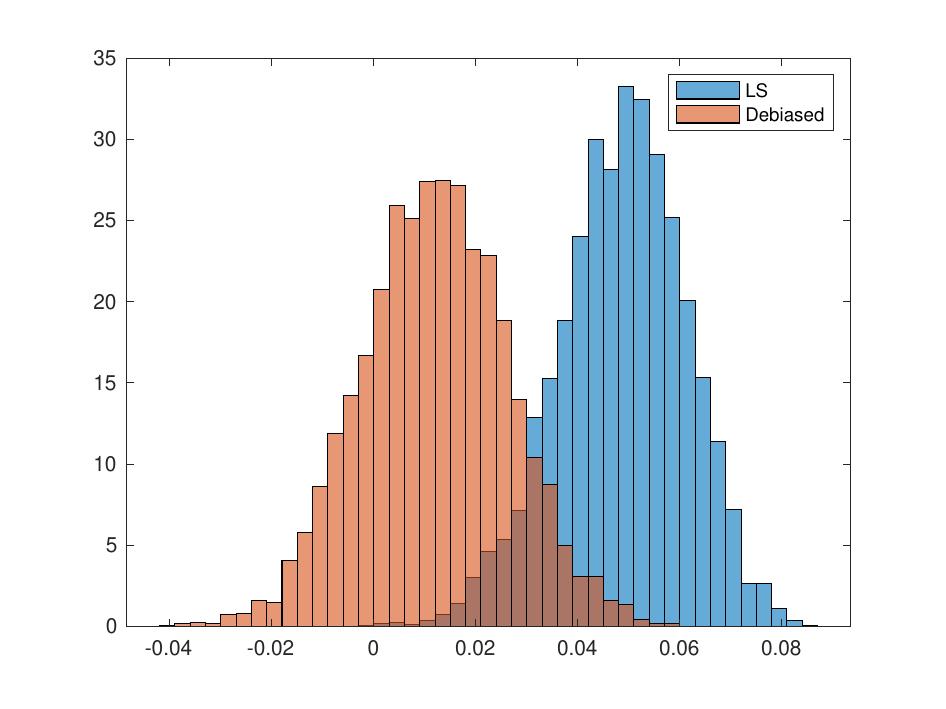}
        \caption{Weak ID ($\kappa = 0.10$)}
      \end{subfigure}
      \begin{subfigure}{0.49\textwidth}
        \includegraphics[width=1.0\linewidth]{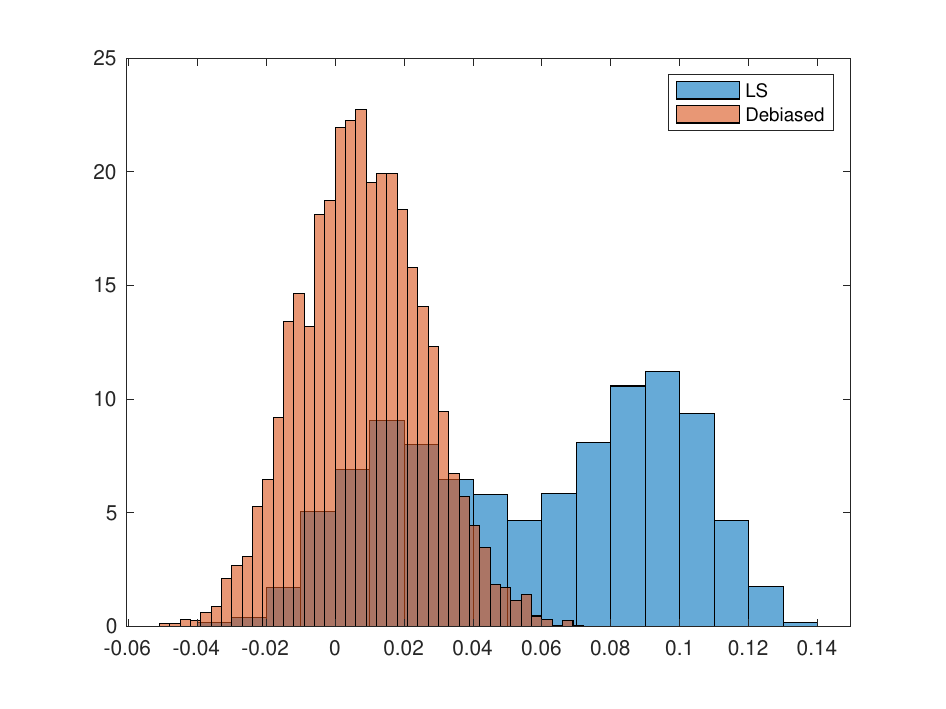}
        \caption{Weak ID ($\kappa = 0.20$)}
      \end{subfigure}
      \begin{subfigure}{0.49\textwidth}
        \includegraphics[width=1.0\linewidth]{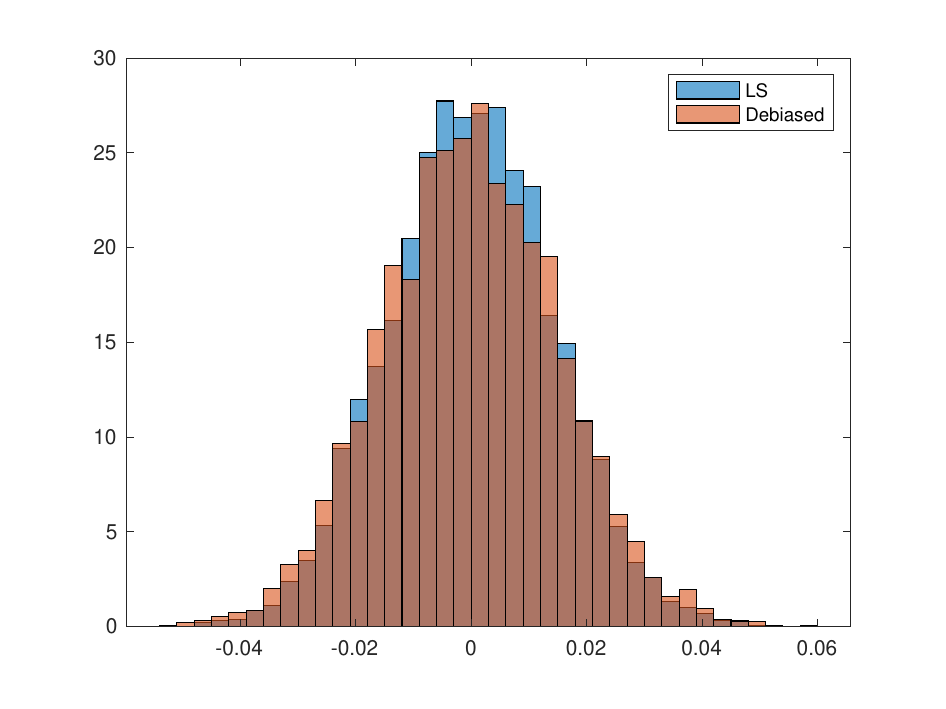}
        \caption{Strong ID ($\kappa = 1.00$)}
      \end{subfigure}
    \end{center}
  \caption{Finite sample distributions of the LS and the debiased estimators, $N=100$, $T=50$, $R=1$}
  \label{fig: intro}
\end{figure}

To illustrate how this can lead to problems with conventional estimates and CIs
for $\beta$, Figure \ref{fig: intro} presents a subset of the results of our Monte Carlo
study.\footnote{A detailed description of the numerical experiment is provided
  in Section \ref{ssec: MC}.}  
When the factors are nonexistent (panel a) or strongly identified
(panel d), the distribution of the LS estimator (in blue) is centered at the
true parameter value $\beta$ (equal to $0$ in this case).  However, when the
factors are present but weak enough that they are difficult to estimate (panels b
and c), the LS estimator is heavily biased and non-normally distributed.  In our
Monte Carlo study in Section \ref{sec: numerical}, we show that this indeed leads to
severe coverage distortion, with conventional CIs based on the LS estimator
having almost zero coverage.

In this paper, we address this issue by developing new tools for estimation and
inference on $\beta$ in the model \eqref{model0}.
We develop a debiased estimator along with a bound on the remaining bias, which
we use to construct a bias-aware confidence interval.
As illustrated in Figure \ref{fig: intro}, our debiased estimator (shown in red)
substantially decreases the bias of the LS estimator when factors are weak,
leading to a large improvement in overall estimation error.
In addition, this improved performance under weak factors does not come at a
substantial cost to performance when factors are strong or nonexistent: our
debiased estimator performs similarly to the LS estimator in these cases.
Importantly, our CI
requires only an upper bound on the number of factors: we show that it is valid
uniformly over a large class of DGPs that allows for weak, strong or nonexistent
factors up to a specified upper bound on the number of factors.  We derive rates
of convergence that hold uniformly over this class of DGPs, and we show that our
estimator achieves a faster uniform rate of convergence than existing approaches
when weak factors are allowed.  In the case where $N$ and $T$ grow at the same
rate, our estimator achieves the parametric $\sqrt{NT}$ rate.

Our debiasing approach uses a preliminary estimate $\hat\Gamma_{\rm pre}$ of the
individual effect matrix $\Gamma$ along with a bound $\hat C$ on the nuclear norm $\Vert \Gamma - \hat
\Gamma_{\rm pre}\Vert_*$ of its estimation error.
Letting $\tilde \Gamma:= \Gamma - \hat \Gamma_{\rm pre}$, we then consider
the augmented outcomes
\begin{align*}
  \tilde Y_{it} := Y_{it} - \hat \Gamma_{{\rm pre}, it} = X_{it} \beta + \sum_{k=1}^K Z_{k,it}\delta_k + \tilde \Gamma_{it} + U_{it}.
\end{align*}
Treating $\tilde \Gamma_{it}$ as nuisance parameters satisfying a convex
constraint $\Vert \tilde \Gamma \Vert_{*} \le \hat C$, we derive linear weights
$A_{it}$ such that the estimator $\sum_{i=1}^N\sum_{t=1}^TA_{it}\tilde Y_{it}$   for $\beta$
optimally uses this constraint, using the theory of minimax linear estimators \citep[see][]{ibragimov_nonparametric_1985,donoho1994statistical,armstrong2018optimal}.
In particular, the resulting weights $A_{it}$ control the remaining omitted
variables bias $\sum_{i=1}^N\sum_{t=1}^TA_{it}\tilde \Gamma_{it}$ due to
possible weak factors in $\tilde \Gamma=\Gamma-\hat\Gamma_{\rm pre}$ not picked
up by the initial estimate $\hat\Gamma_{\rm pre}$.

A key step in deriving our CI is the construction of the preliminary estimator
$\hat\Gamma_{\rm pre}$ and bound $\hat C$ on the nuclear norm of its estimation
error.  Our CI is bias-aware: it uses the bound $\hat C$ to explicitly take into
account any remaining bias in the debiased estimator.  Our bound is feasible
once an upper bound on the number of factors is specified.  In our Monte Carlo
study, we find that, while our CIs are often conservative, they are about as
wide as an ``oracle'' CI that uses an infeasible critical value to correct the
coverage of a CI based on the standard LS estimator.

While our results allow for arbitrary sequences of weak factors, our conditions on other aspects of the model are similar to \citet{bai2009panel} and \citet{MoonWeidner2015}.  An important condition is that the covariate of interest $X_{it}$ must not itself be entirely explained by a low dimensional factor model.
For example, in a panel where $X_{it}$ is the minimum hourly wage in state $i$ and year $t$, we would require that states change their minimum wage laws in different years, and that this is done sufficiently often to generate variation in $X_{it}$ that cannot be explained by a small number of factors $f_t$.
This rules out settings where $X_{it}$ is an indicator variable for a policy that affects a subset of the units and occurs only during a single time period: in this case, $X_{it}=\lambda_i\cdot f_t$ where $\lambda_i$ is an indicator variable for unit $i$ undergoing the policy change and $f_t$ is an indicator variable for periods after the policy change.
See Section \ref{asymptotic_section} for formal conditions and further discussion.

A special case of the factor model is the grouped unobserved heterogeneity model
considered by \citet{bonhomme2015grouped}.  In this model,
$\Gamma_{it}=\alpha_{g(i),t}$, where $g(\cdot)$ is an unknown function mapping
individuals $i$ to a group index $g(i)\in \{1,\ldots, R\}$.  This takes the form
of the factor model (\ref{FactorModel}) with $\lambda_{ir}=1$ if $g(i)=r$ and
$0$ otherwise, and with $f_{tr}=\alpha_{r,t}$.  The strong factor assumption
corresponds to the strong group separation assumption imposed in this literature
(e.g., Assumption 2(b) in
\citealp{bonhomme2015grouped}) which imposes that the group means $\alpha_{r,\cdot}=(\alpha_{r,1},\ldots,\alpha_{r,T})'$
are sufficiently far away for different groups $r$.  Our results apply in this
setting and allow for this assumption to be relaxed.  An interesting question
for future research is whether it is possible to modify our approach to take
advantage of the additional structure in the grouped unobserved heterogeneity
model.

\bigskip

{\noindent \bf Related literature}

\noindent
The papers by \citet{Pesaran2006estimation} and \citet{bai2009panel} mentioned previously have motivated a large follow up literature on large $N$ and $T$ analysis of panel models with interactive effects.  \cite{bai2016econometric} provides a review with further references.  Another literature has proposed alternative estimation methods along with asymptotic analysis in the regime with $T$ fixed and $N$ increasing.  This includes the quasi-difference approach of \citet{HoltzEakin-Newey-Rosen1988} and generalized method of moments approaches of \cite{AhnLeeSchmidt2001,AhnLeeSchmidt2013}.
More recent papers analyzing the fixed $T$ large $N$ regime include \cite{robertson2015iv}, \cite{juodis2018fixed}, \cite{westerlund2019cce}, \cite{higgins2021fixed}, \cite{juodis2022linear}.
 None of these papers provide inference methods that remain valid when factors are weak or rank-deficient (e.g.\ $f=0$).
\cite{Chamberlain2009} derive estimators that satisfy a Bayes-minimax property over a certain class of priors in a finite sample setting that includes a version of the model (\ref{FactorModel}).
This Bayes-minimax property does not, however, translate to a guarantee on coverage or estimation error under weak factors.

A special case of the violation of the  strong factor assumption is when some factor are equal to zero, while all other factors are strong;
the inference results of \cite{bai2009panel} are usually robust towards this specific violation of the strong factor assumption
\citep{MoonWeidner2015}. This robustness, however, does not carry over to more general weak factors in the DGP of $\Gamma_{it}$,
 as illustrated by Figure~\ref{fig: intro}.

The problem of weak factors is related to the problem of omitted variable bias of LASSO estimators in high dimensional regression that is the focus of debiased LASSO estimators \citep[see][]{belloni_inference_2014,javanmard_confidence_2014,van_de_geer_asymptotically_2014,zhang_confidence_2014}.
Just as LASSO estimators omit variables with coefficients that are large enough to cause omitted variables bias but too small to distinguish from zero, weak factors in $\Gamma$ can be difficult to estimate, leading to omitted variables bias in conventional estimates of $\beta$.
Our approach to using minimax linear estimation to debias an initial estimate
mirrors the approach of \citet{javanmard_confidence_2014} to debiasing the LASSO.
We discuss this connection further in Section \ref{sec:comparison_to_other_results}.
\citet{hirshberg_augmented_2020} provide a general discussion and further
references for minimax linear debiasing; we refer to this general approach as augmented
linear estimation following their terminology.  Minimax linear estimation itself
goes back at least to \citet{ibragimov_nonparametric_1985}, with further results
on this approach and its optimality properties in \citet{donoho1994statistical},
\citet{armstrong2018optimal} and \citet{yata_optimal_2021}, among others.  The
particular form of the minimax estimator used for debiasing in our setup follows
from a formula given in \citet{armstrong_bias-aware_2020}.

Requiring $\Gamma_{it}$ to have the factor structure~\eqref{FactorModel} is
equivalent to requiring the matrix of unobserved effects $\Gamma$ to have rank
at most $R$, i.e., having ${\rm rank}(\Gamma) \leq R$. Bounding the nuclear norm of $\tilde \Gamma$ or $\Gamma$ instead can also be seen as a convex relaxation of this requirement. Similar convexifications of the rank constraint have been widely used in the matrix completion literature
(e.g., \citealt{RechtFazelParrilo2010} and \citealt{Hastieetal2015} for recent surveys),
and for reduced rank regression estimation
(e.g., \citealt{RohdeTsybakov2011}).
In the econometrics literature, the numerous applications of this idea include, for example, estimation of pure factor models \citep{BaiNg2017}, estimation of panel regression models with homogeneous \citep{moon2018nuclear,beyhum2019square} and heterogeneous coefficients \citep{chernozhukov2019inference}, estimation of treatment effects \citep{athey_matrix_2021,fernandez2021low}, and many others.\footnote{For example, recent economic applications of nuclear norm and related penalization methods also include latent community detection \citep{alidaee2020recovering,ma2022detecting}, quantile regression \citep{belloni2019high,wang2022low,feng_2023}, and estimation of panel threshold models and high-dimensional VARs (\citealp{miao2020panel} and \citealp{miao2023high}).}
However, none of these papers obtain asymptotically valid CIs or improved rates of convergence under weak factors.

In recent work, \citet{chetverikov2022spectral} propose an estimator that, like ours, achieves a faster rate of convergence than conventional approaches under weak factors.\footnote{The main focus of \citet{chetverikov2022spectral} is the grouped effects model of \citet{bonhomme2015grouped}, which is a special case of the interactive fixed effects setting we consider here.  However, the authors extend their results to the general interactive fixed effects setting.}
While \citet{chetverikov2022spectral} allow for weak factors in some of their estimation results, they assume strong factors when constructing CIs.
The estimation approach in \citet{chetverikov2022spectral} also differs from our approach by using modelling assumptions that place a factor structure on the covariate matrix $X$.

Our focus is on allowing for weak factors without imposing additional assumptions on the error term $U$, such as homoskedasticity or full independence from the individual effects $\Gamma$ and regressor $X$.  Such additional structure allows for further identifying information by making it easier to distinguish between the error term $U$ and the individual effects $\Gamma$, leading to a fundamentally different analysis.  \citet{zhu2019well} derives asymptotic upper and lower bounds for estimators and CIs in a setting with possible weak factors under homoskedastic and fully independent errors.  The estimators and CIs constructed by \citet{zhu2019well} take advantage of the additional structure of \citeauthor{zhu2019well}'s setting, making them inapplicable in ours.  However, the \emph{lower} bounds derived by \citet{zhu2019well} are immediately relevant: they show that no CI can be asymptotically valid under weak factors while mimicking the performance of the CI of \citet{bai2009panel} when factors are strong.

As discussed above, our assumptions rule out the case where $X_{it}$ is an indicator variable for a policy that affects a subset of units starting in the same time period.
Recent papers that analyze such settings include \citet{ferman_synthetic_2021} and \citet{arkhangelsky_synthetic_2021}.
The fact that $X_{it}$ is collinear with the confounding factor model in this setting presents a fundamental identification issue that requires placing additional conditions on the model.
In contrast to this literature, our goal is to leverage variation in $X_{it}$ that cannot be explained by a low dimensional factor model in settings where such variation exists.

\citet{beyhum2022factor}, \citet{fan2022learning}, and \citet{bai2023approximate} consider estimation and inference in various settings under a regime in which a lower bound on the strength of the factors can decrease with $N$ and $T$, but is large enough that factors can be consistently estimated.  This is analogous to the ``semi-strong'' regime in weak instrument and related settings; see \citet{andrews2012estimation}.  While the semi-strong regime requires careful theoretical analysis, the fact that factors can be consistently estimated leads to asymptotically unbiased and normal estimators for the main effect $\beta$.  Our results apply to semi-strong and strong regimes as well, while also allowing for weak factor regimes in which factors cannot be consistently estimated.

Finally, \citet{cox2024weak} develops tools for inference in low-dimensional factor models with weak identification. In \citet{cox2024weak}, the primary objects of interests are the covariance of the factors and the loadings. The baseline model in \citet{cox2024weak} does not include observed covariates, whereas we focus on estimation and inference on $\beta$, the coefficient on $X_{it}$, exclusively.\footnote{\citet{cox2024weak} mentions that observed covariates could, in principle, be incorporated in his framework as long as they are uncorrelated with the unobserved effects, which is a primary worry in the panel literature.}

\bigskip

The rest of this paper is organized as follows. Section \ref{sec:MainIdea} introduces the framework and describes construction of the debiased estimator and bias-aware CI. Section \ref{sec:linear_factor_implementation} provides implementation details. Section \ref{asymptotic_section} provides formal statistical guarantees. Section \ref{sec: numerical} considers numerical and empirical illustrations.
A supplementary appendix contains all proofs 
and additional results for the numerical and empirical illustrations.

\section{Construction of robust estimates and confidence intervals}
\label{sec:MainIdea}

\subsection{Setup}

We consider a panel setting in which we observe a scalar outcome
$Y_{it}$, a scalar covariate $X_{it}$ of interest
and additional control covariates $\{Z_{k,it}\}_{k=1}^K$
for $i=1,\ldots,N$, $t=1,\ldots, T$,
which follow
the regression model (\ref{model0}).
The error term $U_{it}$ is assumed to be mean zero conditional on $X$,
$\{Z_{k,it}\}_{k=1}^K$ and $\Gamma$,\footnote{We note that this requires strict
  exogeneity and in particular rules out
  using lagged outcomes as covariates.  We leave extensions to models with lagged outcomes as a topic for future research.}
but we allow for
heteroskedasticity, which may depend on $X_{it}$ and $\Gamma_{it}$, as well as
some weak dependence.
We write the model in matrix notation as
\begin{align}\label{model0_matrix}
  Y = X \beta + Z\cdot \delta +  \Gamma + U \, ,
  \qquad\qquad
  \e[U|X,Z,\Gamma] = 0 \, ,
\end{align}
where
$Z$ denotes the three dimensional array
$\{Z_{k,it}\}$
and we define $Z\cdot \delta = \sum_{k=1}^K Z_{k}\delta_k$ where $Z_k$ denotes
the matrix with $i,t$-th element $Z_{k,it}$.
We use $\lambda$ to denote the $N\times R$ matrix of loadings $\lambda_{ir}$ and $f$ to denote the $T\times R$ matrix of factors $f_{tr}$, so that (\ref{FactorModel}) can be written in matrix form as $\Gamma=\lambda f'$.

We are interested in the coefficient $\beta$ of $X_{it}$, which can be
interpreted as the effect of a treatment variable $X_{it}$ in a constant
treatment effects model (we discuss extensions to heterogeneous treatment
effects in Remark \ref{het_te_remark}).
For concreteness, we use panel notation, and we refer to $i$ and $t$ as
individuals and time periods respectively.  However, we allow for other settings
such as network data in which $i$ and $t$ both index individuals in a network.
While we will assume a low rank structure on $\Gamma$, we allow for arbitrary
dependence between the covariate $X_{it}$ and the individual effect $\Gamma_{it}$.

A key ingredient in our approach is an initial estimate $\hat\Gamma$ of $\Gamma$ and a bound on its estimator in the nuclear norm, which holds with probability approaching one:
\begin{align}
  \big\| \widetilde \Gamma  \big\|_*
  \leq \hat C \, ,
  \qquad\text{where}\qquad
  \widetilde\Gamma := \Gamma - \hat\Gamma \, .
  \label{NNcondition}
\end{align}
Here, $\|\cdot\|_*$ denotes the nuclear norm of the argument matrix,
and  $\hat C \geq 0$ is a known or estimated constant.
We describe our estimate $\hat\Gamma$ and bound $\hat C$ in Section \ref{sec:linear_factor_implementation}, and we state a formal result giving conditions under which the bound holds with probability approaching one in Section \ref{asymptotic_section}.
This bound depends on an upper bound for the number of factors $R$, which must be specified a priori.  Importantly, our approach does not require specifying the exact number of factors: some of the factors may be zero, in addition to the possibility of being ``weak'' in the sense of being close to zero.
We emphasize that obtaining an computable upper bound $\hat C$ that enables construction of our CI is itself one of the main technical contributions of this paper.\footnote{As we discuss further in Section \ref{asymptotic_section}, a tighter CI can be constructed by bounding the difference between the estimate $\hat\Gamma$ and $\Gamma+P_{\lambda} U$, where $P_\lambda = \lambda (\lambda' \lambda)^+ \lambda$, with $M^+$ denoting the Moore–Penrose inverse of a matrix $M$.
The implementation in Section \ref{sec:linear_factor_implementation} is for the tighter CI that uses these arguments.
For ease of exposition, however, we focus on using the bound \eqref{NNcondition} directly in the remainder of this section.}

\begin{remark}
    Although the main focus of this paper is on models with the linear factor
    structure~\eqref{FactorModel}, the methodology presented in this section applies to general
    interactive fixed effects models as long as it is possible to construct a preliminary estimator $\hat \Gamma$ and a bound $\hat C$ satisfying \eqref{NNcondition}. For example, we conjecture that our method can also be extended to nonlinear factor models with $\Gamma_{it} = g(\lambda_i,f_t)$, where $g(\cdot,\cdot)$ is some unknown function (e.g., \citealp{zeleneev2019identification,freeman2023linear}). As noted, for example, in \citet{fernandez2021low}, such $\Gamma$ can be approximated by a low-rank matrix with a (slowly) growing rank $R$. Hence, we expect that our method can be applied in this setting with $\hat \Gamma$ constructed using a growing $R$ and $\hat C$ adjusted for the low-rank approximation error (if needed), in the same spirit as sieve approximations are used in nonparametric estimation.
\end{remark}

\subsection{Augmented linear estimators and CIs}

We first define a class of estimators and CIs, indexed by an $N\times T$ matrix
$A$.  We then provide a choice of the matrix $A$, based on finite
sample optimality in an idealized setting.
Our class of estimators is given in the following definition.

\begin{definition}
  \label{DefAugLinEst}
  Let $A=A(X,Z)$ be an $N \times T$ matrix of weights $A_{it} \in \mathbb{R}$ that can depend on the matrix $X$
  and array $Z$.  Let $\hat\Gamma$ be an initial estimate of $\Gamma$, and let
  $\widetilde Y=Y-\hat\Gamma$.  The \emph{augmented linear estimator} with
  weight matrix $A$ and initial estimate $\hat\Gamma$ is given by
  \begin{align}
    \hat\beta_A 
    &:= \sum_{i=1}^N \sum_{t=1}^T A_{it}\widetilde Y_{it}
      =\langle A, \widetilde Y \rangle_F.
      \label{DefHatBeta}
  \end{align}

\end{definition}
 Here, $\langle \cdot, \cdot \rangle_F$ denotes the   entry-wise inner product between the argument matrices.

The estimator $\hat\beta_A=\langle
A,\widetilde Y \rangle_F$ applies a linear estimator after an initial estimation
step in which the initial estimate $\hat\Gamma$ is subtracted from the outcome
$Y$.
This mirrors applications of this idea in other settings going back to
\citet{javanmard_confidence_2014}; see \citet{hirshberg_augmented_2020} for
references (the term ``augmented linear estimation'' is used in the latter
paper).

To analyze this class of estimators, note that subtracting the initial estimate
from both sides of the equation (\ref{model0_matrix}) gives
\begin{align}
  \widetilde Y = X \beta + Z\cdot \delta + \widetilde \Gamma + U 
  \label{model}
\end{align}
(recall that $\widetilde Y=Y-\hat\Gamma$ and $\widetilde
\Gamma=\Gamma-\hat\Gamma$).
Our choice of the matrix $A$ will be motivated by a heuristic in which we consider the model (\ref{model}) with $\widetilde Y$ as the observed outcome and $\widetilde \Gamma$ a nuisance parameter such that $U$ is mean zero conditional on $X,Z$ and $\tilde \Gamma$, with the bound (\ref{NNcondition}) interpreted as a deterministic bound that holds with $\hat C$ nonrandom.
This heuristic is not literally true, since $\widetilde\Gamma$ depends on $U$ through the estimation error in the initial estimate $\hat\Gamma$.
Nonetheless, the CIs and estimators we obtain will be asymptotically valid and consistent respectively, under conditions that we give in Section \ref{asymptotic_section}.

Following this heuristic, we consider the decomposition
\begin{align}\label{bias_variance_decomposition}
  &\hat\beta_A-\beta
    =
    \bias_{\beta,\delta,\widetilde\Gamma}(\hat\beta_A)
    + \langle A, U \rangle_F
\end{align}
where
\begin{align}\label{bias_def_eq}
  \bias_{\beta,\delta,\widetilde\Gamma}(\hat\beta_A):=\left( \langle A,X
  \rangle_F - 1 \right)\beta + \langle A,Z\cdot \delta \rangle_F + \langle A,
  \widetilde\Gamma \rangle_F.
\end{align}
Under the heuristic where $\tilde\Gamma$ is nuisance parameter in the model (\ref{model}), $\bias_{\beta,\delta,\widetilde \Gamma}$ gives the bias of the estimator $\hat\beta_{A}$ conditional on $X,Z$ and $\tilde \Gamma$.
In reality, $\bias_{\beta,\delta,\widetilde \Gamma}$ does not literally
give the bias or conditional bias of $\hat\beta_A$, since conditioning on
$\widetilde\Gamma=\Gamma-\hat\Gamma$ means conditioning on an information set
that depends on $Y$ through the preliminary estimate $\hat\Gamma$.
We nonetheless refer to $\bias_{\beta,\delta,\Gamma}(\hat\beta_A)$ as a bias term, following our heuristic.

Let $\widehat{\se}$ be an estimate
of the standard deviation of $\langle A,U \rangle_F=\sum_{i=1}^N\sum_{t=1}^TA_{it}U_{it}$.
For example, to allow for arbitrary
heteroskedasticity in $U_{it}$ while imposing independence across $i$ and $t$,
we can use $\widehat{\se}=\sqrt{\sum_{i=1}^N\sum_{t=1}^T A_{it}^2\hat U_{it}^2}$
where $\hat U_{it}$ denotes residuals from an initial regression.
If $\bias_{\beta,\delta,\Gamma}(\hat\beta_A)$ were zero,
then we could form a CI by adding and subtracting a normal critical value times
$\widehat{\se}$.
To take into account the possibility that
$\bias_{\beta,\delta,\Gamma}(\hat\beta_A)$ will in general be nonnegligible in
our setting, we use the bound (\ref{NNcondition}) to obtain an upper bound on
the bias term.
In particular,
when (\ref{NNcondition}) holds, we have
$\left| \bias_{\beta,\delta,\widetilde\Gamma}(\hat\beta_A) \right|
\le \maxbias_{\hat C}(\hat\beta_A)$, where
for general $C \geq 0$ we define
\begin{align}
  \maxbias_C(\hat\beta_A)
  &:= \sup_{ \beta,\delta,\widetilde\Gamma: \|\widetilde\Gamma\|_*\le C} \bias_{\beta,\delta,\widetilde\Gamma}(\hat\beta_A)
 \nonumber \\
 &   =
   \begin{cases}
     \displaystyle \sup_{ \widetilde\Gamma: \|\widetilde\Gamma\|_*\le C}
  \langle A, \widetilde\Gamma \rangle_F & \text{if} \, \langle A, X \rangle_F=1,\, \text{and } \langle A, Z_k \rangle_F=0,\,  \text{for } k=1,\ldots K,  \\
     \infty & \text{otherwise}
   \end{cases}  
 \nonumber \\
 &   =
   \begin{cases}
     C s_1(A) & \text{if} \, \langle A, X \rangle_F=1,\, \text{and } \langle A, Z_k \rangle_F=0,\,  \text{for } k=1,\ldots K,  \\
     \infty & \text{otherwise.}
   \end{cases}  
   \label{WorstCaseBias} 
\end{align}
Here, for the second equality we used that the supremum over $\beta$ and $\delta$ is unbounded unless 
$\langle A, X \rangle_F=1$ and $\langle A, Z_k \rangle_F=0$, and for the final step we used that 
the nuclear norm $ \| \cdot \|_*$ is dual to the spectral norm,  which we  denote by $s_1(\cdot)$ since it is equal to the
largest singular value of the argument matrix.
We refer to $\maxbias_{\hat C}(\hat\beta_A)$ as the worst-case bias of the
estimator $\hat\beta_A$ (again, this terminology reflects the heuristic in which $\tilde\Gamma$ is treated as a nuisance parameter in (\ref{model}) rather than estimation error from the initial estimate $\hat\Gamma$).

Note that, whereas
$\bias_{\beta,\delta,\widetilde\Gamma}(\hat\beta_A)$ depends on the unknown
matrix of individual effects $\Gamma$ through the matrix
$\widetilde\Gamma=\Gamma-\hat\Gamma$,
$\maxbias_{\hat C}(\hat\beta_A)$ is feasible to compute once a bound $\hat C$ is given.
Taking into account the possible bias leads to a \emph{bias-aware} CI:
\begin{align}\label{general_A_CI_eq}
  \left\{ \hat\beta_A \pm \left[ \maxbias_{\hat C}(\hat\beta_A) + z_{1-\alpha/2} \widehat{\se} \right]   \right\}.
\end{align}
To motivate this CI, note that the probability that the lower endpoint is greater
than $\beta$ is
\begin{align*}
  &P\left( \hat\beta_A - \maxbias_{\hat C}(\hat\beta_A) - z_{1-\alpha/2} \widehat{\se} > \beta
    \right)
    = P\left( \sum_{i=1}^N\sum_{t=1}^TA_{it}U_{it} + \bias_{\beta,\delta,\widetilde\Gamma}(\hat\beta_A) > \maxbias_{\hat C}(\hat\beta_A) + z_{1-\alpha/2} \widehat{\se}
    \right)  \\
  &\le P\left( \sum_{i=1}^N\sum_{t=1}^TA_{it}U_{it} > z_{1-\alpha/2} \widehat{\se}
    \right)
    \approx \alpha/2,
\end{align*}
where the last step assumes that $\sum_{i=1}^N\sum_{t=1}^TA_{it}U_{it}$ is approximately normally distributed with 
zero mean and standard deviation close to $ \widehat{\se}$. We provide formal justifications for this later. By a similar argument, the probability that the upper endpoint is less than $\beta$ can be bounded by $\alpha/2$, and these calculations together imply that the coverage of our CI is approximately at least $1-\alpha$. %

\begin{remark}\label{het_te_remark}
  In principle, our approach can be extended to a heterogeneous treatment effect model where the constant coefficient $\beta$ is replaced by an individual specific coefficient $\beta_{it}$ that is allowed to vary with $i$ and $t$.  In particular, if a bound on the nuclear norm of the matrix of coefficients $\beta_{it}$ or on the error of preliminary estimates of these coefficients is available in addition to such a bound for $\Gamma$, we can use minimax linear debiasing to estimate a linear functional of the individual specific effects $\beta_{it}$.  For example, the linear functional $\frac{1}{NT}\sum_{i=1}^N\sum_{t=1}^T \beta_{it}$ gives the average treatment effect of a one-unit change in $X_{it}$ over the $NT$ units in a setting where $\beta_{it}$ is interpreted as the causal effect of a change in the variable $X_{it}$.
  Deriving a computable bound on the nuclear norm error of an initial estimate of the coefficients $\beta_{it}$ in this case is nontrivial, however, and we leave this question for future research.
\end{remark}

\subsection{Choice of weights $A=(A_{it})$}\label{choice_of_weight_section}

As described in the last subsection,
one can construct valid confidence intervals for $\beta$
of the form \eqref{general_A_CI_eq} for any choice of weight matrix $A$, subject to weak regularity conditions. 
To get a simple baseline procedure, we compute weights that are optimal
in an idealized setting where $U_{it}\stackrel{iid}{\sim} N(0,\sigma^2)$ independently of $X,Z$ and $\tilde \Gamma$ (again, this involves invoking the heuristic of treating $\tilde\Gamma$ as a nuisance parameter in (\ref{model}) rather than estimation error from a preliminary estimate).
In this idealized setting, $\hat\beta_A$ is then normally distributed with
variance $\sigma^2\sum_{i=1}^N\sum_{t=1}^TA_{it}^2=\sigma^2\|A\|_F^2$ (where
$\|\cdot\|_F$ denotes the Frobenius norm), and with bias ranging from
$-\maxbias_{\hat C}(\hat\beta_A)$ to $\maxbias_{\hat C}(\hat\beta_A)$.
Thus, if we choose worst-case MSE under i.i.d. normal errors as our criterion function for the weights, then the optimal weights
are obtained by minimizing  $\left( \maxbias_{\hat C}(\hat\beta_A) \right)^2 + \sigma^2
\|A\|_F^2 $.
By substituting the formula for $\maxbias_{\hat C}(\hat\beta_A)$ from
(\ref{WorstCaseBias}), we obtain the following baseline choice of weights, indexed by
a tuning parameter $\biasweight$ that corresponds to~$\hat C/\sigma$.

\begin{definition}
\label{DefOptimalWeights}
For $\biasweight>0$, define the ``optimal'' $N\times T$ weight matrix by
\begin{align*}
  A^*_{\biasweight} &:= \argmin_{A \in \mathbb{R}^{N \times T} } \, \biasweight^2 s_1(A)^2 + \|A\|_F^2
      \;\; \qquad\text{s.t.}\quad  \text{$\langle A,X \rangle_F = 1$ and $\langle A,Z_k\cdot \delta \rangle_F = 0$,}
\end{align*}
Here, the constraint $\langle A,Z_k\cdot \delta \rangle_F = 0$
is imposed for all $k \in \{1,\ldots,K\}$.
\end{definition}

Heuristically, we expect that a
good choice of $\biasweight$ will correspond to $\hat C/\sigma$ such that the
bound $\hat C$ on the
nuclear norm holds with high probability.
Conveniently, our nuclear norm bound in the exact factor model in Section
\ref{sec:linear_factor_implementation} scales with the standard deviation $\sigma$ in the
homoskedastic case, which gives us a simple and feasible choice of the tuning
parameter $\biasweight$.

We emphasize again that while the definition of $A^*_\biasweight$ is
motivated by the idealized setting $U_{it}\stackrel{iid}{\sim} N(0,\sigma^2)$, we do {\it not} assume that the error terms $U_{it}$ satisfy this strong
assumption.
Choosing $A=A^*_\biasweight$
to construct the estimator  $\hat\beta_A $ and the confidence intervals \eqref{general_A_CI_eq} under more general error distributions just means that
the resulting estimates and confidence intervals will not be optimal (in finite samples), but we will nevertheless show them to be consistent and valid, respectively.

\begin{remark}\label{other_criteria_remark}
  While we have used MSE to motivate our baseline choice of weights $A^*_\biasweight$, one
  could use other criteria corresponding to different weights on bias and
  variance.  For example, optimizing CI length when $\hat C/\sigma=b$
  would give the criterion $\biasweight s_1(A)+ z_{1-\alpha}\|A\|_F$.  If $\beta$ gives
  the net welfare gain of an all-or-nothing policy change, then one can target
  minimax welfare regret as in \citet{ishihara_evidence_2021} and
  \citet{yata_optimal_2021}.
  In our Monte Carlo simulations however, we find that the exact choice of
  criterion has little effect on performance.
\end{remark}

\subsection{Practical implementation}\label{general_practical_implementation_section}

The definition of $A^*_\biasweight$ is a convex optimization problem
that can easily be solved numerically for any given input $X$, $Z$,
$\biasweight$.
Using results from
\citet{armstrong_bias-aware_2020}, it follows that $A^*_\biasweight$
can also be computed using the residuals of a nuclear norm regularized
regression of $X$ on $Z_1,\ldots,Z_K$ and a matrix of individual effects.
When there are no additional covariates~$Z$, this nuclear
norm regularized regression simplifies further: it can be solved by computing
the singular value decomposition of $X$, and then performing soft thresholding
on the singular values.  The resulting weights $A^*_\biasweight$ obtained from
the residuals of this regression replace the largest singular values of $X$
with a constant.
We provide details in Appendix~\ref{computational_details_appendix}.

In addition to giving alternative methods for computing the weights $A^*_\biasweight$,
these results provide some intuition for these weights.
The
residuals from this nuclear norm regularized regression of $X$ on $Z_1,\ldots,
Z_K$ and the individual effects ``partial out'' potential correlation of $X$ with the
estimation error $\widetilde\Gamma$, similar to the estimator of
\citet{robinson_root-n-consistent_1988} in the partially linear model.
When there are no additional covariates $Z$, this amounts to removing the largest singular values of $X$ and replacing them with a constant.

To summarize, we can compute an estimator $\hat\beta_{A}$ using Definition
\ref{DefAugLinEst} using any matrix of weights $A$.  We can also compute a CI $\left\{ \hat\beta_A \pm \left[ \maxbias_{\hat C}(\hat\beta_A) + z_{1-\alpha/2}
    \widehat{\se} \right]   \right\}$
as in (\ref{general_A_CI_eq}),
once we have a standard error $\widehat{\se}$ and an upper bound $\hat C$ for the
nuclear norm of the error in the initial estimate of $\Gamma$.  Definition \ref{DefOptimalWeights}
gives us a heuristic for computing a reasonable choice of the matrix $A$, once
we have an initial choice of $\biasweight$ for the ratio $\hat C/\sigma$ of the nuclear norm
bound to variance of $U_{it}$.

Thus, to apply our approach, we need an initial choice
$\biasweight$ to compute the weights $A^*_\biasweight$
using Definition \ref{DefOptimalWeights}.
We also need a robust upper bound $\hat C$ such that the bound
(\ref{NNcondition}) holds with high probability.
Finally, we need a robust standard error $\widehat{\se}$.
Our CI then takes the form in (\ref{general_A_CI_eq}) with $A=A^*_\biasweight$
and the given bound $\hat C$ and standard error $\widehat{\se}$.
In Section \ref{sec:linear_factor_implementation}, we give details of these choices, as well as how to compute
the initial estimate of $\Gamma$.

\section{Implementation}\label{sec:linear_factor_implementation}

In this section, we describe the implementation of our approach.
Our approach relies on bounds for the nuclear norm of the initial estimate of $\Gamma$, derived formally in Section \ref{asymptotic_section}.  
As explained in Section \ref{asymptotic_section}, a tighter CI can be derived using a more nuanced argument that bounds the difference between $\hat\Gamma$ and $\Gamma+P_{\lambda} U$, where $P_\lambda = \lambda (\lambda' \lambda)^+ \lambda$, and $M^+$ denotes the Moore–Penrose inverse of a matrix $M$.
In particular, we show that
the bound
$\hat C\approx 2Rs_1(U)$
can be used.
The bound $R$ on the number of factors must be specified by the researcher, similar to other methods in this literature \citep[e.g.][]{bai2009panel}.
Furthermore, the weights $A^*_\biasweight$ are designed to be optimal when
$U_{it}\stackrel{iid}{\sim} N(0,\sigma^2)$, which leads to the approximation
$s_1(U)/\sigma\approx \sqrt{N}+\sqrt{T}$ \citep{geman1980limit}.  We therefore use
$\biasweight=\bTW:=2R(\sqrt{N}+\sqrt{T})$ as our default choice to calibrate
$\hat C/\sigma$ when computing the weights in Definition \ref{DefOptimalWeights}.
We then use an upper bound $\hat C$ that is valid under heteroskedasticity when
computing $\maxbias_{\hat C}(\hat\beta_{A^*_\bTW})$ in the construction of the CI.

Our initial estimate is formed in two steps.  First, we form the least squares estimate $\hat\Gamma_{\rm LS}$.  We then apply our debiasing approach to get an estimate of the coefficients $\beta$ and $\delta$ and form an $\hat\Gamma_{\rm pre}$ by applying least squares to estimate $\Gamma$, with $\beta$ and $\delta$ fixed at this initial debiased estimate.  This estimate $\hat\Gamma_{\rm pre}$ is then used as the initial estimate in our procedure.
Thus, our procedure involves applying our debiasing approach twice.  This appears to be necessary to get an initial estimate $\hat\Gamma_{\rm pre}$ the best possible nuclear norm bounds on the estimation error.

Below we provide the %
details of our implementation algorithm.\footnote{Implementation of this algorithm in R is also available at \url{https://github.com/chenweihsiang/PanelIFE/tree/main}. We thank Chen-Wei Hsiang for his excellent assistance in preparing this R package.}\textbf{}

\begin{algorithm}[Implementation for the factor model]\mbox{}
  \label{alg: factor model}
  \begin{description}
  \item[Input] Data $Y,X,Z$ and $R$ pre-specified by the user, along with tuning
    parameter $\varepsilon$.

  \item[Output] Estimator and CI for $\beta$.

  \end{description}

  \begin{enumerate}
  \item Compute the least squares (LS) estimator
    \begin{align}
      \left(\hat \beta_{\rm LS}, \hat \delta_{\rm LS}, \hat \Gamma_{\rm LS}\right) = \argmin_{\left\{\beta \in \mathbb R, \delta \in \mathbb R^K, G \in \mathbb{R}^{N \times T} \, : \, {\rm rank}(G) \leq R \right\}}
      \sum_{i=1}^N \sum_{t=1}^T \left( Y_{it} - X_{it} \beta - Z_{it}' \delta - G_{it} \right)^2 .
      \label{eq: LS def}
    \end{align}

  \item Compute $\widetilde Y_{\rm pre} = Y - \hat \Gamma_{\rm LS}$
    and let $\bTW=2R(\sqrt{N}+\sqrt{T})$.
    Let
    \begin{align*}
      \hat\beta_{\rm pre} = \langle A^*_{\bTW}, \widetilde Y_{\rm pre} \rangle_F.
    \end{align*}
    Construct $\hat\delta_{\rm pre}$ with the $j$-th element $\hat\delta _{{\rm
  pre},j}$ computed in the same way as $\hat\beta_{\rm pre}$, but with $X$ and $Z_j$
switched.

  \item Compute $\hat \Gamma_{\rm pre}$ as
  \begin{align*}
    \hat \Gamma_{\rm pre} = \argmin_{\left\{ G \in \mathbb{R}^{N \times T} \, : \, {\rm rank}(G) \leq R \right\}}
    \sum_{i=1}^N \sum_{t=1}^T \left( Y_{it} - X_{it} \hat \beta_{\rm pre} - Z_{it}' \hat \delta_{\rm pre} - G_{it} \right)^2 .
  \end{align*}
  The solution $\hat \Gamma_{\rm pre}$ to this least squares problem is simply
  given by the leading $R$ principal components of the residuals $Y_{it} -
  X_{it} \hat \beta_{\rm pre} - Z_{it}' \hat \delta_{\rm pre}$.
  Compute $\widetilde Y = Y - \hat \Gamma_{\rm pre}$.

  \item Compute the final estimate
    \begin{align*}
      \hat\beta = \hat\beta_{A^*_{\bTW}} = \langle A^*_{\bTW}, \widetilde Y \rangle_F.
    \end{align*}
    To compute the CI, let $\hat C = (2 +
    \varepsilon) R s_1(\hat U_{\rm pre})$ and $\widehat{\se}^2=\sum_{i=1}^N\sum_{t=1}^T A_{b^*,it}^{*2}\hat U_{{\rm pre},it}^2$, where
    \begin{align*}
      \hat U_{\rm pre} = Y - X \hat \beta_{\rm pre} -Z \cdot \hat \delta_{\rm pre} - \hat \Gamma_{\rm pre}.
    \end{align*}
    Compute the CI
    \begin{align}
      \label{eq: bias aware CI}
      \hat\beta_{A^*_{\bTW}} \pm \left[ \maxbias_{\hat C}(\hat\beta_{A^*_{\bTW}}) + z_{1-\alpha/2}\widehat{\se}\right]
    \end{align}
    where $\maxbias_{\hat C}(\hat\beta_{A^*_{\bTW}})=\hat C s_1(A^*_{\bTW})$.
  
  \end{enumerate}
\end{algorithm}

\begin{remark}[Behavior of estimator under strong factors]
  \label{remark: strong factors CS}
  In Section~\ref{ssec: semi strong}, we show that, in the absence of weak
  factors, $\maxbias_{\hat C}(\hat\beta_{A^*_{\bTW}})$ becomes negligible
  relative to $\widehat{\se}$ in the construction of the CI in \eqref{eq: bias
    aware CI}.  Thus, the CI
  \begin{align}\label{eq:non_bias_aware_ci}
    \hat\beta_{A^*_{\bTW}}
    \pm z_{1-\alpha/2}\widehat{\se},
  \end{align}
  which uses our bias-corrected estimator but
  ignores bias when computing the critical value, will have correct asymptotic
  coverage in the strong factor case.
  In our Monte Carlos provided in Section~\ref{ssec: MC non-robust coverage},
  we find that this CI is (i) comparable to alternative non-robust CIs in terms
  of the length and (ii) despite its non-robustness, substantially less
  size-distorted if there is a weak factor(s) since it is based on the debiased
  estimator.
  In settings where the bias-aware CI (\ref{eq: bias aware CI}) is too wide to
  yield precise inference, we recommend reporting the CI
  (\ref{eq:non_bias_aware_ci}) alongside the bias-aware CI (\ref{eq: bias aware
    CI}) as a compromise between ignoring weak factors and a fully robust
  approach.

\end{remark}

\begin{remark}[Choice of $R$ and $\varepsilon$]
  The  quantity $\varepsilon$ is used in the bound $\hat C = (2 +
  \varepsilon) R s_1(\hat U_{\rm pre})$ on $\|\widetilde \Gamma\|_*$ needed
  to compute the CI in the final step. While $\varepsilon>0$ is necessary for theoretical guarantees, 
   in our Monte Carlos, we find that  we get good coverage when choosing
   $\varepsilon=0$.

  In contrast, the choice of $R$ has a substantive effect on the CI, both
  through the bound $\hat C = (2 + \varepsilon) R s_1(\hat U_{\rm pre})$ and
  through the point estimate.
  Since the number of weak factors cannot be determined from the data, the
  researcher must specify an a priori bound on the total number of factors $R$.
  Nonetheless, the data can be informative about the number of strong or semi-strong
  factors, which provides a lower bound for $R$.  We recommend forming an
  estimate $\hat R_s$ of the number of strong factors using one of the standard
  methods (e.g., \citealp{BaiNg2002,Onatski2010,AhnHorenstein2013}) and using
  this as a starting point for examining the sensitivity of the results to the
  choice of $R$.  For example, by taking $R = \hat R_s + 1$, the researcher
  allows for the potential presence of an additional weak factor (or $R - \hat
  R_s$ weak factors in general for a bigger $R$).
\end{remark}

\begin{remark}[Lindeberg condition]\label{lindeberg_remark}
  The asymptotic validity of the CI depends on asymptotic normality of the
  stochastic term $\langle A, U \rangle_F$
  where $A=A^*_{\bTW}$ is a non-random matrix of weights.  This, in
  turn, depends on a Lindeberg condition on the weights $A$.
  To ensure that this holds, we can modify our optimization procedure for
  computing the weights $A=A^*_{\bTW}$ by imposing a bound on the
  Lindeberg weights
  \begin{align}\label{lindeberg_weights_def_eq}
    \operatorname{Lind}(A)=\frac{\max_{1\le i\le N, 1\le t\le T} A_{it}^2}{\sum_{i=1}^N\sum_{t=1}^T A_{it}^2}.
  \end{align}
  A similar approach to showing asymptotic validity is taken in
  \citet{javanmard_confidence_2014} in a different setting.

  To make this approach practical, we need guidance on what makes
  $\operatorname{Lind}(A)$ ``small enough to use the central limit theorem'' in
  a given sample size.  A formal answer to this question is elusive, due to the
  difficulty of obtaining finite sample bounds on approximation error in the
  central limit theorem that are practically useful.  As a heuristic, we can use
  comparisons to other settings where the central limit theorem is used.  For
  example, the sample mean $\bar W=\frac{1}{n}\sum_{i=1}^n W_i$
  with $n$ observations corresponds to an estimator with Lindeberg constant
  $(1/n)^2/[n\cdot (1/n)^2]=1/n$.  If we are comfortable using the normal
  approximation in such a setting with, say, $n=50$, then we can impose a bound
  $\operatorname{Lind}(A)\le 1/50$.
  \citet{noack_bias-aware_2024} provide some discussion of these issues in a
  related setting involving inference in fuzzy regression discontinuity.

  In our Monte Carlos, we find that $\operatorname{Lind}(A)$ is very small for
  the weights used in Algorithm \ref{alg: factor model} once $N$ and $T$ are
  larger than, say, 20.  Thus, imposing a bound on these weights does not appear
  to be necessary in practice in the data generating processes we have examined.
  
\end{remark}

\begin{remark}[Standard error]
  The standard error $\widehat{\se}^2=\sum_{i=1}^N\sum_{t=1}^T A_{it}^2\hat U_{{\rm pre},it}^2$ assumes that $U_{it}$ is
  uncorrelated across $i$ and $t$, but allows for heteroskedasticity.  Such an
  assumption will be reasonable if $\Gamma_{it}$ captures all of the dependence
  in errors for the outcome.  However, incorporating all dependence in
  $\Gamma_{it}$ may lead to an unnecessarily conservative choice of the upper bound $R$ on the number of factors, leading to a wider CI.  To avoid such
  conservative bounds on $\Gamma$, one can incorporate any dependence that is
  not directly correlated with $X_{it}$ into the error term $U_{it}$, and allow
  for such dependence when constructing the standard error. For example, to allow for (arbitrary) time dependence of $U_{it}$ (while maintaining uncorrelatedness across $i$), one could simply use clustered standard errors (see, e.g., \citealp{arellano1987computing,hansen2007asymptotic})
  \begin{align*}
    \widehat{\se}^2=\sum_{i=1}^N \left(\sum_{t=1}^T A_{it}\hat U_{{\rm pre},it}\right)^2.
  \end{align*}
\end{remark}

\section{Asymptotic results}\label{asymptotic_section}

This section gives formal asymptotic results for the estimators and CIs given in
Sections \ref{sec:MainIdea} and \ref{sec:linear_factor_implementation}.
We consider the following decomposition of our regression model:
\begin{align}
  \label{eq: augmented model}
  \widetilde Y: = Y - \hat \Gamma_{\rm pre} = X \beta + Z\cdot \delta + \Gamma-\hat\Gamma_{\rm pre} + U = X \beta + Z\cdot \delta + \widetilde \Gamma + \widetilde U.
\end{align}
Here, $\widetilde\Gamma$ and $\widetilde U$ an be any $N\times T$ matrices
chosen compatibly so that $\widetilde \Gamma + \widetilde U =
\Gamma-\hat\Gamma_{\rm pre} + U$.
While our discussion so far has focused on the case where $\widetilde \Gamma =
\Gamma - \hat \Gamma_{\rm pre}$ and $\widetilde U = U$, it turns out that
allowing for other choices of $\widetilde\Gamma$ and $\widetilde U$ allows for
an improvement in the width of our CI.

To formally state asymptotic results that allow for weak factors and an unknown
error distribution, we introduce some additional notation.
We consider uniform-in-the-underlying distribution asymptotics over a set
$\mathcal{P}$ of distributions $P$ for $\Gamma$ and $X,Z_1,\ldots,Z_K,U$ and a
set $\Theta$ of parameters
$\theta=(\beta,\delta')'$.
While we treat $\Gamma,X,Z_1,\ldots,Z_k$ as random variables determined by the
unknown probability distribution $P$ for notational purposes, we note that a
fixed design setting in which $\Gamma,X,Z_1,\ldots,Z_k$ are non-random
(sequences of) matrices can be incorporated by considering a set $\mathcal{P}$
that places a probability one mass on a given value of
$\Gamma,X,Z_1,\ldots,Z_k$.
We
use $\pr_{P,\theta}$ to denote probability under the given distribution $P$ and
parameters $\theta$.  Formally, we consider large $N$, large $T$ asymptotics in
which $N=N_n\to\infty$ and $T=T_n\to\infty$, and we consider sequences of
distributions $\mathcal{P}=\mathcal{P}_n$ and parameter spaces
$\Theta=\Theta_n$.
Asymptotic statements are then taken in the
sequence $n$.  However, we suppress the dependence on an index sequence $n$ in
order to save on notation.
For a sequence of vectors or matrices $A_{N,T}=A_{N,T}(\theta,P)$ of fixed
dimension (which may depend on $\theta,P$),
we use the notation $A_{N,T}=\mathcal{O}_{\Theta,\mathcal{P}}(r_{N,T})$ when, for every
$\varepsilon>0$, there exists $C_\varepsilon$ such that
\begin{align*}
  \limsup \sup_{P\in\mathcal{P},\theta\in\Theta} \pr_{P,\theta}\left( r_{N,T}^{-1} \| A_{N,T} \|\ge C_\varepsilon \right)\le \varepsilon,
\end{align*}
and we use the notation $A_{N,T}=o_{\Theta,\mathcal{P}}(r_{N,T})$ when, for
every $\varepsilon>0$, we have
\begin{align*}
  \limsup \sup_{P\in\mathcal{P},\theta\in\Theta} \pr_{P,\theta}\left( r_{N,T}^{-1} \| A_{N,T} \|\ge \varepsilon \right) \to 0.
\end{align*}
We use the notation $A_{N,T}\asymp_{\Theta,\mathcal{P}}r_{N,T}$ when
$A_{N,T}=\mathcal{O}_{\Theta,\mathcal{P}}(r_{N,T})$ and
$A_{N,T}^{-1}=\mathcal{O}_{\Theta,\mathcal{P}}(r_{N,T}^{-1})$.
We use the notation
$A_{N,T} \underset{\Theta,\mathcal{P}}{\overset{d}{\to}} \mathcal{L}$
to denote the statement
\begin{align*}
  \limsup 
  \sup_{\theta\in\Theta,P\in\mathcal{P}}
  \Big|  \pr_{\theta,P}\left( A_{N,T} \le t \right) - F_{\mathcal{L}}(t) \Big| \to 1
  \text{ for all }t
\end{align*}
where $F_{\mathcal{L}}$ denotes the cdf of the probability law $\mathcal{L}$.

\subsection{General results}

We first show asymptotic validity of
the CI (\ref{general_A_CI_eq}) under the following high level assumption imposed on the augmented model \eqref{eq: augmented model} and the weights $A_{it}$.

\begin{assumption}\label{ass: augmented high level}
  \mbox{}
  \begin{enumerate}[(i)]
  \item\label{item: augmented high level C hat bound} $\inf_{\theta\in\Theta,P\in\mathcal{P}}\pr_{\theta,P}\left( \|\widetilde \Gamma\|_* \le \hat C \right)\to 1$;
    
  \item\label{item: augmented high level CLT} $\frac{\langle A, \widetilde U
      \rangle_F}{\widehat{\se}}
    \underset{\Theta,\mathcal{P}}{\overset{d}{\to}}
    N(0,1)$.
  \end{enumerate}
\end{assumption}

\begin{theorem}\label{new_high_level_validity_thm}
  Suppose that Assumption \ref{ass: augmented high level} holds.
  Then
  \begin{align*}
    \liminf \inf_{\theta\in\Theta,P\in\mathcal{P}} \pr_{\theta,P}\left( \beta \in \left\{ \hat\beta_A \pm \left[  \maxbias_{\hat C}(\hat\beta_A) + z_{1-\alpha/2}\widehat{\se}  \right] \right\} \right) \ge 1-\alpha.
  \end{align*}
\end{theorem}

\subsection{Primitive conditions}

We now apply these results to the initial estimate and bound given in Section
\ref{sec:linear_factor_implementation}, under the assumption of a linear factor
model for $\Gamma$.
We allow for a side condition on the
Lindeberg weights
$\operatorname{Lind}(A)$
defined in (\ref{lindeberg_weights_def_eq}), as described in Remark
\ref{lindeberg_remark}.
Let $A^*_{\biasweight, c}$ be defined in the same way as $A^*_{\biasweight}$, with the modification that
we impose the constraint $\operatorname{Lind}(A)\le c$:
\begin{align}\label{weight_optimization_lindeberg}
  &\min_{A} \left\| A \right\|_F^2  +  \biasweight^2 s_1(A)^2,  \nonumber  \\
  &\text{s.t.}
    \quad \operatorname{Lind}(A)\le c,
    \quad \langle A,X \rangle_F = 1,
    \quad \langle A,Z_k \rangle_F = 0 \text{ for } k=1,\ldots,K.
\end{align}
In particular, the
weights used in Algorithm \ref{alg: factor model} are given by
$A^*_{\bTW,\infty}=A^*_{\bTW}$, and the weights
$A^*_{\bTW,c}$ with $c<\infty$ correspond to the modification described in
Remark \ref{lindeberg_remark}.
In our asymptotic theory we will require
$c=c_{NT}$ to converge to zero as $N,T \rightarrow \infty$.

We impose the following conditions.

\begin{assumption}[Factor Model]
  \label{ass: factor model}
  Suppose that $\rank{\Gamma} \le R$, i.e., $\Gamma = \lambda f'$ for some $N \times R$ matrix $\lambda$ and some $T \times R$ matrix $f$, 
  with probability one for all $P\in\mathcal{P}$
  and the following conditions hold:
  \leavevmode
  \begin{enumerate}[(i)]
    \item \label{item: NC} 
     Write $W$ for $X,Z_1,\ldots, Z_K$ and $W\cdot \gamma=X\beta +
     \sum_{k=1}^KZ_k\delta_k$ where $\gamma=(\beta,\delta')'$.  We assume that
     there exists $\underline s^2 > 0$ such that
    \begin{align*}
      \min_{\gamma \in \mathbb R^{K+1}: \norm{\gamma} = 1}\frac{1}{NT} \sum_{r = 2 R + 1}^{\min\{N,T\}} s_r^2 (W \cdot \gamma) \ge \underline s^2
    \end{align*}
    with probability approaching 1 uniformly over $P \in \mathcal P$;

    \item \label{item: SN} $s_1(X) = \mathcal{O}_{\Theta,\mathcal{P}} \left(\sqrt{NT}\right)$, $s_1(Z_k) = \mathcal{O}_{\Theta,\mathcal{P}} \left(\sqrt{NT}\right)$ for $k \in \{1, \ldots, K\}$, and $s_1(U) \asymp_{\Theta,\mathcal{P}} \maxSNT $;

          \item \label{item: EX} $\langle X, U \rangle_F  =
            \mathcal{O}_{\Theta,\mathcal{P}} (\sqrt{NT})$ and $\langle Z_k, U
              \rangle_F = \mathcal{O}_{\Theta,\mathcal{P}} (\sqrt{NT})$ for $k \in \{1, \ldots, K\}$;

    \item \label{item: high level singular values of U} $\left(s_1(U) - s_{r}(U)\right)/s_1(U) = o_{\Theta,\mathcal{P}}(1)$ for any fixed positive integer $r$;

    \item \label{A_U_inner_product_bound} For any sequence of matrices $A=A_{N,T}(X,Z)$ that is a function of
    $X,Z_1,\ldots,Z_k$, we have $\langle A, U \rangle_F =
    \mathcal{O}_{\Theta,\mathcal{P}}(\|A\|_F)$.

  \end{enumerate}
\end{assumption}

Assumption \ref{ass: factor model}\eqref{item: NC}
is a generalized non-collinearity condition, which
requires that there is enough variation in the regressors after concentrating out $2R$ arbitrary factors. 
It is closely related to Assumption~A of \cite{bai2009panel}, but our version here avoids mentioning the unobserved factor loadings. The same generalized non-collinearity assumption is imposed in \cite{MoonWeidner2015}. 
The assumption would be violated if
some linear combination $W \cdot \gamma$
of the covariates were to have rank smaller or equal 
to $2R$. In particular,
 ``low-rank regressors'' are ruled out by this condition.  Intuitively, Assumption \ref{ass: factor model}(\ref{item: NC}) holds provided that $X_{it}$ and $Z_{it}$ have non-collinear idiosyncratic components. This intuition is formalized in Appendix~\ref{appendix:VerifyAss2i}, where we also show that, when $X$ is the only regressor in the model (i.e.,\ $K=0$), then Assumption~\ref{ass: X decomposition}\eqref{V_assump} and \eqref{H_assump} below already guarantee that  
 Assumption \ref{ass: factor model}(\ref{item: NC})  holds.

Assumption \ref{ass: factor model}(\ref{item: SN}) places mild bounds on $X$ and
$Z_k$. For example,
if  the second moments of $X_{it}$ are (uniformly) bounded, then $
   \mathbb{E}[s_{1}(X)^2] \leq \mathbb{E}[\| X \|_F^2] = \sum_{i=1}^N \sum_{t=1}^T \mathbb{E}[X_{it}^2] = \Op(NT),
   $
   which, by Markov's inequality, implies 
   $s_{1}(X) = \Op(\sqrt{NT})$.
In addition Assumption \ref{ass: factor model}(\ref{item: SN}) also places a rate restriction on $s_1(U)$
that will hold as long as $U_{it}$
does not exhibit too much dependence over $i$ and $t$.
This rate for $s_1(U)$ is closely related to 
Assumption \ref{ass: factor model}(\ref{item: high level singular values of U}), which is discussed below.
Assumption \ref{ass: factor model}(\ref{item: EX}) again holds as long as $U_{it}$ does not
exhibit too much dependence over $i$ and $t$, and is uncorrelated with $X_{it}$ and $Z_{it}$.
Finally,
Assumption \ref{ass: factor model}(\ref{A_U_inner_product_bound}) holds as long as
$U$ is mean zero given $X$ and $Z$ and satisfies bounds on dependence and second
moments.

Assumption \ref{ass: factor model}(\ref{item: high level singular values of U})  is a high level assumption on the first few singular values of $U$ (note that $r$
is fixed as $N$ and $T$ converge to infinity). 
The singular values of $U$ are 
the square roots of the eigenvalues of $UU'$.
The random matrix theory literature shows that,
if $U$ is an appropriate noise matrix,  the 
largest few eigenvalues of  $UU'$ converge to the Tracy-Widom law, after appropriate rescaling:
if $N$ and $T$ grow at the same rate, then
each of the largest eigenvalues of $UU'$
grows at rate $N$, while the gaps between them grow
at rate $N^{1/3}$.
\cite{Johnstone2001} establish the  Tracy-Widom law
for the largest  eigenvalues of $UU'$, for
the case of i.i.d.\ normal
error $U_{it}$. The subsequent literature has
shown the universality of this result for
more general error distributions,
see e.g.\ \cite{Soshnikov2002}, \cite{pillai2012edge}
and \cite{yang2019edge}.

We also place conditions on the matrix $X$ requiring that there is sufficient
variation after controlling for individual effects and the additional covariates
$Z$.
The constant $c=c_{N,T}$ in the following
assumption is the one that appears
in our construction of $A^*_{b,c}$.

\begin{assumption}\label{ass: X decomposition}
  For all $P\in\mathcal{P}$, there exists uniformly bounded $\pi=\pi_P$ and random matrices $H$
  and $V$ such that $X=Z\cdot \pi + H + V$ and the following conditions hold:
  \begin{enumerate}[(i)]
  \item \label{V_assump} $\norm{V}_F \asy \sqrt{NT}$, $s_1 (V) = \Op (\maxSNT)$;

  \item \label{H_assump} $\norm{H}_F = \Op(\sqrt{NT})$ and $\langle H,V
    \rangle_F = \Op(\sqrt{NT})$;

  \item \label{Z_V_inner_product_assump} $\|Z_k\|_F=\Op(\sqrt{NT})$ and
    $\langle Z_k, V \rangle_F = \Op (\sqrt{N T})$ for $k \in \{1, \ldots, K\}$;
    
  \item \label{Z_prime_Z_assump} $\left(\bf{ Z{}' Z}\right)^{-1} = \Op\left(\frac{1}{NT}\right)$ where
    ${\bf Z} = [{\rm vec}(Z_1), \ldots, {\rm vec}(Z_K)]$;

  \item \label{max_V_Z_assump} $\max_{i,t} V_{it}^2 = \op (NT c_{N,T})$ and $\max_{i,t} Z_{k,it}^2 = \op \left( (NT)^2 c_{N,T}\right)$ for $k \in \{1,\ldots,K\}$.
  \end{enumerate}
  
\end{assumption}

Assumption \ref{ass: X decomposition} uses a decomposition of $X_{it}$ that
depends on an individual effect $H_{it}$ and a random variable $V_{it}$ that is
approximately independent and uncorrelated with $Z_{1,it}\ldots,Z_{k,it}$ as
well as being approximately uncorrelated with the individual effect $H_{it}$.
Importantly, the individual effect $H_{it}$ can be arbitrarily correlated with
$\Gamma_{it}$ and with the variables $Z_{k,it}$.  Note also that we do not place
any assumptions on the rank or nuclear norm of the matrix $H_{it}$.

Part (\ref{max_V_Z_assump}) holds under a tail bound on $V_{it}$ and $Z_{k,it}$.
For example, if $V_{it}$ are (uniformly) sub-Gaussian then $\max_{i,t} V_{it}^2 = \Op ( \log
(N + T) )$, and the condition $\max_{i,t} V_{it}^2 = \op (NT c_{N,T})$ is
satisfied provided that $N T c_{N,T} / \log (N + T) \rightarrow \infty$.
The only other requirement on $c_{N,T}$ is the requirement that
$c_{N,T} \max\{N,T\}\to 0$ in Theorem
\ref{new_validity_theorem} below.  Thus, our results allow
for a range of choices of $c_{N,T}$.

Define $P_\lambda = \lambda (\lambda' \lambda)^+ \lambda$ where $M^+$ denotes the Moore–Penrose inverse of a matrix $M$.

\begin{theorem}\label{new_nuclear_bound_error_thm}
  Let $\hat \Gamma_{\rm pre}$  be defined in Algorithm \ref{alg: factor model},
  with the modification described in Remark \ref{lindeberg_remark}.
  Suppose that Assumption \ref{ass: factor model} holds, and that
  Assumption \ref{ass: X decomposition} holds as stated and with $Z_k$ and $X$
  interchanged for each $k=1,\ldots,K$, for the given sequence $c=c_{N,T}$. Then, for any $\varepsilon > 0$, Assumption \ref{ass: augmented high level}\eqref{item: augmented high level C hat bound} holds with
  \begin{enumerate}[(i)]
    \item $\widetilde \Gamma = \Gamma - \hat \Gamma_{\rm pre}$ and $\hat C = 3 R s_1 (\hat U_{\rm pre}) (1 + \varepsilon) = \Op (\maxSNT)$; \label{item: 3 R bound}

    \item $\widetilde \Gamma = \Gamma + P_\lambda U - \hat \Gamma_{\rm pre}$ and $\hat C = 2 R s_1 (\hat U_{\rm pre}) (1 + \varepsilon) = \Op (\maxSNT)$. \label{item: 2 R bound}
  \end{enumerate}
\end{theorem}

Theorem \ref{new_nuclear_bound_error_thm} is the main novel technical result that allows us to construct a feasible CI.  It provides an explicit bound on the nuclear norm error of our initial estimate.  As we show in the proof of Theorem \ref{new_validity_theorem} below, the term $\langle A, P_\lambda U\rangle_F$ is asymptotically negligible under our assumptions.  Thus, redefining the target parameter to be $\Gamma + P_\lambda U$ instead of $\Gamma$ and using the bound in part (ii) of the theorem does not affect the construction of the CI.
This leads to a shorter CI using the bound in part (ii) compared to using the bound in part (i).  For this reason, we use the bound in part (ii) in the implementation described in Section \ref{sec:linear_factor_implementation} and in our formal coverage results below.

We now turn to the rate of convergence of the debiased estimator and the coverage of the CI.  The proofs of these theorems use the nuclear norm bounds in Theorem \ref{new_nuclear_bound_error_thm}.

\begin{theorem}\label{new_beta_rates_thm}
    Let $\hat \beta=\hat\beta_{A^{*}_{\bTW,c}}$ be defined in Algorithm \ref{alg: factor model},
    with the modification described in Remark \ref{lindeberg_remark}.
    Suppose that Assumption \ref{ass: factor model} holds, and that
    Assumption \ref{ass: X decomposition} holds as stated and with $Z_k$ and $X$
    interchanged for each $k=1,\ldots,K$, for the given sequence $c=c_{N,T}$.   Then
    \begin{align*}
      \hat\beta-\beta=\mathcal{O}_{\Theta,\mathcal{P}}(1/\min\{N,T\}).
    \end{align*}
\end{theorem}

To obtain primitive conditions for a central limit theorem and asymptotic
validity of the confidence interval, we impose that the errors are independent,
but not necessarily identically distributed, conditional on $X$, $Z$ and $\Gamma$.

\begin{assumption}\label{U_conditional_moment_bound_assump}
  There exist constants $\underline \sigma>0$ and $\eta>0$ such that, for all
  $P\in\mathcal{P}$, $U_{it}$ is independent over $i,t$ conditional on
  $W,\Gamma$ and, for all $i,t$,
  \begin{align*}
    \e_P [U_{it}|W,\Gamma]=0,
    \quad \e_P [U_{it}^2|W,\Gamma]>\underline\sigma^2,
    \quad \e_P [U_{it}^{4}|W,\Gamma]<1/\eta.
  \end{align*}
\end{assumption}

\begin{theorem}\label{new_validity_theorem}
    Let $\hat \beta=\hat\beta_{A^{*}_{\bTW,c}}$ and $\hat C = 2Rs_1(\hat
  U_{\rm pre})(1+\varepsilon)$ be defined in Algorithm \ref{alg: factor model},
  with the modification described in Remark \ref{lindeberg_remark} for
  $c=c_{N,T}$ with $c_{N,T} \max \{N,T\}\to 0$.
  Suppose that Assumptions \ref{ass: factor model}(\ref{item: NC})-(\ref{item: high level singular values of U}) hold, and that
  Assumption \ref{ass: X decomposition} holds as stated and with $Z_k$ and $X$
  interchanged for each $k=1,\ldots,K$, for the given sequence $c=c_{N,T}$,
  and that Assumption \ref{U_conditional_moment_bound_assump} holds.
  Let
  $\widehat{\se}^2=\sum_{i=1}^N\sum_{t=1}^T A_{it}^2\hat U_{it}^2$
  where $A=A^{*}_{\bTW,c}$ and $\hat U_{it}$ is the residual from the least
  squares estimator.
  Then
  \begin{align*}
    \hat\beta-\beta=\mathcal{O}_{\Theta,\mathcal{P}}(1/\min\{N,T\})
  \end{align*}
  and
  \begin{align*}
    \liminf \inf_{\theta\in\Theta,P\in\mathcal{P}} \pr_{\theta,P}\left( \beta \in \left\{ \hat\beta \pm \left[  \maxbias_{\hat C}(\hat\beta) + z_{1-\alpha/2}\widehat{\se}  \right] \right\} \right) \ge 1-\alpha.
  \end{align*}
\end{theorem}

\subsection{Strong factor case}
\label{ssec: semi strong}

Numerous studies on the estimation of panel regressions with unobserved factors assume that these factors are ``strong'' or ``semi-strong''. This assumption implies that the unobserved error structure, $\Gamma_{it} + U_{it}$ in model \eqref{model0}, viewed as an $N \times T$ matrix, contains an $R = \operatorname{rank}(\Gamma)$ factor component $\Gamma = \lambda f'$ with singular values that asymptotically diverge faster than the largest singular value, $s_1(U)$, of the idiosyncratic error part $U$. Specifically, as $N, T \rightarrow \infty$, the ratio $s_1(U) / s_R(\Gamma)$ approaches zero (at a certain rate) under the semi-strong (strong) factor assumptions. Both \cite{Pesaran2006estimation} and \cite{bai2009panel} impose conditions that imply strong factors in this sense, as do many subsequent papers.

A key motivation for the estimation approach in this paper is to avoid assuming strong factors, instead providing an inference method that remains uniformly valid regardless of factor strength. Nevertheless, it is natural to consider how our approach behaves when factors are, in fact, strong, if only to facilitate comparison with much of the existing literature. The following theorem, therefore, extends Theorem~\ref{new_nuclear_bound_error_thm} to accommodate the case of strong factors.

\begin{theorem}
   \label{th:SemiStrong}
     Let $\hat \Gamma_{\rm pre}$  be defined in Algorithm \ref{alg: factor model},
  with the modification described in Remark \ref{lindeberg_remark}. Suppose that the hypotheses of Theorem~\ref{new_validity_theorem} hold, and furthermore, assume that 
   $s_1(U) / s_R(\Gamma)= \op (1)$. Then  
   Assumption \ref{ass: augmented high level}\eqref{item: augmented high level C hat bound} holds with
    $\widetilde \Gamma =  \Gamma + M_{\lambda} U P_{f}  + P_{\lambda} U M_{f}  - \hat \Gamma_{\rm pre}$
   and 
   $\hat C =  o_{\Theta, \mathcal P} \left(   s_1(U)   \right)    = o_{\Theta, \mathcal P} (\maxSNT) $,
   where $P_f = f (f' f)^+ f'$, $M_\lambda = \mathbb I_N - P_\lambda$, and $M_f = \mathbb I_T - P_f$.
\end{theorem}

Theorem~\ref{th:SemiStrong} additionally requires $s_1(U) / s_R(\Gamma)= o_P(1)$, i.e., that the factors
are semi-strong or strong, and also that $R = \operatorname{rank}(\Gamma)$. This implies that 
$\hat \Gamma_{\rm pre}$ converges to 
$\Gamma + M_{\lambda} U P_{f}  + P_{\lambda} U M_{f}$ in second leading order
(see e.g.\ Lemma S.3 in the supplement to \citealp{MoonWeidner2015}). 
Theorem~\ref{th:SemiStrong} states that with this change of target matrix, the nuclear norm bound 
$\hat C$ is smaller than $ s_1(U) \asymp_{\Theta,\mathcal{P}} {(\maxSNT)}.$ This implies that, when $N$ and $T$ grow at the same rate, the worst-case bias $\maxbias_{\hat C}(\hat\beta)$ used in the construction of the CI in Theorem~\ref{new_validity_theorem} becomes negligible relative to the standard error $\widehat{\se}$. At the same time, the 
extra term $\langle A, M_{\lambda} U P_{f}  + P_{\lambda} U M_{f} \rangle_F$ is asymptotically negligible by exactly the same arguments given in the proof of Theorem~\ref{new_validity_theorem} for the negligibility of the term $\langle A, P_\lambda U \rangle_F$. As a result, in the considered regime, the debiased estimator is asymptotically unbiased and the (non-bias) aware CI \eqref{eq:non_bias_aware_ci} is asymptotically valid.

\begin{remark}
  \label{rem: R_w bound}

  Theorem \ref{th:SemiStrong} shows that the bias term is asymptotically negligible when $N$ and $T$ grow at the same rate and all $R$ factors are strong.  This justifies the CI (\ref{eq:non_bias_aware_ci}) discussed in Remark \ref{remark: strong factors CS} in the strong factor setting.  More generally, in the case where $R_w$ factors are weak and $R_s=R-R_w$ factors are strong, we conjecture that Theorem \ref{th:SemiStrong} could be extended to show that the bias-aware CI (\ref{eq: bias aware CI}) is valid with $\hat C = 2 R_w s_1 (\hat U_{\rm pre}) (1 + \varepsilon)$. 
 In other words, we conjecture that the worst-case bias of our estimator only depends on the number of weak factors.
\end{remark}

\subsection{Comparison to other results in the literature}\label{sec:comparison_to_other_results}

Our debiasing approach leads to the faster rate $\min \{N,T\}$ compared to the rate $\min \{\sqrt{N},\sqrt{T}\}$ for $\hat \beta_{\rm LS}$ (see, e.g., \citealp{MoonWeidner2015}).
While our results appear to be the first to demonstrate a $\min\{N,T\}$ rate of convergence under the conditions above, recent papers have proposed estimators that use additional structure to construct estimators that achieve the same or better rates.
\citet{chetverikov2022spectral} impose a factor structure on $X$, which corresponds to imposing a low-rank assumption on the matrix $H$ in our Assumption \ref{ass: X decomposition}.  They use this assumption to construct an estimator that, like ours, achieves a $\min\{N, T\}$ rate under weak factors.  \citet{zhu2019well} imposes homoskedastic and independent errors in addition to a factor structure on $X$, and shows that this allows for a faster $\sqrt{NT}$ rate of convergence, even under weak factors.

While robust to weak factors, our CI will be wider than a CI based on the strong factor asymptotics in \citet{bai2009panel}.  Ideally, one would like to form a CI that is \emph{adaptive} to the strength of factors.  Such a CI would be robust to weak factors, while being asymptotically equivalent to the CI in \citet{bai2009panel} when factors are strong.  However, as shown by \citet{zhu2019well}, such an adaptive CI cannot be obtained, even if one imposes homoskedastic errors and additional structure on the covariate matrix $X$.  Thus, while there may be some room for efficiency gains over our CI, one must allow for some increase in CI length relative to the CI in \citet{bai2009panel} in order to allow for weak factors.

As discussed in the introduction, our debiasing approach is analogous to the approach to debiasing the LASSO taken in \citet{javanmard_confidence_2014} and, more broadly, other papers in the debiased LASSO literature such as \citet{belloni_inference_2014}, \citet{van_de_geer_asymptotically_2014} and \citet{zhang_confidence_2014}.
Interestingly, this analogy extends to the rates of convergence in our asymptotic results.  The debiased lasso applies to a high dimensional regression model with $s$ nonzero coefficients and $n$ observations.  The resulting estimator has bias of order $s/n$, up to log terms, and variance $1/n$.  Note that $s$ is the dimension of the constraint set for the unknown parameter, while $n$ is the total number of observations.  In our setting, the debiased estimator has bias of order $\max\{N,T\}/(NT)$ and variance $1/(NT)$.  The set of matrices $\Gamma$ with rank at most $R$ has dimension of order $\max\{N,T\}$ so, just as with the debiased lasso, the bias term is of the same order of magnitude as the ratio of the dimension of the constraint set to the total number of observations.  In the debiased lasso setting, one can justify a CI that ignores bias by assuming that $s$ increases slowly enough relative to $n$ for the order $s/n$ bias term to be asymptotically negligible relative to the order $1/\sqrt{n}$ standard deviation term.  Unfortunately, this cannot occur in our setting even if $R=1$, since the bias term is of order $\max\{N,T\}/(NT)$ which is always of at least the same order of magnitude as the standard deviation $1/\sqrt{NT}$.  This necessitates our bias-aware approach.

\section{Numerical Evidence}
\label{sec: numerical}
\subsection{Simulation Study}
\label{ssec: MC}
We consider the following design:
\begin{align*}
  Y_{it} &= X_{it} \beta + \sum_{r=1}^R \kappa_r \lambda_{ir} f_{tr} + U_{it} ,\\
  X_{it} &= \sum_{r=1}^R \lambda_{ir} f_{tr} + V_{it},
\end{align*}
where $\kappa_r$ controls the strength of factor $f_{tr}$, and $R$ stands for the number of factors. In addition, $\lambda_i$, $f_t$, $U_{it}$ and $V_{it}$ are all mutually independent across both $i$, $t$, and $(i,t)$, and
\begin{align*}
  \lambda_i \sim N(0,I_R) \perp f_t \sim N(0,I_R) \perp  \begin{pmatrix} U_{it} \\ V_{it} \end{pmatrix} \sim N \left(\begin{pmatrix} 0 \\ 0 \end{pmatrix}, \begin{pmatrix} \sigma_U^2 & 0 \\ 0 & \sigma_V^2 \end{pmatrix}\right).
\end{align*}

In the designs considered below, we fix $(\beta, \sigma_U^2, \sigma_V^2) = (0,1,1)$ and vary $N$, $T$, the number of factors $R$, and their strengths controlled by $\kappa_r$. The number of simulations in all of the considered designs is 5000.
As before, we are interested in estimation of and inference on $\beta$.

In Tables \ref{tab: N100 R1}-\ref{tab: N100 T50 R2 inference}, we report the bias, standard deviation, and rmse for the benchmark LS estimator of \citet{bai2009panel} and for the proposed debiased estimator in various designs with 1 and 2 factors.\footnote{Note that the CCE estimator of \cite{Pesaran2006estimation} would not work in these designs, regardless of whether the factors are strong or not, because the cross-sectional averages of $\lambda_{ir}$ equal zero.} We also report the size of the corresponding tests (with $5\%$ nominal size) and the average length of the CIs (with $95\%$ nominal coverage). For simplicity and for brevity of the reported results, we assume that the number of factors is known. We consider a case when the number of factors is overspecified in Appendix \ref{ssec: MC extended}.

The LS estimator is heavily biased and the associated tests and CIs are heavily size distorted unless all the factors are strong. At the same time, the proposed estimator effectively reduces the ``weak factors'' bias without inflating the variance. As a result, the potential efficiency gains from using the debiased estimator can be very large when there is a weak factor, especially for larger sample sizes (see Appendix \ref{ssec: additional simulation results} for additional simulation results). Importantly, even if all the factors are strong, the debiased estimator performs comparably to the LS estimator.

When weak factors are present, the LS CIs can have zero coverage because they are (i) centered around the biased LS estimator and (ii) too short. Hence, the average length of the LS CIs is not a proper benchmark to compare the average length of the bias-aware CIs. To provide a relevant comparison, we also construct identification robust CIs by inverting the (absolute value of the) LS based t-statistic using appropriate identification robust critical values (instead of $z_{1-\alpha/2}$). Specifically, for a given design (here, for fixed $N$, $T$, and $R$), we (numerically) compute the least favorable (over $\kappa$) critical value for the absolute value of the t-statistic based on the LS estimator. We also construct analogous CIs by inverting the (absolute value of the) t-statistic based on the debiased estimator using the corresponding least favorable critical values. We refer to such CIs as the LS and debiased oracle CIs (because they are based on unknown design-specific least favorable critical values) and report their average length denoted by ``length*'' in the tables below.

Notice that the average length of the LS oracle CIs (the ``length*'' column under the LS heading) 
is at least comparable to but mostly significantly greater than the actual length of the bias-aware CIs (the ``length'' column under the debiased heading), especially for larger sample sizes (again, see Appendix \ref{ssec: additional simulation results} for additional simulation results).
Thus, the bias-aware CI outperforms the LS CI once one corrects the LS CI to compensate for its severe undercoverage.

Another important comparison is between the actual length and oracle length of the bias-aware CI.  Throughout most of the designs, the oracle length of the bias-aware CI is slightly less than half the length of the actual bias-aware CI that we compute.  This gives a bound on how conservative our CI is: our bias-aware critical value cannot be decreased by more than a factor of about two without sacrificing coverage in these Monte Carlos.  There are two possible sources of this conservativeness: (1) the bound in Theorem \ref{new_nuclear_bound_error_thm} may be conservative or (2) there may be some additional structure in the initial error or its correlation with the data that our nuclear norm debiasing method does not exploit.  While further improving the CI using the proof techniques in this paper appears difficult, we cannot rule out these possibilities.  On the other hand, it is possible that these Monte Carlos overstate the conservativeness of our bias-aware CI: there may be other DGPs for which our bias-aware critical value cannot be decreased without sacrificing coverage.

Despite the simplicity of the design considered in this section, the presented findings seem to be characteristic of more complicated and settings. Specifically, in Appendix \ref{ssec: MC extended}, we consider a design with an additional covariate and non-Gaussian, heteroskedastic, and serially correlated errors and establish qualitatively similar results regardless of whether the correct number of factors is known or overspecified.

\subsection{Empirical Illustration}
\label{ssec: empirical}
In this section, we illustrate the finite sample properties of the proposed estimator and confidence intervals in a numerical experiment calibrated to imitate an actual empirical setting. Specifically, we calibrate our experiment based on the seminal studies of the effects of unilateral divorce law reforms on the US divorce rates by \citet{Friedberg1998} and \citet{Wolfers2006}, subsequently revisited by \citet{KimOka2014} and \citet{MoonWeidner2015} in the context of interactive fixed effects models.

For simplicity of the experiment, as a benchmark, we use the following static specification also considered in \citet{Friedberg1998} and \citet{Wolfers2006}
\begin{align*}
  Y_{it} = X_{it} \beta + \alpha_i + \zeta_i t + \nu_i t^2 + \phi_t + U_{it},
\end{align*}
where $Y_{it}$ denotes the annual divorce rate (per 1,000 persons) in state $i$ in year $t$, and $X_{it}$ is a dummy variable indicating if state $i$ had a unilateral divorce law in year $t$. Following \citet{Friedberg1998} and \citet{Wolfers2006}, we also control for state-specific quadratic time trends and time effects.

We follow \citet{KimOka2014} and use their data to construct a balanced panel with $N = 48$ states and $T = 33$ years. As in \citet{MoonWeidner2015}, first we profile out the individual trends and time effects from $Y_{it}$ and $X_{it}$ to form the projected model
\begin{align*}
  Y_{it}^{\perp} = X_{it}^{\perp} \beta + U_{it}^{\perp}
\end{align*}
and obtain the estimates $\hat \beta$ and $\hat \sigma_{U^\perp}^2$. We also extract the first principal component of the matrix of regressors $X^\perp$ denoted by $\Gamma^{X^\perp} = \lambda_i^{X^\perp} f_t^{X^\perp}$.

In our numerical experiment, we fix $X^\perp$, $\{\lambda_i^{X^\perp}\}_{i=1}^N$, and $\{f_t^{X^\perp}\}_{t=1}^T$, and consider the following DGP
\begin{align*}
  Y_{it}^{\perp} = X_{it}^\perp \hat \beta + \kappa \lambda_i^{X^\perp} f_t^{X^\perp} + U_{it}^\perp,
\end{align*}
where we introduce an additional factor $f_{t}^{X^\perp}$ and a parameter $\kappa$ controlling the strength of $f_{t}^{X^\perp}$. For every repetition, we draw $U_{it}^{\perp}$ as iid $N(0, \hat \sigma_U^2)$ and treat all the other parts of the DGP as fixed.

As before, we compare the LS estimator and inference performance with the proposed approach for various values of $\kappa$. Both approaches use the correctly specified number of factors $R = 1$. The results are based on 5,000 simulations and provided in Table \ref{tab: empriical MC}. We report the same statistics as in Section \ref{ssec: MC}.

The results are qualitatively similar to the results in Section \ref{ssec: MC}. The LS estimator is heavily biased when the factor is weak, and the standard tests and confidence intervals are severely size distorted. Compared to the LS estimator, the debiased estimator has a substantially smaller bias, standard deviation, and rmse when the factor is weak. It also performs competitively if the factor is strong. The LS CIs are much shorter than the bias-aware CIs but have very poor coverage. The oracle CIs based on the LS estimator have the correct coverage and are also considerably wider than the naive CIs and comparable with the bias-aware CIs. Again, the oracle CIs based on the debiased estimator are considerably shorter than the bias-aware CIs and LS oracle CIs, indicating that there is a potential scope for improvement.

Overall, our empirically calibrated simulation study shows that the presence of a weak factor can lead to poor performance of conventional estimators and inference procedures in an actual empirical setting. It also demonstrates that in such settings, the gains from using the debiased estimator could be substantial.

Finally, we also report estimation and inference results for the actual data set. For consistency with the numerical experiment above, we focus on the same single covariate $X_{it}$. In Appendix \ref{sec: empirical app}, we also consider a specification with dynamic treatment effects as in \citet{Wolfers2006}. Similarly to \citet{KimOka2014} and \citet{MoonWeidner2015}, we estimate
\begin{align*}
  Y_{it} = X_{it} \beta + \alpha_i + \zeta_i t + \nu_i t^2 + \phi_t + \sum_{r=1}^R \lambda_{ir} f_{tr} + U_{it}
\end{align*}
for various values of $R$ using the LS and the debiased approaches and construct 95\% CIs for $\beta$. As before, we first profile out the individual trends and time effects, and then use the residual outcomes and regressors as inputs for the LS and debiased estimators.

The results are provided in Table \ref{tab: empirical res}.
We construct three types of CIs based on the debiased estimator using $\hat C$ as in Remark~\ref{rem: R_w bound} with different values of $R_w$. The first one is constructed assuming that there are no weak factors among $R$ factors ($R_w = 0$), i.e., under the same assumption under which the standard LS CI is valid. This is the same CI as introduced in Remark~\ref{remark: strong factors CS}. In this application, these CIs are as short as or even shorter than the LS CIs. Thus, if the researcher wants to obtain shorter CIs at the cost of non-robustness to the potential presence of weak factors, they can still do that using our debiased estimator. As pointed out in Remark~\ref{remark: strong factors CS} and documented in Section~\ref{ssec: MC non-robust coverage}, such CIs are still likely to have much better coverage than the LS ones when there is a weak factors since they are based on the debiased estimator.

The second type of CIs is constructed assuming that among $R$ factors there is up to one weak factor ($R_w = 1$). The corresponding bias-aware CIs are substantially wider than the non-robust ones. However, as the numerical experiment considered earlier in this sections suggests, this is how wide identification robust CIs appear to have to be in this setting. In the considered application, we find that the potential presence of one weak factor is likely to be sufficient to nullify the significance of the previously obtained non-robust estimates.

Finally, we also report our bias-aware CIs as in \eqref{eq: bias aware CI} corresponding to $R_w = R$.  These CIs are uniformly valid regardless of the strength of identification of the factors.

The results for a specification with dynamic treatment effects are qualitatively similar and provided in Appendix~\ref{sec: empirical app}.

\begin{table}[H]
\caption{Simulation results for the experiment in Section \ref{ssec: MC}, $N = 100$, $R = 1$}
\label{tab: N100 R1}
\begin{center}
\resizebox{\textwidth}{!}{
\begin{tabular}{c c c c c c c| c c c c c c}
&\multicolumn{6}{c}{LS}&\multicolumn{6}{c}{Debiased} \\
\cmidrule(lr){2-7}  \cmidrule(lr){8-13}
{$\kappa$}&{bias}&{std}&{rmse}&{size}&{length}&{length*}&{bias}&{std}&{rmse}&{size}&{length}&{length*}\\
\midrule
\multicolumn{13}{c}{$T = 20$}\\
\midrule
0.00&-0.0000&0.0171&0.0171&7.1&0.061&0.299&0.0002&0.0206&0.0206&0.0&0.304&0.137 \\
0.05&0.0242&0.0178&0.0300&37.3&0.062&0.300&0.0095&0.0207&0.0228&0.0&0.304&0.137 \\
0.10&0.0478&0.0200&0.0518&79.3&0.062&0.302&0.0181&0.0215&0.0281&0.0&0.305&0.137 \\
0.15&0.0690&0.0249&0.0734&91.6&0.063&0.308&0.0244&0.0235&0.0339&0.0&0.306&0.138 \\
0.20&0.0792&0.0382&0.0879&85.7&0.067&0.324&0.0250&0.0276&0.0372&0.0&0.309&0.138 \\
0.25&0.0670&0.0531&0.0855&64.8&0.074&0.358&0.0189&0.0306&0.0360&0.0&0.311&0.139 \\
0.50&0.0049&0.0244&0.0248&8.2&0.087&0.425&0.0013&0.0239&0.0240&0.0&0.314&0.139  \\
1.00&0.0004&0.0232&0.0232&5.9&0.088&0.427&0.0001&0.0237&0.0237&0.0&0.315&0.140  \\
\midrule
\multicolumn{13}{c}{$T = 50$}\\
\midrule
0.00&-0.0002&0.0103&0.0103&5.9&0.039&0.228&-0.0001&0.0136&0.0136&0.0&0.173&0.079\\
0.05&0.0244&0.0108&0.0267&67.5&0.039&0.228&0.0064&0.0137&0.0151&0.0&0.173&0.080 \\
0.10&0.0484&0.0124&0.0500&98.2&0.039&0.230&0.0121&0.0143&0.0187&0.0&0.174&0.080 \\
0.15&0.0683&0.0189&0.0709&96.8&0.040&0.237&0.0135&0.0164&0.0213&0.0&0.175&0.080 \\
0.20&0.0580&0.0390&0.0699&72.4&0.046&0.269&0.0084&0.0180&0.0198&0.0&0.177&0.080 \\
0.25&0.0229&0.0306&0.0382&33.5&0.053&0.308&0.0032&0.0164&0.0167&0.0&0.177&0.080 \\
0.50&0.0016&0.0144&0.0145&5.7&0.055&0.324&0.0002&0.0151&0.0151&0.0&0.177&0.080  \\
1.00&0.0001&0.0142&0.0142&5.1&0.055&0.324&-0.0001&0.0151&0.0151&0.0&0.178&0.080 \\
\midrule
\multicolumn{13}{c}{$T = 100$}\\
\midrule
0.00&-0.0001&0.0073&0.0073&6.1&0.028&0.183&-0.0001&0.0108&0.0108&0.0&0.122&0.057 \\
0.05&0.0246&0.0077&0.0258&91.0&0.028&0.183&0.0051&0.0108&0.0120&0.0&0.122&0.057  \\
0.10&0.0486&0.0093&0.0495&99.9&0.028&0.185&0.0089&0.0117&0.0147&0.0&0.123&0.057  \\
0.15&0.0619&0.0224&0.0658&92.9&0.030&0.197&0.0068&0.0132&0.0149&0.0&0.124&0.058  \\
0.20&0.0239&0.0267&0.0358&47.4&0.037&0.243&0.0023&0.0124&0.0127&0.0&0.124&0.058  \\
0.25&0.0077&0.0120&0.0143&17.9&0.039&0.256&0.0009&0.0119&0.0120&0.0&0.124&0.058  \\
0.50&0.0009&0.0103&0.0103&5.9&0.039&0.260&0.0001&0.0118&0.0118&0.0&0.124&0.058   \\
1.00&0.0001&0.0102&0.0102&5.4&0.039&0.260&-0.0000&0.0118&0.0118&0.0&0.124&0.058  \\
\midrule
\multicolumn{13}{c}{$T = 300$}\\
\midrule
0.00&-0.0000&0.0042&0.0042&5.3&0.016&0.121&0.0000&0.0056&0.0056&0.0&0.080&0.033 \\
0.05&0.0247&0.0046&0.0252&100.0&0.016&0.122&0.0047&0.0056&0.0073&0.0&0.080&0.033\\
0.10&0.0482&0.0070&0.0487&99.8&0.016&0.123&0.0057&0.0067&0.0088&0.0&0.080&0.033 \\
0.15&0.0178&0.0173&0.0248&60.5&0.021&0.161&0.0015&0.0063&0.0065&0.0&0.081&0.033 \\
0.20&0.0047&0.0064&0.0080&16.4&0.022&0.170&0.0005&0.0061&0.0061&0.0&0.081&0.033 \\
0.25&0.0023&0.0060&0.0064&7.6&0.023&0.171&0.0003&0.0061&0.0061&0.0&0.081&0.033  \\
0.50&0.0003&0.0057&0.0058&4.9&0.023&0.172&0.0001&0.0060&0.0060&0.0&0.081&0.033  \\
1.00&0.0001&0.0057&0.0057&4.9&0.023&0.172&0.0000&0.0060&0.0060&0.0&0.081&0.033  \\
\bottomrule
\bottomrule
\multicolumn{13}{p{1.2\textwidth}}{${\rm Lind}(A) \in \{0.0063,0.0028,0.0015,0.0006\}$ for $T \in \{20,50,100,300\}$.} 
\end{tabular}
}
\end{center}
\end{table}

\newgeometry{left=1.5cm,right=1.5cm,top=3cm,bottom=2cm}
\begin{landscape}
\begin{table}[H]
\caption{Simulation results for the experiment in Section \ref{ssec: MC}, $N = 100$, $T = 50$, $R = 2$}
\label{tab: N100 T50 R2 estimation}
\begin{scriptsize}
\begin{center}
\begin{tabular}{ l| c c c c c c c c c c| c c c c c c c c c c}
     & \multicolumn{10}{c}{LS} & \multicolumn{10}{c}{Debiased}\\
     \diagbox[width=12mm,height=8mm]{$\kappa_1$}{$\kappa_2$}&{0.00}&{0.05}&{0.10}&{0.15}&{0.20}&{0.25}&{0.30}&{0.40}&{0.50}&{1.00}&{0.00}&{0.05}&{0.10}&{0.15}&{0.20}&{0.25}&{0.30}&{0.40}&{0.50}&{1.00}\\
     \midrule
     \multicolumn{21}{c}{bias}\\
     \midrule
     0.00 & -0.000&0.016&0.031&0.035&0.018&0.008&0.004&0.002&0.001&0.000&-0.000&0.006&0.010&0.010&0.005&0.002&0.001&0.000&0.000&-0.000\\
     0.05 & 0.016&0.033&0.049&0.060&0.052&0.036&0.030&0.027&0.026&0.025&0.006&0.012&0.017&0.018&0.014&0.010&0.009&0.008&0.007&0.007\\
     0.10 & 0.031&0.049&0.066&0.080&0.083&0.067&0.057&0.051&0.050&0.049&0.010&0.017&0.022&0.025&0.022&0.018&0.015&0.014&0.014&0.013\\
     0.15 & 0.035&0.060&0.080&0.097&0.108&0.099&0.082&0.072&0.070&0.068&0.010&0.018&0.025&0.029&0.028&0.024&0.020&0.018&0.017&0.017\\
     0.20 & 0.018&0.052&0.083&0.108&0.123&0.114&0.089&0.068&0.064&0.060&0.005&0.014&0.022&0.028&0.029&0.024&0.018&0.014&0.013&0.012\\
     0.25 & 0.008&0.036&0.067&0.099&0.114&0.088&0.054&0.032&0.028&0.025&0.002&0.010&0.018&0.024&0.024&0.017&0.010&0.006&0.005&0.005\\
     0.30 & 0.004&0.030&0.057&0.082&0.089&0.054&0.027&0.015&0.012&0.010&0.001&0.009&0.015&0.020&0.018&0.010&0.005&0.003&0.002&0.002\\
     0.40 & 0.002&0.027&0.051&0.072&0.068&0.032&0.015&0.007&0.005&0.004&0.000&0.008&0.014&0.018&0.014&0.006&0.003&0.001&0.001&0.001\\
     0.50 & 0.001&0.026&0.050&0.070&0.064&0.028&0.012&0.005&0.004&0.002&0.000&0.007&0.014&0.017&0.013&0.005&0.002&0.001&0.001&0.000\\
     1.00 & 0.000&0.025&0.049&0.068&0.060&0.025&0.010&0.004&0.002&0.000&-0.000&0.007&0.013&0.017&0.012&0.005&0.002&0.001&0.000&-0.000\\
     \midrule
     \multicolumn{21}{c}{std}\\
     \midrule
     0.00 & 0.009&0.009&0.012&0.019&0.019&0.013&0.011&0.011&0.010&0.010&0.013&0.013&0.014&0.015&0.015&0.014&0.014&0.014&0.014&0.014\\
     0.05 & 0.009&0.009&0.010&0.014&0.021&0.015&0.012&0.011&0.011&0.011&0.013&0.013&0.013&0.015&0.015&0.014&0.014&0.014&0.014&0.014\\
     0.10 & 0.012&0.010&0.010&0.012&0.019&0.020&0.014&0.013&0.013&0.013&0.014&0.013&0.014&0.015&0.016&0.015&0.015&0.014&0.014&0.014\\
     0.15 & 0.019&0.014&0.012&0.012&0.017&0.025&0.021&0.019&0.019&0.020&0.015&0.015&0.015&0.016&0.016&0.017&0.016&0.016&0.016&0.016\\
     0.20 & 0.019&0.021&0.019&0.017&0.025&0.043&0.045&0.040&0.039&0.039&0.015&0.015&0.016&0.016&0.018&0.020&0.020&0.019&0.019&0.019\\
     0.25 & 0.013&0.015&0.020&0.025&0.043&0.064&0.055&0.037&0.034&0.032&0.014&0.014&0.015&0.017&0.020&0.023&0.021&0.018&0.018&0.017\\
     0.30 & 0.011&0.012&0.014&0.021&0.045&0.055&0.036&0.021&0.019&0.019&0.014&0.014&0.015&0.016&0.020&0.021&0.018&0.016&0.016&0.016\\
     0.40 & 0.011&0.011&0.013&0.019&0.040&0.037&0.021&0.016&0.015&0.015&0.014&0.014&0.014&0.016&0.019&0.018&0.016&0.016&0.016&0.016\\
     0.50 & 0.010&0.011&0.013&0.019&0.039&0.034&0.019&0.015&0.015&0.015&0.014&0.014&0.014&0.016&0.019&0.018&0.016&0.016&0.015&0.015\\
     1.00 & 0.010&0.011&0.013&0.020&0.039&0.032&0.019&0.015&0.015&0.015&0.014&0.014&0.014&0.016&0.019&0.017&0.016&0.016&0.015&0.015\\
     \midrule
     \multicolumn{21}{c}{rmse}\\
     \midrule
     0.00 & 0.009&0.019&0.033&0.039&0.026&0.015&0.012&0.011&0.010&0.010&0.013&0.014&0.017&0.018&0.016&0.014&0.014&0.014&0.014&0.014\\
     0.05 & 0.019&0.034&0.050&0.061&0.056&0.039&0.032&0.029&0.028&0.027&0.014&0.017&0.021&0.023&0.021&0.017&0.016&0.016&0.016&0.016\\
     0.10 & 0.033&0.050&0.066&0.081&0.085&0.070&0.058&0.053&0.051&0.050&0.017&0.021&0.026&0.029&0.027&0.023&0.021&0.020&0.020&0.020\\
     0.15 & 0.039&0.061&0.081&0.098&0.109&0.102&0.085&0.075&0.073&0.071&0.018&0.023&0.029&0.033&0.033&0.029&0.026&0.024&0.024&0.023\\
     0.20 & 0.026&0.056&0.085&0.109&0.125&0.122&0.099&0.079&0.075&0.071&0.016&0.021&0.027&0.033&0.034&0.032&0.027&0.024&0.023&0.022\\
     0.25 & 0.015&0.039&0.070&0.102&0.122&0.109&0.077&0.049&0.044&0.041&0.014&0.017&0.023&0.029&0.032&0.028&0.023&0.019&0.018&0.018\\
     0.30 & 0.012&0.032&0.058&0.085&0.099&0.077&0.045&0.026&0.023&0.021&0.014&0.016&0.021&0.026&0.027&0.023&0.018&0.016&0.016&0.016\\
     0.40 & 0.011&0.029&0.053&0.075&0.079&0.049&0.026&0.017&0.016&0.016&0.014&0.016&0.020&0.024&0.024&0.019&0.016&0.016&0.016&0.016\\
     0.50 & 0.010&0.028&0.051&0.073&0.075&0.044&0.023&0.016&0.015&0.015&0.014&0.016&0.020&0.024&0.023&0.018&0.016&0.016&0.016&0.015\\
     1.00 & 0.010&0.027&0.050&0.071&0.071&0.041&0.021&0.016&0.015&0.015&0.014&0.016&0.020&0.023&0.022&0.018&0.016&0.016&0.015&0.015\\
     \bottomrule
     \bottomrule
\end{tabular}
\end{center}
\end{scriptsize}
\end{table}

\begin{table}[H]
\caption{Simulation results for the experiment in Section \ref{ssec: MC}, $N = 100$, $T = 50$, $R = 2$}
\label{tab: N100 T50 R2 inference}
\begin{scriptsize}
\begin{center}
\begin{tabular}{ l| c c c c c c c c c c| c c c c c c c c c c}
     & \multicolumn{10}{c}{LS} & \multicolumn{10}{c}{Debiased}\\
     \diagbox[width=12mm,height=8mm]{$\kappa_1$}{$\kappa_2$}&{0.00}&{0.05}&{0.10}&{0.15}&{0.20}&{0.25}&{0.30}&{0.40}&{0.50}&{1.00}&{0.00}&{0.05}&{0.10}&{0.15}&{0.20}&{0.25}&{0.30}&{0.40}&{0.50}&{1.00}\\
     \midrule
     \multicolumn{21}{c}{size}\\
     \midrule
     0.00 & 7.6&52.3&88.3&80.0&42.8&17.9&10.6&6.7&6.3&5.9&0.0&0.0&0.0&0.0&0.0&0.0&0.0&0.0&0.0&0.0\\
     0.05 & 52.3&96.2&99.8&99.6&95.9&88.3&81.6&73.7&70.8&68.0&0.0&0.0&0.0&0.0&0.0&0.0&0.0&0.0&0.0&0.0\\
     0.10 & 88.3&99.8&100.0&100.0&100.0&99.7&99.2&98.6&98.4&98.0&0.0&0.0&0.0&0.0&0.0&0.0&0.0&0.0&0.0&0.0\\
     0.15 & 80.0&99.6&100.0&100.0&99.8&99.5&98.7&97.9&97.4&96.6&0.0&0.0&0.0&0.0&0.0&0.0&0.0&0.0&0.0&0.0\\
     0.20 & 42.8&95.9&100.0&99.8&98.3&93.1&87.3&80.0&76.9&73.8&0.0&0.0&0.0&0.0&0.0&0.0&0.0&0.0&0.0&0.0\\
     0.25 & 17.9&88.3&99.7&99.5&93.1&76.1&60.5&45.2&39.9&36.2&0.0&0.0&0.0&0.0&0.0&0.0&0.0&0.0&0.0&0.0\\
     0.30 & 10.6&81.6&99.2&98.7&87.3&60.5&38.9&23.3&18.9&16.2&0.0&0.0&0.0&0.0&0.0&0.0&0.0&0.0&0.0&0.0\\
     0.40 & 6.7&73.7&98.6&97.9&80.0&45.2&23.3&11.3&9.2&7.9&0.0&0.0&0.0&0.0&0.0&0.0&0.0&0.0&0.0&0.0\\
     0.50 & 6.3&70.8&98.4&97.4&76.9&39.9&18.9&9.2&7.4&6.4&0.0&0.0&0.0&0.0&0.0&0.0&0.0&0.0&0.0&0.0\\
     1.00 & 5.9&68.0&98.0&96.6&73.8&36.2&16.2&7.9&6.4&5.7&0.0&0.0&0.0&0.0&0.0&0.0&0.0&0.0&0.0&0.0\\
     \midrule
     \multicolumn{21}{c}{length}\\
     \midrule
     0.00 & 0.031&0.031&0.032&0.034&0.037&0.038&0.039&0.039&0.039&0.039&0.257&0.257&0.258&0.260&0.261&0.262&0.262&0.262&0.262&0.262\\
     0.05 & 0.031&0.031&0.032&0.033&0.036&0.038&0.038&0.039&0.039&0.039&0.257&0.257&0.258&0.259&0.261&0.262&0.262&0.262&0.262&0.262\\
     0.10 & 0.032&0.032&0.032&0.032&0.034&0.037&0.039&0.039&0.039&0.039&0.258&0.258&0.258&0.260&0.261&0.262&0.263&0.263&0.263&0.263\\
     0.15 & 0.034&0.033&0.032&0.032&0.033&0.036&0.039&0.040&0.040&0.040&0.260&0.259&0.260&0.261&0.263&0.264&0.265&0.265&0.265&0.265\\
     0.20 & 0.037&0.036&0.034&0.033&0.034&0.037&0.042&0.045&0.045&0.046&0.261&0.261&0.261&0.263&0.265&0.266&0.267&0.267&0.268&0.268\\
     0.25 & 0.038&0.038&0.037&0.036&0.037&0.043&0.049&0.052&0.052&0.052&0.262&0.262&0.262&0.264&0.266&0.268&0.268&0.269&0.269&0.269\\
     0.30 & 0.039&0.038&0.039&0.039&0.042&0.049&0.053&0.054&0.054&0.054&0.262&0.262&0.263&0.265&0.267&0.268&0.269&0.269&0.269&0.269\\
     0.40 & 0.039&0.039&0.039&0.040&0.045&0.052&0.054&0.055&0.055&0.055&0.262&0.262&0.263&0.265&0.267&0.269&0.269&0.269&0.269&0.270\\
     0.50 & 0.039&0.039&0.039&0.040&0.045&0.052&0.054&0.055&0.055&0.055&0.262&0.262&0.263&0.265&0.268&0.269&0.269&0.269&0.270&0.270\\
     1.00 & 0.039&0.039&0.039&0.040&0.046&0.052&0.054&0.055&0.055&0.055&0.262&0.262&0.263&0.265&0.268&0.269&0.269&0.270&0.270&0.270\\
     \midrule
     \multicolumn{21}{c}{length*}\\
     \midrule
     0.00 & 0.345&0.346&0.353&0.373&0.406&0.420&0.424&0.427&0.427&0.428&0.114&0.114&0.114&0.115&0.115&0.115&0.115&0.115&0.115&0.115\\
     0.05 & 0.346&0.346&0.349&0.360&0.391&0.416&0.424&0.427&0.428&0.429&0.114&0.114&0.114&0.115&0.115&0.115&0.115&0.115&0.115&0.115\\
     0.10 & 0.353&0.349&0.348&0.354&0.375&0.409&0.424&0.430&0.432&0.432&0.114&0.114&0.114&0.115&0.115&0.115&0.115&0.115&0.115&0.116\\
     0.15 & 0.373&0.360&0.354&0.354&0.365&0.399&0.428&0.441&0.443&0.445&0.115&0.115&0.115&0.115&0.115&0.116&0.116&0.116&0.116&0.116\\
     0.20 & 0.406&0.391&0.375&0.365&0.371&0.410&0.459&0.491&0.497&0.503&0.115&0.115&0.115&0.115&0.116&0.116&0.116&0.116&0.116&0.116\\
     0.25 & 0.420&0.416&0.409&0.399&0.410&0.475&0.536&0.568&0.573&0.576&0.115&0.115&0.115&0.116&0.116&0.116&0.116&0.116&0.116&0.116\\
     0.30 & 0.424&0.424&0.424&0.428&0.459&0.536&0.579&0.595&0.598&0.599&0.115&0.115&0.115&0.116&0.116&0.116&0.116&0.116&0.116&0.117\\
     0.40 & 0.427&0.427&0.430&0.441&0.491&0.568&0.595&0.604&0.606&0.607&0.115&0.115&0.115&0.116&0.116&0.116&0.116&0.117&0.117&0.117\\
     0.50 & 0.427&0.428&0.432&0.443&0.497&0.573&0.598&0.606&0.608&0.609&0.115&0.115&0.115&0.116&0.116&0.116&0.116&0.117&0.117&0.117\\
     1.00 & 0.428&0.429&0.432&0.445&0.503&0.576&0.599&0.607&0.609&0.610&0.115&0.115&0.116&0.116&0.116&0.116&0.117&0.117&0.117&0.117\\
     \bottomrule
     \bottomrule
\end{tabular}
\end{center}
\end{scriptsize}
\end{table}

\end{landscape}
\restoregeometry

\begin{table}[H]
\begin{center}
\caption{Simulation results for the empirically calibrated experiment, $N = 48$, $T = 33$, $R = 1$}
\label{tab: empriical MC}
\resizebox{\textwidth}{!}{
\begin{tabular}{c c c c c c c| c c c c c c}
&\multicolumn{6}{c}{LS}&\multicolumn{6}{c}{Debiased}\\
\cmidrule(lr){2-7}  \cmidrule(lr){8-13}
{$\kappa$}&{bias}&{std}&{rmse}&{size}&{length}&{length*}&{bias}&{std}&{rmse}&{size}&{length}&{length*}\\
\midrule
0.00& -0.0007& 0.0647& 0.0647& 6.9& 0.236& 1.062& -0.0010& 0.0797& 0.0797& 0.0& 1.374& 0.683  \\
0.20& 0.0920& 0.0656& 0.1130& 35.0& 0.237& 1.067& 0.0517& 0.0805& 0.0957& 0.0& 1.376& 0.685  \\
0.40& 0.1822& 0.0703& 0.1953& 81.9& 0.240& 1.083& 0.1013& 0.0842& 0.1317& 0.0& 1.381& 0.690  \\
0.60& 0.2620& 0.0890& 0.2767& 93.4& 0.248& 1.116& 0.1392& 0.0966& 0.1694& 0.0& 1.392& 0.698  \\
0.80& 0.2999& 0.1467& 0.3339& 84.6& 0.262& 1.180& 0.1428& 0.1267& 0.1910& 0.0& 1.406& 0.703  \\
1.00& 0.2356& 0.2168& 0.3202& 57.5& 0.286& 1.288& 0.0972& 0.1470& 0.1762& 0.0& 1.418& 0.700  \\
1.20& 0.1134& 0.1955& 0.2260& 26.8& 0.307& 1.381& 0.0420& 0.1258& 0.1326& 0.0& 1.424& 0.693  \\
1.40& 0.0427& 0.1168& 0.1243& 12.6& 0.315& 1.419& 0.0177& 0.1042& 0.1057& 0.0& 1.426& 0.690  \\
1.60& 0.0235& 0.0921& 0.0950& 8.8& 0.317& 1.428& 0.0104& 0.0992& 0.0997& 0.0& 1.427& 0.690  \\
1.80& 0.0157& 0.0879& 0.0893& 7.3& 0.318& 1.431& 0.0070& 0.0981& 0.0984& 0.0& 1.428& 0.689  \\
2.00& 0.0112& 0.0867& 0.0875& 6.8& 0.318& 1.432& 0.0050& 0.0977& 0.0978& 0.0& 1.428& 0.689  \\
\bottomrule
\bottomrule
\end{tabular}
}
\end{center}
\end{table}

\begin{table}[H]
\caption{LS and debiased estimates and 95\% CIs for $\beta$}
\label{tab: empirical res}
\begin{small}
\begin{center}
\begin{tabular}{l c c c c c c c c c}
& $R=1$& $R=2$& $R=3$& $R=4$& $R=5$& $R=6$ \\ 
\midrule
\multicolumn{7}{c}{LS}\\ 
\midrule
 & $0.047$& $0.160$& $0.101$& $0.043$& $0.028$& $0.091$ \\ 
& $[-0.06,0.15]$& $[0.04,0.28]$& $[-0.02,0.22]$& $[-0.07,0.16]$& $[-0.10,0.16]$& $[-0.04,0.22]$ \\ 
\midrule
\multicolumn{7}{c}{Debiased}\\ 
\midrule
  & $0.089$& $0.162$& $0.130$& $0.084$& $0.071$& $0.106$ \\ 
$R_w = 0$& $[-0.01,0.19]$& $[0.07,0.26]$& $[0.05,0.21]$& $[0.01,0.16]$& $[-0.01,0.15]$& $[0.04,0.18]$ \\ 
$R_w = 1$& $[-0.77,0.95]$& $[-0.56,0.88]$& $[-0.45,0.71]$& $[-0.40,0.57]$& $[-0.34,0.48]$& $[-0.24,0.45]$ \\ 
$R_w = R$& $[-0.77,0.95]$& $[-1.18,1.50]$& $[-1.43,1.69]$& $[-1.62,1.79]$& $[-1.67,1.81]$& $[-1.61,1.82]$ \\ 
\bottomrule
\bottomrule
\end{tabular}
\end{center}
\end{small}
\end{table}

\setlength{\bibsep}{2pt}

\newpage
\appendix
\setcounter{page}{1}
\renewcommand{\thepage}{A.\arabic{page}}

\begin{LARGE}
    \begin{center}
        Supplementary Appendix to ``Robust Estimation and Inference in Panels with Interactive Fixed Effects''
    \end{center}
\end{LARGE}

\begin{large}
    \begin{center}
        Timothy B. Armstrong, Martin Weidner and Andrei Zeleneev
    \end{center}
\end{large}

\begin{large}
    \begin{center}
         May 2025
    \end{center}
\end{large}

\section{Proofs}
\label{sec: proofs}

This section contains proofs of the results in the main text.
Section \ref{general_rate_thm_sec} states and proves a general result on rates
of convergence using high level conditions on the covariates $X$ and $Z$ and the bound $\hat C$ on $\Vert \Gamma - \hat \Gamma \Vert_*$. %
The proofs of Theorems \ref{new_high_level_validity_thm}-\ref{new_validity_theorem} are provided in Sections \ref{validity_proof_sec}-\ref{ssec: proof of new validity theorem}, respectively.

\bigskip

\noindent{\bf Notation} In this section, mirroring the notation used in Assumption \ref{ass: factor model}\eqref{item: NC},  we write $W$ for $X,Z_1,\ldots, Z_K$ and $W\cdot \gamma=X\beta + \sum_{k=1}^KZ_k\delta_k$ for a generic $\gamma=(\beta,\delta')'$. In particular, we also use $\hat \gamma = (\hat \beta, \hat \delta')'$, $\hat \gamma_{\rm LS} = (\hat \beta_{\rm LS}, \hat \delta_{\rm LS}')'$, and $\hat \gamma_{\rm pre} = (\hat \beta_{\rm pre}, \hat \delta_{\rm pre}')'$.

\subsection{General result for rates of convergence}\label{general_rate_thm_sec}

We first prove a result giving rates of convergence for estimators
$\hat\beta=\langle A^*_{b,c}, \widetilde Y \rangle_F$ given in Definition
\ref{DefAugLinEst} with weights $A^*_{b,c}$ given in
(\ref{weight_optimization_lindeberg})
under a high level condition on the bound $\hat C$ on the initial estimation error in (\ref{NNcondition}). Lemma \ref{preliminary_LS_rates_lemma}\eqref{item: Gamma LS rate} and Theorem \ref{new_nuclear_bound_error_thm}\eqref{item: 3 R bound} verify this condition for $\hat \Gamma = \hat \Gamma_{\rm LS}$ and $\hat \Gamma = \hat \Gamma_{\rm pre}$ (and provide appropriate bounds $\hat C$), respectively.

We make the following assumption on the class of distributions of $X$,
$Z_1,\ldots,Z_k$ and $U$ and the sequence $c=c_{N,T}$ used in the Lindeberg constraint. After we prove the main result of this section, we will also verify this assumption under a set of primitive conditions.

\begin{assumption}\label{rate_assumption}
  \mbox{}
  There exists a sequence of $N\times T$ random matrices $\Xi$ such that
  \begin{align*}
    &\|\Xi\|_F = \mathcal{O}_{\Theta,\mathcal{P}}(\sqrt{NT}),
    \quad
    \left| \langle \Xi,X \rangle_F \right|^{-1} = \mathcal{O}_{\Theta,\mathcal{P}}((NT)^{-1}),  \\
    &s_1(\Xi) = \mathcal{O}_{\Theta,\mathcal{P}}(\max\{\sqrt{N},\sqrt{T}\}),
  \end{align*}
  and, with probability approaching one,
  \begin{align*}
    \operatorname{Lind}(\Xi) \le c_{N,T}
    \quad\text{and}\quad \langle \Xi, Z_k \rangle_F=0 \text{ for }k=1,\ldots,K.
  \end{align*}
\end{assumption}

\begin{theorem}\label{high_level_rate_thm}
  Let $\hat\beta=\langle A^*_{\biasweight,c}, \widetilde Y \rangle_F$ for $\widetilde Y = \Gamma - \hat \Gamma$ and some sequences $c=c_{N,T}$ and
  $\biasweight=\biasweight_{N,T}$.
  Suppose Assumption \ref{rate_assumption} and Assumption~\ref{ass: factor model}(\ref{A_U_inner_product_bound}) hold and that Assumption
  \ref{ass: augmented high level}\eqref{item: augmented high level C hat bound}
  holds
  with $\widetilde \Gamma = \Gamma - \hat \Gamma$ and $\hat C=\mathcal{O}_{\Theta,\mathcal{P}}(\overline C_{N,T})$ for some sequence
  $\overline C_{N,T}$.  Then
  \begin{align*}
    |\hat\beta-\beta|
    = \mathcal{O}_{\Theta,\mathcal{P}}\left(
    \max\left\{ \overline C_{N,T}/\biasweight_{N,T}, 1  \right\}
    \cdot \max\left\{ (NT)^{-1/2}, \biasweight_{N,T}\cdot \max\{\sqrt{N},\sqrt{T}\}/(NT) \right\}
    \right).
  \end{align*}
\end{theorem}
\begin{proof}
We have
\begin{align*}
  |\hat\beta - \beta|
  \le \left| \langle A^*_{\biasweight,c}, U \rangle_F \right|
  + \maxbias_{\tilde C}(A^*_{\biasweight,c})
  = \left| \langle A^*_{\biasweight,c}, U \rangle_F \right|
  + \tilde C s_1(A^*_{\biasweight,c})
\end{align*}
where $\tilde C=\|\Gamma-\hat\Gamma\|_*$. %
Thus,
\begin{align}\label{beta_error_square_bound_eq}
  &|\hat\beta - \beta|^2
  \le 2 \left| \langle A^*_{\biasweight,c}, U \rangle_F \right|^2
  + 2 \tilde C^2 s_1(A^*_{\biasweight,c})^2  \nonumber  \\
  &\le 2 \max\left\{ \frac{\left| \langle A^*_{\biasweight,c}, U \rangle_F \right|^2}{\|A^*_{\biasweight,c}\|_F^2},
  \frac{\tilde C^2}{\biasweight^2} \right\}
  \cdot \left[ \|A^*_{\biasweight,c}\|_F^2 + \biasweight^2 s_1(A^*_{\biasweight,c})^2 \right].
\end{align}
Consider the oracle weights $\tilde A = \Xi/\langle \Xi, X \rangle_F$.
With probability approaching one uniformly over $\theta, P$,
the weights
$\tilde A$
are feasible for (\ref{weight_optimization_lindeberg}), so that
\begin{align}
  &\|A^*_{\biasweight,c}\|_F^2 + \biasweight^2 s_1(A^*_{\biasweight,c})^2
    \le \|\tilde A\|_F^2 + \biasweight^2 s_1(\tilde A)^2
    = \frac{\|\Xi\|_F^2 + \biasweight^2 s_1(\Xi)^2}{\left| \langle \Xi, X \rangle_F \right|^2} \notag \\
  &= \mathcal{O}_{\Theta,\mathcal{P}}((NT)^{-1})
    + \biasweight^2 \cdot \mathcal{O}_{\Theta,\mathcal{P}}(\max\{N,T\}/(NT)^2). \label{eq: A mse bound}
\end{align}
Plugging this into (\ref{beta_error_square_bound_eq}) gives the result.
\end{proof}

Next, we verify that Assumption \ref{rate_assumption} holds under low level conditions provided in Assumption \ref{ass: X decomposition}.

\begin{lemma}\label{lem: Xi verification}
  Suppose that Assumption \ref{ass: X
    decomposition} holds.  Then Assumption
  \ref{rate_assumption} holds
  with
  $\Xi_{it}$ given by the residual in the regression of $V_{it}$ on $Z_{it}$, i.e., ${\rm vec} (\Xi) = M_{\bf Z} {\rm vec} (V)$ where $M_{\bf Z} = I_{NT} - \bf Z (\bf Z' \bf Z)^{-1} \bf Z'$.
\end{lemma}
\begin{proof}
  First, notice that $\Xi = V - Z \cdot \hat \varphi = V - \sum_{k=1}^K Z_k \hat
  \varphi_k$ where $\hat \varphi = {{\bf( Z' Z)}^{-1} {\bf Z}'} {\rm vec}(V)$.
  Also, it follows from
  Assumption \ref{ass: X decomposition}(\ref{Z_V_inner_product_assump}) and
  (\ref{Z_prime_Z_assump})
  that $\norm{\hat \varphi} = \Op \left(\frac{1}{\sqrt{NT}}\right)$.

  Next we verify all the conditions required by Assumption \ref{rate_assumption}.

  \medskip

  {\noindent \it Verification of $\langle \Xi, Z_k \rangle_F = 0$ for $k = 1, \ldots, K$.} By construction.

  \medskip

  {\noindent \it Verification of $\norm{\Xi}_F = \Op (\sqrt{NT})$.} $\norm{\Xi}_F = \norm{{\rm vec}(\Xi)} \leq \norm{{\rm vec}(V)} = \norm{V}_F = \Op (\sqrt{NT})$.

  \medskip

  {\noindent \it Verification of $s_1(\Xi) = \Op(\maxSNT)$.} Notice that
  \begin{align*}
    s_1(\Xi) = s_1 \left(V - \sum_{k=1}^K Z_k \hat \varphi_k\right) \leq s_1(V) + \sum_{k=1}^K \abs{\hat \varphi_k} s_1 (Z_k) = \Op (\maxSNT),
  \end{align*}
  using the fact that
  $\abs{\hat \varphi_k} s_1(Z_k) = \Op (1)$
  since
  $\hat \varphi_k=\Op(1/\sqrt{NT})$ and $s_1(Z_k)\le \|Z_k\|_F=\Op(\sqrt{NT})$.

  \medskip

  {\noindent \it Verification of $\abs{\langle \Xi, X \rangle_F}^{-1} = \Op
    ((NT)^{-1})$.}
  Using the fact that $\langle \Xi, Z_k \rangle_F=0$ for each $k$, we have
  \begin{align*}
    {\langle \Xi, X \rangle}_F = {\langle \Xi, H \rangle}_F + {\langle \Xi, V \rangle}_F,
    = \langle V, H \rangle_F
    - \sum_{k=1}^K\hat\varphi_k \langle Z_k,H \rangle_F
    + \|V\|_F^2
    - \sum_{k=1}^K\hat\varphi_k \langle Z_k,V \rangle_F.
  \end{align*}
  The first term is $\Op(\sqrt{NT})=\op(NT)$ by Assumption \ref{ass: X
    decomposition}(\ref{H_assump}).
  The second term is $\Op(\sqrt{NT})=\op(NT)$ since $\hat \varphi_k=\Op(1/\sqrt{NT})$ and $\langle
  Z_k, H \rangle_F \le \|H\|_F\cdot \|Z_k\|_F=\Op(NT)$
  under these assumptions.
  Similarly, the fourth term is $\Op(1)=\op(NT)$.
  Thus, ${\langle \Xi, X \rangle}_F=\|V\|_F^2 + \op(NT)$ and the result follows
  since $\|V\|_F^2 \asy NT$ by Assumption \ref{ass: X decomposition}(\ref{V_assump}).

  \medskip

  {\noindent \it Verification of ${\rm Lind} (\Xi) \leq c_{N,T}$ with probability approaching one.}

  \begin{align*}
    {\rm Lind} (\Xi) = \frac{\max_{i,t} \Xi_{it}^2}{\norm{\Xi}_F^2},
  \end{align*}
  where
  \begin{align*}
    \norm{\Xi}_F^2 = \norm{V}_F^2 - 2 \sum_{k=1}^K \hat \varphi_k \langle V, Z_k \rangle_F + \norm{\sum_{k=1}^K Z_k \hat \varphi_k}_F^2,
  \end{align*}
  where $\sum_{k=1}^K \hat \varphi_k \langle V, Z_k \rangle_F = \Op (1)$ and $\norm{\sum_{k=1}^K Z_k \hat \varphi_k}_F^2 = \Op (1)$, so $\norm{\Xi}_F^2 \asy NT$. Next,
  \begin{align*}
    \max_{i,t} \Xi_{it}^2 &= \max_{i,t} \left(V_{it} - \sum_{k=1}^K \hat \varphi_k Z_{k,it} \right)^2 \\ &\leq (K+1)^2 \left(\max_{i,t} V_{it}^2 + \sum_{k=1}^K \hat \varphi_k^2 \max_{i,t} Z_{k,it}^2\right) = \op (NT c_{N,T}).
  \end{align*}
  Hence, ${\rm Lind} (\Xi) = \op (c_{N,T})$, which completes the proof.
\end{proof}

\subsection{Proof of Theorem \ref{new_high_level_validity_thm}}\label{validity_proof_sec}

  The probability that the upper endpoint of the CI is less than $\beta$ is
  \begin{align*}
    &\pr_{\theta,P}\left( \hat\beta + \maxbias_{\hat C}(\hat\beta) + z_{1-\alpha/2}\widehat{\se} < \beta \right)  \\
    &= \pr_{\theta,P}\left( \langle A, X\beta + Z\cdot \delta + \widetilde \Gamma \rangle_F - \beta + \maxbias_{\hat C}(\hat\beta) + \langle A, \widetilde U \rangle_F < -z_{1-\alpha/2}\widehat{\se} \right)  \\
    &\le \pr_{\theta,P}\left( \langle A, X\beta + Z\cdot \delta + \widetilde \Gamma \rangle_F - \beta < - \maxbias_{\hat C}(\hat\beta)  \right) + \pr_{\theta,P}\left( \langle A, \widetilde U \rangle_F < -z_{1-\alpha/2}\widehat{\se} \right).
  \end{align*}
  The first term is, by definition, bounded by $\pr_{\theta,P}(\|\widetilde \Gamma \|_* >\hat C)$, which converges to zero uniformly
  over $\theta\in\Theta, P\in\mathcal{P}$ by Assumption \ref{ass: augmented high level}\eqref{item: augmented high level C hat bound}.
  The second term converges to $\alpha/2$ uniformly over $\theta\in\Theta,
  P\in\mathcal{P}$ by Assumption
  \ref{ass: augmented high level}\eqref{item: augmented high level CLT}.
  Applying a symmetric argument to the probability that the lower endpoint of
  the CI is greater than $\beta$ gives the result.

\subsection{Proof of Theorem \ref{new_nuclear_bound_error_thm}}

To prove Theorem \ref{new_nuclear_bound_error_thm}, we first state and prove a series of auxiliary lemmas in Section \ref{sssec: nuclear norm bound aux lemmas} below. In Section \ref{sssec: proof of nunclear bound thm}, we then prove the final results.

\subsubsection{\label{sssec: nuclear norm bound aux lemmas}Auxiliary Lemmas}
\begin{lemma}\label{generic_Gamma_lemma}
  Consider
  \begin{align*}
    \widehat \Gamma  
          = \argmin_{\left\{ G \, : \, {\rm rank}(G) \leq  R \right\}} \norm{Y - W \cdot \hat \gamma -G}_F^2
  \end{align*}  
  for some $\hat \gamma$. Suppose that $\Gamma = \lambda f'$ for some $N \times R$ matrix $\lambda$ and $T \times R$ matrix $f$. Then, we have
  \begin{enumerate}[(i)]
    \item $\|\widehat \Gamma - \Gamma\|_* \leq 3 R s_1 (\hat U)$;

    \item $\|\widehat \Gamma - \Gamma - P_{\lambda} \hat U\|_* \leq 2 R s_1 (\hat U)$;
  \end{enumerate}
  where $\hat U \coloneqq U - W \cdot (\hat \gamma - \gamma)$.
\end{lemma}

\begin{proof}%
  Let 
  \begin{align*}
      Q(G) \coloneqq \left\| Y - W \cdot \hat \gamma  - G \right\|_F^2,
  \end{align*}
  $\hat Y \coloneqq Y - W \cdot \hat \gamma = \Gamma + \hat U$, and $\Gamma^\dagger \coloneqq \Gamma + P_{\lambda} \hat U$, where $M_{\lambda} = \mathbb I_N - P_{\lambda}$ and $\mathbb I_N$ stands for a $N \times N$ identity matrix. Notice that we have $M_{\lambda} \Gamma^\dagger = 0$, and therefore ${\rm rank}\left( \Gamma^{\dagger} \right) \leq R$.

  Next, note that
  \begin{align*}
    Q (G) = \norm{\hat Y - G}_F^2 = \norm{\Gamma + \hat U - G}_F^2 = \norm{P_{\lambda} (\Gamma^\dagger - G) + M_{\lambda} (\hat U - G)}_F^2,
  \end{align*}
  so we also have
  \begin{align*}
        Q( \widehat \Gamma)      &=  
                 \left\| P_\lambda \left(\widehat \Gamma -   \Gamma^\dagger    \right)   \right\|_F ^2
                   + \left\| M_\lambda \left(  \widehat \Gamma  - \hat U \right)   \right\|_F^2   
          \\
          &=   \left\|         \widehat \Gamma -  \Gamma^\dagger \right\|_F^2   
             -  \left\| M_\lambda  \widehat \Gamma     \right\|_F^2   
          + \left\| M_\lambda \left(  \widehat \Gamma  - \hat U \right)    \right\|_F^2   
          \\
          &=   \left\|         \widehat \Gamma -  \Gamma^\dagger \right\|_F^2      
          + \left\| M_\lambda  \hat U    \right\|_F^2   
          - 2 {\rm Tr}\left(\hat U' M_{\lambda} \widehat \Gamma    \right)
    \\
         &\geq   \left\|         \widehat \Gamma -  \Gamma^\dagger \right\|_F^2   
          + \left\| M_\lambda   \hat U    \right\|_F^2       -    2  \, s_1(\hat U) \,  \left\| M_\lambda  \widehat \Gamma       \right\|_*.
  \end{align*}
  Combining this with
  \begin{align*}
     Q( \widehat \Gamma)  \leq   Q(  \Gamma^\dagger) &=  \left\| M_\lambda \hat U     \right\|_F^2   ,
  \end{align*}
  we obtain that
  \begin{align*}
       \left\|         \widehat \Gamma -  \Gamma^\dagger \right\|_*^2 &\leq 2R \left\|         \widehat \Gamma -  \Gamma^\dagger \right\|_F^2 ,
    \\
      &\leq 4 R  \, s_1(\hat U) \,  \left\| M_\lambda  \widehat \Gamma       \right\|_*
    \\
        &\leq 4 R^2  \, s_1(\hat U) \, s_1( M_\lambda  \widehat \Gamma     ) ,
  \end{align*}
  and therefore
  \begin{align}
       \left\|         \widehat \Gamma -  \Gamma^\dagger \right\|_* \leq 2 \, R \, \sqrt{s_1(\hat U) \, s_1( M_\lambda  \widehat \Gamma     )}. \label{eq: 2R bound preliminar}
  \end{align}
  Next, since $\widehat \Gamma $ is given by the $R$ leading principal components of $\hat Y$, we know that
  $\hat Y  \hat Y'  \geq (\widehat \Gamma)(\widehat \Gamma)'$, that is the difference $\hat Y  \hat Y' -  \widehat \Gamma \widehat \Gamma'$
  is positive-definitive, which implies that
  \begin{align*}
       [s_1( M_\lambda  \widehat \Gamma   )]^2 &=  \mu_1\left(M_\lambda  \widehat \Gamma  \widehat \Gamma' M_\lambda \right)
      = \max_{\left\{ v \in \mathbb{R}^N \, : \, \|v\|=1 \right\}} v' M_\lambda  \widehat \Gamma  \widehat \Gamma' M_\lambda v
       = \hat v' M_\lambda  \widehat \Gamma  \widehat \Gamma' M_\lambda \hat v
    \\
       &\leq \hat v' M_\lambda  \hat Y \hat Y' M_\lambda \hat v
       \leq  \max_{\left\{ v \in \mathbb{R}^N \, : \, \|v\|=1 \right\}} v' M_\lambda  \hat Y \hat Y' M_\lambda v
       = \mu_1(M_\lambda  \hat Y \hat Y' M_\lambda) =   [s_1( M_\lambda  \hat Y )]^2 
     \\
       &=   [s_1( M_\lambda  \hat U )]^2 \leq  [s_1(   \hat U )]^2.
  \end{align*}
  We have thus shown that $s_1( M_\lambda  \widehat \Gamma   ) \leq s_1(   U )$, and, combining this with \eqref{eq: 2R bound preliminar}, we obtain
  \begin{align*}
       \left\|         \widehat \Gamma -  \Gamma^\dagger \right\|_* \leq 2 \, R \,  s_1(\hat U),
  \end{align*}
  which proofs the second statement of the lemma. To prove the first, notice that
  \begin{align*}
    \norm{\widehat{\Gamma} - \Gamma}_* \leq \left\|         \widehat \Gamma -  \Gamma^\dagger \right\|_* + \norm{P_\lambda \hat U}_* \leq 3 R s_1 (\hat U),
  \end{align*}
  where the last inequality uses $\norm{P_\lambda \hat U}_* \leq R s_1 (\hat U)$.
\end{proof}

\begin{lemma}\label{preliminary_LS_rates_lemma}
  Under Assumptions \ref{ass: factor model}\eqref{item: NC}-\eqref{item: EX},
  \begin{enumerate}[(i)]
    \item \label{item: gamma LS rate} $\hat \gamma_{\rm LS} - \gamma = \Op (1/\minSNT)$;

    \item \label{item: Gamma LS rate} $\norm{\hat \Gamma_{\rm LS} - \Gamma}_* \leq \hat C$ for some $\hat C = \Op (\maxSNT)$.
  \end{enumerate}
\end{lemma}

\begin{proof}
  The first statement follows from the proof of Theorem 4.1 in \citet{MoonWeidner2015}.
  Assumptions \ref{ass: factor model} \eqref{item: NC}-\eqref{item: EX} are
  uniform analogues of Assumptions NC, SN, and EX in \citet{MoonWeidner2015}. The
  derived rate of convergence is immediately uniform over $\theta\in\Theta, P \in \mathcal P$ because the proof of Theorem 4.1 in \citet{MoonWeidner2015} explicitly bounds $\norm{\hat \gamma_{\rm LS} - \gamma_0}$.

  Next, we combine this result with Lemma \ref{generic_Gamma_lemma}\eqref{item: 3 R bound} to obtain
  \begin{align}
    \norm{\hat \Gamma_{\rm LS} - \Gamma}_* &\leq 3 R s_1 (U - W \cdot (\hat \gamma_{\rm LS} - \gamma) ) \notag \\
    &\leq 3 R \left(s_1 (U) + s_1 (W \cdot (\hat \gamma_{\rm LS} - \gamma) ) \right) \notag \\
    &= \Op (\maxSNT), \label{eq: delta Gamma LS explicit bound}
  \end{align}
  where the last equality follows from Assumption \ref{ass: factor model}\eqref{item: SN} and $\hat \gamma_{\rm LS} - \gamma = \Op (1/\minSNT)$.
\end{proof}

\begin{lemma}\label{gamma_pre_rate_lemma}
  Suppose that Assumption \ref{ass: factor model} holds, and that Assumption
  \ref{rate_assumption} holds as stated and with $Z_k$ and $X$ interchanged for
  each $k=1,\ldots,K$.
  Then
  \begin{align*}
    \hat\gamma_{\rm pre} - \gamma
    = \mathcal{O}_{\Theta,\mathcal{P}}\left(
    1/\min\{N,T\}
    \right).
  \end{align*}
\end{lemma}
\begin{proof}
  The result is immediate from Lemma \ref{preliminary_LS_rates_lemma}\eqref{item: Gamma LS rate} and Theorem
  \ref{high_level_rate_thm}, using the fact that $\bTW$ is bounded from above
  and below by a constant times $\maxSNT$.
\end{proof}

\begin{lemma}\label{s1_hat_U_pre_lemma}
  Suppose that Assumption \ref{ass: factor model} holds, and that Assumption
  \ref{rate_assumption} holds as stated and with $Z_k$ and $X$ interchanged for
  each $k=1,\ldots,K$. Then, $s_1 (\hat U_{\rm pre}) \asymp_{\Theta, \mathcal P} \maxSNT$ and
  \begin{align*}
    s_1(U)\le s_1(\hat U_{\rm pre})(1+o_{\Theta,\mathcal{P}}(1)).
  \end{align*}
\end{lemma}
\begin{proof}
  First note that, letting
  $\Delta_\Gamma=\hat\Gamma_{\rm pre} - \Gamma$, we have
  \begin{align}\label{s1_U_Delta_bound_eq}
    \left| s_1(\hat U_{\rm pre})-s_1(U - \Delta_\Gamma) \right|
    \le s_1(W \cdot (\hat \gamma_{\rm pre} - \gamma)) = o_{\Theta,P}(\maxSNT)
  \end{align}
    the equality follows from Assumption \ref{ass:
    factor model}(\ref{item: SN}) and Lemma \ref{gamma_pre_rate_lemma}. Also, notice that
    \begin{align*}
      s_1 (U - \Delta_\Gamma) \leq s_1 (U) + s_1 (\Delta_\Gamma) = \Op (\maxSNT),
    \end{align*}
    where the equality follows from Assumption \ref{ass:
    factor model}(\ref{item: SN}) and $s_1 (\Delta_\Gamma) \leq \norm{\Delta_\Gamma}_* = \Op (\maxSNT)$, which can be verified analogously to \eqref{eq: delta Gamma LS explicit bound} plugging $\hat \gamma_{\rm pre}$ instead of $\hat \gamma_{\rm LS}$ and using the result of Lemma \ref{gamma_pre_rate_lemma}. Combining this with \eqref{s1_U_Delta_bound_eq}, we conclude
    \begin{align}
      \label{eq: s_1 hat U bound}
      s_1 (\hat U_{\rm pre}) = \Op (\maxSNT).
    \end{align}
    
    Next, using the fact that $\rank{\Delta_\Gamma}\le 2R$ and
    the general singular value inequality $s_{i+j-1}(A+B) \leq s_i(A)+s_j(B)$
    (see equation (7.3.13) in Problem 7.3.P16
 of \citealt{horn2013matrix}, or alternatively \citealt{fan1951maximum})
    with $A= U -
    \Delta_\Gamma$, $B= \Delta_\Gamma$, $i=1$, $j=2R+1$ gives $s_{2R+1}(U) \leq s_{1}(U - \Delta_\Gamma)$.  Thus,
    \begin{align}
      s_1(U)\le s_{2R+1}(U)+o_{\Theta,P}(\maxSNT)\le s_1(\hat U_{\rm pre})+o_{\Theta,P}(\maxSNT),
      \label{eq: s1 U and U hat}
    \end{align}
    where we apply Assumption \ref{ass: factor model}\eqref{item: high level
      singular values of U} for the first inequality and
    (\ref{s1_U_Delta_bound_eq}) for the second inequality. Notice that this, together with Assumption \ref{ass: factor model}\eqref{item: SN} and \eqref{eq: s_1 hat U bound}, implies $s_1 (\hat U_{\rm pre}) \asymp_{\Theta, \mathcal P} \maxSNT$. Together with \eqref{eq: s1 U and U hat}, this completes the proof.
\end{proof}

\subsubsection{\label{sssec: proof of nunclear bound thm}Proof of Theorem \ref{new_nuclear_bound_error_thm}}
First, notice that Assumption \ref{rate_assumption} holds according to Lemma \ref{lem: Xi verification}, so we can apply Lemmas \ref{gamma_pre_rate_lemma} and \ref{s1_hat_U_pre_lemma}.

We prove the second statement of the theorem. The proof of the first statement is analogous.

Applying the second result of Lemma \ref{generic_Gamma_lemma} with $\hat \gamma = \gamma_{\rm pre}$ and $\hat U = U - W \cdot (\hat \gamma_{\rm pre} - \gamma)$, we obtain 
\begin{align*}
  2 R s_1 (\hat U) \geq \norm{\hat \Gamma_{\rm pre} - \Gamma - P_\lambda \hat U}_* \geq \norm{\hat \Gamma_{\rm pre} - \Gamma - P_\lambda U }_* - \norm{P_{\lambda} W \cdot (\hat \gamma_{\rm pre} - \gamma) }_*.
\end{align*}
Note that
\begin{align*}
  \norm{P_{\lambda} W \cdot (\hat \gamma_{\rm pre} - \gamma) }_* \leq R s_1  \left(W \cdot (\hat \gamma_{\rm pre} - \gamma)\right) = \op (\maxSNT),
\end{align*}
where the equality follows from Assumption \ref{ass: factor model}(\ref{item: SN}) and Lemma \ref{gamma_pre_rate_lemma}. Similarly,
\begin{align*}
  s_1 (\hat U) \leq s_1 (U) + s_1 (W \cdot (\hat \gamma_{\rm pre} - \gamma)) = s_1(U) + \op (\maxSNT).
\end{align*}
Hence,
\begin{align*}
  \norm{\hat \Gamma_{\rm pre} - \Gamma - P_\lambda U }_* \leq 2 R s_1 (U) + \op (\maxSNT).
\end{align*}
Combining this with Lemma \ref{s1_hat_U_pre_lemma} proves the second statement of the theorem.

\subsection{Proof of Theorem \ref{new_beta_rates_thm}}
The result follows from the first result of Theorem \ref{new_nuclear_bound_error_thm} and Theorem
\ref{high_level_rate_thm},
along with Lemma \ref{lem: Xi verification} verifying Assumption \ref{rate_assumption}.

\subsection{Proof of Theorem \ref{new_validity_theorem}}
\label{ssec: proof of new validity theorem}

The first statement of the theorem follows from Theorem \ref{new_beta_rates_thm} once we verify Assumption \ref{ass: factor model}(\ref{A_U_inner_product_bound}). Notice that Assumption \ref{ass: factor model}(\ref{A_U_inner_product_bound})
is immediate from Assumption \ref{U_conditional_moment_bound_assump} and
Chebyshev's inequality. Similarly, later in the proof, we will also invoke some
of the previously derived results which rely on Assumption \ref{ass: factor
  model}(\ref{A_U_inner_product_bound}).

To prove the second statement of the theorem, we verify Assumption \ref{ass:
  augmented high level} with $\tilde\Gamma=\Gamma+P_\lambda U - \hat\Gamma_{\rm
  pre}$ and $\tilde U=U-P_{\lambda}U$.  Part (\ref{item: augmented high level C
  hat bound}) of Assumption \ref{ass:
  augmented high level} holds by construction, so we just need to verify part
(\ref{item: augmented high level CLT}) with this choice of $\tilde U$.  This
will follow if we can show
\begin{align}
  \label{eq: AU CLT}
  \langle A, U \rangle_F/\widehat{\se}
  \underset{\Theta,\mathcal{P}}{\overset{d}{\to}}
  N(0,1)
\end{align}
and
\begin{align}
  \label{eq: APlambda}
  \langle A, P_\lambda U
  \rangle_F/\widehat{\se}=o_{\Theta,\mathcal{P}}(1)
\end{align}
where $A=A^*_{b,c}$.  Section \ref{sssec: aux lemmas for new validity theorem}
verifies (\ref{eq: AU CLT}) and Section \ref{sssec: proof of new validity
  theorem} verifies (\ref{eq: APlambda}).

\subsubsection{Verification of (\ref{eq: AU CLT})}
\label{sssec: aux lemmas for new validity theorem}

In this section, we verify that (\ref{eq: AU CLT}) holds under the hypotheses of
Theorem \ref{new_validity_theorem}.
Specifically, we show that
$\langle A, U \rangle_F/\widehat{\se}
\underset{\Theta,\mathcal{P}}{\overset{d}{\to}}
N(0,1)$
for
$\widehat{\se}^2=\sum_{i=1}^N\sum_{t=1}^T A_{it}^2\hat U_{it}^2$
with any sequence of matrices $A$
satisfying $\operatorname{Lind}(A)\le c_{N,T}$ with
$c_{N,T}$ satisfying the condition $c_{N,T} \max\{N,T\}\to 0$ given in the statement of the theorem.

To this end, we first prove a bound on $\|\hat U-U\|_F$ (Lemma \ref{lem: hat U
  Frobenius rate}), and then use this to show consistency of the standard error
(Lemma \ref{lem: se consistency}, using a condition verified in
Lemma \ref{lem: omega sum verification}). Lemma \ref{lem: u-normality
  verification} completes the proof.
We note that the conditions of Lemma \ref{lem: hat U Frobenius rate} hold under
the conditions of Theorem \ref{new_validity_theorem} by
Lemma \ref{preliminary_LS_rates_lemma}.

\begin{lemma}
  \label{lem: hat U Frobenius rate}
  Let $\hat U = Y - W \cdot \hat \gamma - \hat \Gamma$, where
  \begin{align*}
    \hat \Gamma = \argmin_{\left\{ G \in \mathbb{R}^{N \times T} \, : \, {\rm rank}(G) \leq R \right\}}
    \sum_{i=1}^N \sum_{t=1}^T \left( Y_{it} - W_{it}' \hat \gamma - G_{it} \right)^2.
  \end{align*}
  Suppose that
  \begin{enumerate}[(i)]
    \item \label{item: gam rate} $\hat \gamma - \gamma = \Op \left(\frac{1}{\minSNT}\right)$;
    \item \label{item: XZ Frobenius} $\norm{X}_F = \Op (\sqrt{NT})$ and $\norm{Z_k}_F = \Op (\sqrt{NT})$ for $k \in \{1, \ldots, K\}$;
    \item \label{item: SN aux} $s_1(X) = \mathcal{O}_{\Theta,\mathcal{P}} \left(\sqrt{NT}\right)$, $s_1(Z_k) = \mathcal{O}_{\Theta,\mathcal{P}} \left(\sqrt{NT}\right)$ for $k \in \{1, \ldots, K\}$, and $s_1(U) = \Op (\maxSNT)$.
  \end{enumerate}
  Then,
  \begin{align*}
    \norm{\hat U - U}_F^2 = \Op (\max \{N,T\}).
  \end{align*}
\end{lemma}
\begin{proof}%
  Using $\hat U = W \cdot (\gamma - \hat \gamma) + \Gamma - \hat \Gamma + U$,
  \begin{align*}
    \norm{\hat U - U}_F^2 = \norm{W \cdot (\hat \gamma - \gamma)}_F^2 + \norm{\hat \Gamma - \Gamma}_F^2 + 2 \langle W \cdot (\hat \gamma - \gamma), \hat \Gamma - \Gamma \rangle_F.
  \end{align*}
  To prove the result, we show that all the terms on the right hand side of the equation above are $\Op (\max \{N,T\})$.

  First,
  \begin{align*}
    \norm{W \cdot (\hat \gamma - \gamma)}_F \leq \norm{X}_F \abs{\hat \beta - \beta}  + \sum_{k=1}^K \norm{{Z_k}}_F \abs{\hat \delta_{k} - \delta_{k}} = \Op \left(\maxSNT\right).
  \end{align*}
  where we used conditions \eqref{item: gam rate} and \eqref{item: XZ Frobenius}.

  Second,
  \begin{align*}
    \norm{\hat \Gamma - \Gamma}_F \leq \norm{\hat \Gamma - \Gamma}_*  =  \Op (\maxSNT)
  \end{align*}
  where the equality holds analogously to the previously derived result \eqref{eq: delta Gamma LS explicit bound} with $\hat \Gamma$ and $\hat \gamma$ replacing $\hat \Gamma_{\rm LS}$ and $\hat \gamma_{\rm LS}$ correspondingly.

  Third,
  \begin{align*}
    \abs{\langle W \cdot (\hat \gamma - \gamma), \hat \Gamma - \Gamma \rangle_F} \leq \abs{\langle X (\hat \beta - \beta), \hat \Gamma - \Gamma \rangle_F} + \sum_{k=1}^K \abs{\langle Z_k (\hat \delta_k - \delta_k), \hat \Gamma - \Gamma \rangle_F},
  \end{align*}
  where
  \begin{align*}
    \abs{\langle X (\hat \beta - \beta), \hat \Gamma - \Gamma \rangle_F} \leq \norm{X}_F \norm{\hat \Gamma - \Gamma}_F \abs{\hat \beta - \beta} = \Op \left(\max \{N,T\}\right).
  \end{align*}
  Similarly,
  \begin{align*}
    \sum_{k=1}^K \abs{\langle Z_k (\hat \delta_k - \delta_k), \hat \Gamma - \Gamma \rangle_F} = \Op \left(\max \{N,T\}\right),
  \end{align*}
  which implies
  \begin{align*}
    \abs{\langle W \cdot (\hat \gamma - \gamma), \hat \Gamma - \Gamma \rangle_F} = \Op (\max \{N,T\})
  \end{align*}
  and completes the proof.
\end{proof}

\begin{lemma}
  \label{lem: se consistency}
  Suppose that the hypotheses of Lemma \ref{lem: hat U Frobenius rate} are satisfied. Suppose, in addition, that the following conditions hold:
  \begin{enumerate}[(i)]
    \item \label{item: HL var U rate} for any collections of weights $\{\omega_{it}\}_{1 \leq i \leq N, 1 \leq t \leq T}$, which are non-random conditional on $W$ and $\Gamma$, such that $\abs{\omega_{it}} \leq \overline \omega$ a.s. for all $W$ and $\Gamma$ and for all $i$, $t$, $N$, and $T$, we have 
    \begin{align*}
      \frac{1}{NT} \sum_{i=1}^N \sum_{t=1}^T \omega_{it} U_{it}^2 - \frac{1}{NT} \sum_{i=1}^N \sum_{t=1}^T \omega_{it} \ex{U_{it}^2|W,\Gamma} = \Op \left(\frac{1}{\sqrt{NT}}\right);
    \end{align*}
    \item \label{item: var U lower bound} for some $\underline \sigma^2 > 0$, $\ex{U_{it}^2|W, \Gamma} \geq \underline \sigma^2$ a.s. for all $i$, $t$, $N$, and $T$;
    \item \label{item: Lind rate} ${\rm Lind} (A) \leq c_{N,T}$ and $\max \{N,T\} \;c_{N,T} \rightarrow 0$.
  \end{enumerate}
  Then,
  \begin{align*}
    \frac{\sum_{i=1}^N \sum_{t=1}^T A_{it}^2 \hat U_{it}^2}{\sum_{i=1}^N \sum_{t=1}^T A_{it}^2 U_{it}^2} - 1 = \op (1),
  \end{align*}
  where $\hat U$ is defined in Lemma \ref{lem: hat U Frobenius rate}.
\end{lemma}
\begin{proof}%
  For simplicity of notation, we use $\sum_{i,t} \equiv \sum_{i=1}^N \sum_{t=1}^T$ and $\max_{i,t} \equiv \max_{1 \leq i \leq N, 1 \leq t \leq T}$ throughout the proof.

  Notice that
  \begin{align}
    \frac{\sum_{i,t} A_{it}^2 \hat U_{it}^2}{\sum_{i,t} A_{it}^2 U_{it}^2} - 1 &= \frac{\sum_{i,t} A_{it}^2 \left(\hat U_{it}^2 - U_{it}^2\right)}{\sum_{i,t} A_{it}^2 U_{it}^2} \notag \\
    &= \frac{\sum_{i,t} A_{it}^2 \left(\hat U_{it} - U_{it}\right) \left(\hat U_{it} - U_{it} + 2 U_{it}\right)}{\sum_{i,t} A_{it}^2 U_{it}^2} \notag \\
    &= \frac{\sum_{i,t} A_{it}^2 \left(\hat U_{it} - U_{it}\right)^2}{\sum_{i,t} A_{it}^2 U_{it}^2} + \frac{2 \sum_{i,t} A_{it}^2 U_{it} \left(\hat U_{it} - U_{it}\right)}{\sum_{i,t} A_{it}^2 U_{it}^2} \label{eq: 2 terms to bound}.
  \end{align}
  The first term in \eqref{eq: 2 terms to bound} can be bounded as
  \begin{align*}
    \frac{\sum_{i,t} A_{it}^2 \left(\hat U_{it} - U_{it}\right)^2}{\sum_{i,t} A_{it}^2 U_{it}^2} \leq \frac{\max_{i,t} A_{it}^2 \norm{\hat U - U}_{F}^2}{\sum_{i,t} A_{it}^2 U_{it}^2},
  \end{align*}
  and and the second term in \eqref{eq: 2 terms to bound} can be bounded as
  \begin{align*}
    \frac{\sum_{i,t} A_{it}^2 U_{it} \left(\hat U_{it} - U_{it}\right)}{\sum_{i,t} A_{it}^2 U_{it}^2} &\leq \frac{\left(\sum_{i,t} A_{it}^4 U_{it}^2\right)^{1/2} \left(\sum_{i,t} \left(\hat U_{it} - U_{it}\right)^2\right)^{1/2}}{\sum_{i,t} A_{it}^2 U_{it}^2} \\
    &\leq \sqrt{\frac{\max_{i,t} A_{it}^2 \norm{\hat U - U}_{F}^2}{\sum_{i,t} A_{it}^2 U_{it}^2}},
  \end{align*}
  where the first inequality follows from the Cauchy-Schwarz inequality.

  Hence, to complete the proof, it is sufficient to demonstrate
  \begin{align*}
    \frac{\max_{i,t} A_{it}^2 \norm{\hat U - U}_{F}^2}{\sum_{i,t} A_{it}^2 U_{it}^2} = \op (1).
  \end{align*}

  Next, notice that
  \begin{align*}
    \frac{1}{NT} \sum_{i,t} \frac{A_{it}^2}{\max_{i,t} A_{it}^2} U_{it}^2 &= \frac{1}{NT} \sum_{i,t} \frac{A_{it}^2}{\max_{i,t} A_{it}^2} \ex{U_{it}^2|W,\Gamma} + \Op \left(\frac{1}{\sqrt{NT}}\right)\\
    &\geq \frac{\underline \sigma^2}{NT \; {\rm Lind}(A)} + \Op \left(\frac{1}{\sqrt{NT}}\right)\\
    &\geq \frac{\underline \sigma^2}{NT \; c_{N,T}} + \Op \left(\frac{1}{\sqrt{NT}}\right) >_{\Theta, \mathcal P} 0,
  \end{align*}
  where we used condition \eqref{item: HL var U rate}, \eqref{item: var U lower bound}, and \eqref{item: Lind rate} consequently, and the last inequality (which holds holds wpa1 uniformly) is ensured by condition $\eqref{item: Lind rate}$.

  Then
  \begin{align*}
    \frac{\max_{i,t} A_{it}^2 \norm{\hat U - U}_{F}^2}{\sum_{i,t} A_{it}^2 U_{it}^2} &= \frac{\frac{1}{NT}\norm{\hat U - U}_F^2}{\frac{1}{NT} \sum_{i,t} \frac{A_{it}^2}{\max_{i,t} A_{it}^2} U_{it}^2} \\
    &\leq \frac{c_{N,T} \norm{\hat U - U}_F^2}{\underline \sigma^2 + \Op \left(\sqrt{NT} \; c_{N,T}\right)} \\
    &\leq \frac{c_{N,T} \norm{\hat U - U}_F^2}{\underline \sigma^2 + \op (1)} \\
    &= \op (1),
  \end{align*}
  where the last inequality uses condition \eqref{item: Lind rate}, and the last equality follows from $\norm{\hat U - U}_F^2 = \Op(\max \{N,T\})$ (the result of Lemma \ref{lem: hat U Frobenius rate}) and condition \eqref{item: Lind rate}. This completes the proof.  
\end{proof}

\begin{lemma}\label{lem: omega sum verification}
  Condition (\ref{item: HL var U rate}) of Lemma \ref{lem: se consistency}
  holds under Assumption \ref{U_conditional_moment_bound_assump}.
\end{lemma}
\begin{proof}

  The quantity in condition (\ref{item: HL var U rate}) of Lemma \ref{lem: se
    consistency} has mean zero and variance conditional on $W,\Gamma$ bounded by
  \begin{align*}
    \frac{\overline\omega^2}{(NT)^2}\sum_{i=1}^N\sum_{t=1}^T \e_P [U_i^4|W,\Gamma]
    \le \frac{\overline\omega^2/\eta}{NT}.
  \end{align*}
  This gives the $\Op(1/\sqrt{NT})$ rate as claimed.

\end{proof}

\begin{lemma}
  \label{lem: u-normality verification}
  Suppose that the hypotheses of Lemma \ref{lem: se consistency} are satisfied,
  and that Assumption \ref{U_conditional_moment_bound_assump} holds.
  Then $\langle A, U \rangle_F/\widehat{\se}
  \underset{\Theta,\mathcal{P}}{\overset{d}{\to}}
  N(0,1)$.
\end{lemma}

\begin{proof}[Proof of Lemma \ref{lem: u-normality verification}]
  First, we verify
  \begin{align*}
    \frac{\sum_{i,t} A_{it}^2 U_{it}^2}{\sum_{i,t} A_{it}^2 \sigma_{it}^2} - 1 = \op (1).
  \end{align*}
  Here $\sigma_{it}^2 \equiv \sigma_{it}^2 (W, \Gamma) = \e [U_{it}^2|W, \Gamma]$, where we drop the dependence of $\sigma_{it}^2(W,\Gamma)$ on $W$ and $\Gamma$ for brevity of notation. Notice that
  \begin{align*}
    \frac{\sum_{i,t} A_{it}^2 U_{it}^2}{\sum_{i,t} A_{it}^2 \sigma_{it}^2} - 1 = \underbrace{\frac{\sqrt{NT} \max_{i,t} A_{it}^2}{\sum_{i,t} A_{it}^2 \sigma_{it}^2}}_{\uto 0}  \underbrace{\frac{1}{\sqrt{NT}} \sum_{i,t} \frac{A_{it}^2}{\max_{i,t} A_{it}^2}  (U_{it}^2 - \sigma_{it}^2)}_{\Op(1)} = \op (1),
  \end{align*}
  where the first factor (uniformly) converges to zero due to conditions \eqref{item: var U lower bound} and \eqref{item: Lind rate} of Lemma \ref{lem: se consistency}, and the second factor is (uniformly) bounded in probability due to condition \eqref{item: HL var U rate} of Lemma \ref{lem: se consistency}. Combining this result with the result of Lemma \ref{lem: se consistency}, we obtain
  \begin{align}
    \label{eq: uniform se consistency}
    \sqrt{\frac{\sum_{it} A_{it}^2 \hat U_{it}^2}{\sum_{it} A_{it}^2 \sigma_{it}^2}} - 1 = \op (1).
  \end{align}

  Second, we demonstrate
  \begin{align}
    \label{eq: uniform CLT}
    \frac{\sum_{i,t} A_{it} U_{it}}{\sqrt{\sum_{i,t} A_{it}^2 \sigma_{it}^2}} \dist N(0,1).
  \end{align}
  Let $Q_{it} = A_{it} U_{it} / \sqrt{\sum_{i,t} A_{it}^2 \sigma_{it}^2}$ and
  $S_{N,T} = \sum_{i,t} Q_{it}$. Following the lines of the proof of Lemma F.1
  in \citet{armstrong2018optimal}
  (and using Assumption \ref{U_conditional_moment_bound_assump}
  and conditions \eqref{item: var U lower bound} and \eqref{item: Lind rate} of Lemma \ref{lem: se consistency}), we conclude that for all sequences of $W = W_{N,T}$ and $\Gamma = \Gamma_{N,T}$ we have for any fixed $\varepsilon > 0$
  \begin{align}
    \label{eq: u lind}
    \sum_{i,t} \e \left[Q_{it}^2 \mathbbm{1} \{\abs{Q_{it}} > \varepsilon  \} | W, \Gamma \right] \uto 0.
  \end{align}
  Note that $\eqref{eq: u lind}$ is a uniform version of the Lindeberg condition (applied conditional on $W$ and $\Gamma$). Hence, following the lines of the proof of the Lindeberg CLT (see, for example, Theorem 27.2 and its proof in \citealp{billingsley1995probability}), we establish that, for any fixed $t \in \mathbb R$,
  \begin{align*}
    \abs{\ex{e^{i S_{N,T} t}|W, \Gamma} - e^{-t^2/2}} \uleq r_{N,T} \quad \text{a.s.}
  \end{align*}
  for some $r_{N,T} \downarrow 0$. Hence, we also have
  \begin{align*}
    \abs{\ex{e^{i S_{N,T} t}} - e^{-t^2/2}} \uto 0,
  \end{align*}
  which implies $S_{N,T} \dist N(0,1)$ and verifies \eqref{eq: uniform CLT}. \eqref{eq: uniform se consistency} and \eqref{eq: uniform CLT} together deliver the result.
\end{proof}

\subsubsection{Verification of (\ref{eq: APlambda})}
\label{sssec: proof of new validity theorem}

Using \eqref{eq: uniform se consistency}, we have
\begin{align*}
  \frac{\langle A, P_{\lambda} U \rangle_F}{\widehat{\se}} = \frac{\langle A, P_{\lambda} U \rangle_F}{\sqrt{\sum_{i,t} A_{it}^2 \sigma_{it}^2 }} (1 + \op(1)),
\end{align*}
so it is sufficient to show that
\begin{align}
  \label{eq: small P lambda U}
  \frac{\langle A, P_{\lambda} U \rangle_F}{\sqrt{\sum_{i,t} A_{it}^2 \sigma_{it}^2 }} = \op(1).
\end{align}

Before we proceed notice that since the low rank representation $\Gamma = \lambda' f$ is arbitrary (even though it is not unique), we can proceed with any compatible mapping $\lambda = \lambda (\Gamma)$ and make $\lambda$ non-random conditional on $\Gamma$.

First, we bound $\var_P (\langle A, P_{\lambda} U \rangle_F|W,\Gamma)$. Note that
\begin{align*}
  \var_P (\langle A, P_{\lambda} U \rangle_F|W,\Gamma) &= \var_P (\langle P_{\lambda} A, U \rangle_F|W,\Gamma) \\
  &= \sum_{i=1}^N \sum_{t=1}^T \e_P [(P_\lambda A)_{it}^2 U_{it}^2 |W, \Gamma]\\
  &\leq \overline \sigma^2 \norm{P_\lambda A}_F^2,
\end{align*}
where $\overline \sigma^2$ is a uniform upper bound on $\e_P [U_{it}^2|W,\Gamma]$. Thus,
\begin{align*}
  \var_P\left(\frac{\langle A, P_{\lambda} U \rangle_F}{\sqrt{\sum_{i,t} A_{it}^2 \sigma_{it}^2 }}\Bigg|W,\Gamma\right) \leq  \frac{\overline \sigma^2}{\underline \sigma^2} \frac{\norm{P_\lambda A}_F^2}{\norm{A}_F^2}.
\end{align*}
Note that to complete the proof it is sufficient to show that
\begin{align*}
  \frac{\norm{P_\lambda A}_F^2}{\norm{A}_F^2} \uto 0
\end{align*}
because then \eqref{eq: small P lambda U} will follow Chebyshev's inequality.

Since $\rank{P_\lambda} \leq R$, we have
\begin{align}
  \label{eq: frac A_F bound}
  \frac{\norm{P_\lambda A}_F^2}{\norm{A}_F^2} \leq \frac{R s_1 (P_\lambda A)^2}{\norm{A}_F^2} \leq \frac{R s_1(A)^2}{\norm{A}_F^2}.
\end{align}
Since $A = A^*_{b^*,c}$, according to \eqref{eq: A mse bound}, we also have
\begin{align}
  \label{eq: s_1 A bound}
  s_1 (A)^2 \leq \Op (\max\{N,T\}/(NT)^2),
\end{align}
where we also used $b^{*2} \propto \max\{N,T\}$.

Next, we also want to bound $\norm{A}_F^2$ from below. Consider the following minimization problem
\begin{align*}
  \min_A \norm{A}_F^2 \quad \text{s.t.} \quad \langle A,X \rangle_F = 1.
\end{align*}
Notice that by the Gauss-Markov theorem the solution is given by $A_{\rm GM} = X/\norm{X}_F^2$. Since the constructed $A = A^*_{b^*,c}$ also needs to satisfy the constraint $\langle A,X \rangle_F = 1$, we also have $\norm{A}_F^2 \geq \norm{A_{\rm GM}}^2_F = 1/\norm{X}_F^2$. Combining this with \eqref{eq: frac A_F bound} and \eqref{eq: s_1 A bound}, we obtain
\begin{align*}
  \frac{\norm{P_\lambda A}_F^2}{\norm{A}_F^2} \leq R s_1 (A) \norm{X}_F^2 = \Op (1/\min\{N,T\}),
\end{align*}
where we used $\norm{X}_F^2 = \Op (NT)$, which follows from Assumptions \ref{ass: X decomposition}\eqref{V_assump}-\eqref{Z_V_inner_product_assump}. This completes the proof.

{\subsection{Proof of Theorem~\ref{th:SemiStrong}}

To prove Theorem~\ref{th:SemiStrong}, we first state and prove an auxiliary lemma in Section \ref{sssec: rankRmatrix lemma} below. In Section \ref{sssec: semi strong thm proof}, we then prove the final results.

 \subsubsection{Auxiliary Lemma}
 \label{sssec: rankRmatrix lemma}
 \begin{lemma}
    \label{lemma:RankRmatrix}
    Let  $A$, $\lambda$, $f$ be $N\times T$, $N \times R$, and $T \times R$  matrices, respectively.
    Assume that the matrices $P_\lambda A P_f$,  $\lambda$, and $f$ all have rank equal to $R$.
    Then we have
    \begin{align*}
         {\rm rank}(A) = R
         \qquad 
         \Longleftrightarrow
           \qquad 
           M_\lambda  A  M_f 
     =   (M_\lambda A  P_f )  \left( P_\lambda  A  P_f  \right)^+ (P_\lambda  A  M_f).
    \end{align*}
 \end{lemma}
 
 Notice that   $ (M_\lambda A  P_f )  \left( P_\lambda  A  P_f  \right)^+ (P_\lambda  A  M_f)$ in the statement of the lemma
 can equivalently be written as   $ M_\lambda A    \left( P_\lambda  A  P_f  \right)^+  A  M_f$, because
 $ \left( P_\lambda  A  P_f  \right)^+ = P_f  \left( P_\lambda  A  P_f  \right)^+ P_\lambda$.
 
 \begin{proof}%
     First, consider the special case
     \begin{align}
          \lambda &=  { \mathbb{I}_R \choose 0_{(N-R) \times R} } ,
          &
          f &=  { \mathbb{I}_R \choose 0_{(T-R) \times R} } .
          \label{SpecialCase}
     \end{align}
     In that case, let $[P_{\lambda} A P_{f}]_\# = \lambda' A f$ be the non-zero $R \times R$ block of the $N \times T$ matrix 
     $P_{\lambda} A P_{f}$, and analogously, let  $[M_{\lambda} A P_{f}]_\# $, $[P_{\lambda} A M_{f}]_\# $, $[M_{\lambda} A M_{f}]_\# $
     be the non-zero $(N-R) \times R$, $R \times (T-R)$, $(N-R) \times (T-R)$ blocks of
     the $N \times T$ matrices $M_{\lambda} A P_{f}$, $P_{\lambda} A M_{f}$, $M_{\lambda} A M_{f}$, respectively.
     With those definitions we have
\begin{align*}
   A &= \left( \begin{array}{cc} 
         \,  [P_{\lambda} A P_{f}]_{\#}  &  [P_{\lambda} A M_{f}]_{\#}
            \\
         \,   [M_{\lambda} A P_{f}]_{\#} &  [M_{\lambda} A M_{f}]_{\#}
    \end{array} \right)
\end{align*}
 For any $i=1,\ldots,N-R$ and $t=1,\ldots,T-R$, we now construct an $(R+1) \times (R+1)$ submatrix of this matrix as follows
\begin{align*}
   \left( \begin{array}{cc} 
         \,  [P_{\lambda} A P_{f}]_{\#}  &  [P_{\lambda} A M_{f}]_{\#} \; e_t
            \\
         \,  e_i' \; [M_{\lambda} A P_{f}]_{\#} & e_i'  \; [M_{\lambda} A M_{f}]_{\#} \; e_t
    \end{array} \right)
\end{align*}
where $e_k$ refers to the $k$'th standard basis vector of appropriate dimension.
The determinant of this submatrix is given by
\begin{align*}
    \det([P_{\lambda} A P_{f}]_{\#} )
   \left[   
     e_i'  \; [M_{\lambda} A M_{f}]_{\#} \; e_t
     -   e_i' \; [M_{\lambda} A P_{f}]_{\#}    \left( [P_{\lambda} A P_{f}]_{\#}  \right)^{-1}  [P_{\lambda} A M_{f}]_{\#} \; e_t
     \right] = 0.
\end{align*}
If this determinant is zero for any such  $(R+1) \times (R+1)$ submatrix, then we can conclude that
 $  {\rm rank}(A) \leq R$, which together with our assumption  $  {\rm rank}(P_\lambda A P_f) = R$
 implies that  $  {\rm rank}(A) = R$. Conversely, if  $  {\rm rank}(A) = R$, then the determinant
 of any such  $(R+1) \times (R+1)$ submatrix needs to be zero. 

The assumption  $  {\rm rank}(P_\lambda A P_f) = R$ guarantees that   $ \det([P_{\lambda} A P_{f}]_{\#} ) \neq 0$. 
Thus, we have $  {\rm rank}(A) = R$ if and only if
\begin{align*}
 e_i'  \;  \left[   
      [M_{\lambda} A M_{f}]_{\#}  
     -    [M_{\lambda} A P_{f}]_{\#}    \left( [P_{\lambda} A P_{f}]_{\#}  \right)^{-1}  [P_{\lambda} A M_{f}]_{\#} 
     \right]  \; e_t = 0.
\end{align*}
for all $i=1,\ldots,N-R$ and $t=1,\ldots,T-R$. This can equivalently be written as
\begin{align*}
       [M_{\lambda} A M_{f}]_{\#}  
     &=   [M_{\lambda} A P_{f}]_{\#}    \left( [P_{\lambda} A P_{f}]_{\#}  \right)^{-1}  [P_{\lambda} A M_{f}]_{\#} ,
\end{align*}
or  equivalently
\begin{align*}
    M_{\lambda} A M_{f} &= M_{\lambda} A P_{f}    \left( P_{\lambda} A P_{f}   \right)^{+} P_{\lambda} A M_{f}.
\end{align*}
We have thus shown the lemma for the special case that $\lambda$ and $f$ are of the form \eqref{SpecialCase}.

For any other $\lambda$ and $f$ that have full rank $R$ (as assumed in the lemma) we can choose an orthogonal
$N \times N$ matrix $O_1$ and an orthogonal $T \times T$ matrix $O_2$ such that
$O_1 \lambda =  { \mathbb{I}_R \choose 0_{(N-R) \times R} } $
and $O_2 f =   { \mathbb{I}_R \choose 0_{(T-R) \times R} }$. By applying the result already proven to the
transformed data $O_1  A O_2'$, $O_1 \lambda $, $O_2 f$ we then obtain the result of the lemma more generally,
because the statement of the lemma is invariant under such orthogonal transformations.%
 \end{proof}
 
\subsubsection{Proof of Theorem~\ref{th:SemiStrong}}
\label{sssec: semi strong thm proof}
First, let $\hat U \coloneqq U - W \cdot (\hat \gamma_{\rm pre} - \gamma)$ and notice that we have
\begin{align}
  \label{eq: A s_1 U hat aux}
  s_1 (\hat U) = s_1 (U) + \op(\maxSNT) = s_1 (U) ( 1 + \op(1)) \asy \maxSNT, 
\end{align}
where we used Assumption \ref{ass: factor model}(\ref{item: SN}) and Lemma \ref{gamma_pre_rate_lemma}. Together with the hypothesis of the theorem it implies that
\begin{align}
    \frac{s_1(\hat U) } {  s_R(\Gamma)} = \op(1).
    \label{StrongFactorAssumption}
\end{align}
This condition implies that $s_R(\Gamma) \neq 0$, which together
with our assumption ${\rm rank}(\Gamma) \leq R$ implies that
${\rm rank}(\Gamma) = R$. Since $\Gamma = \lambda f'$, this also implies that the rank of both $\lambda$ and $f$ is equal to $R$
as well.

Next, define
\begin{align*}
    \Gamma^\dagger &:= \Gamma + M_{\lambda} \hat U P_{f}  + P_{\lambda} \hat U M_{f} + M_{\lambda} \hat U   \Gamma^{+} \hat U M_{f} .
\end{align*}
Using Lemma~\ref{lemma:RankRmatrix} and that $ \Gamma= P_\lambda \Gamma P_f$, we conclude that
\begin{align*}
     {\rm rank}\left( \Gamma^\dagger \right)  = R.
\end{align*} 
Our first goal in the following is to derive a bound on $  \Vert \widehat  \Gamma_{\rm pre}  - \Gamma^\dagger  \Vert_*  $. Denote
\begin{align*}
  Q(G) \coloneqq \left\| Y - W \cdot \hat \gamma_{\rm pre}  - G \right\|_F^2 = \left\| \hat U + \Gamma  - G \right\|_F^2.
\end{align*}
Then
\begin{align*}
    Q(\Gamma^\dagger) &=    \left\| \hat U + \Gamma  -  \Gamma^\dagger \right\|_F^2 
     \\
     &= \left\| \hat U - M_{\lambda} \hat U P_{f}  - P_{\lambda} \hat U M_{f} - M_{\lambda} \hat U   \Gamma^{+} U M_{f} \right\|_F^2      
     \\
     &= \left\| P_{\lambda} \hat U P_{f}  + M_{\lambda} \hat U M_{f} - M_{\lambda} \hat U   \Gamma^{+} \hat U M_{f} \right\|_F^2      
     \\
     &=    
         \left\|  M_\lambda  \hat U  M_f  - M_{\lambda} \hat U   \Gamma^{+} \hat U M_{f} \right\|_F^2 +  \left\| P_\lambda \hat U P_f  \right\|_F^2 
      \\
      &=  \left\|  M_\lambda  \hat U  M_f  \right\|_F^2  
        - 2 {\rm Tr}\left( M_{f}  \hat U' M_{\lambda} \hat U   \Gamma^{+} \hat U  M_\lambda  \right)
       + \left\|    M_{\lambda} \hat U   \Gamma^{+} U M_{f} \right\|^2_F 
  +  \left\| P_\lambda \hat U P_f  \right\|_F^2 ,
\end{align*}
where in the last step we used that $\|A+B\|_F^2 = {\rm Tr}((A+B)'(A+B))  = \|A\|_F^2 + 2{\rm Tr}(A'B)  + \|B\|_F^2 $.
Furthermore, using  
\begin{align}
     \left\|    M_{\lambda} \hat U   \Gamma^{+} \hat U M_{f} \right\|_F 
     &\leq
   R \,   s_1(   M_{\lambda} \hat U   \Gamma^{+} \hat U M_{f}   ) 
   \leq R \,    [s_1(\hat U)]^2 s_1(  \Gamma^{+} )
   =   \frac{R \,   [s_1(\hat U)]^2 } {s_R(\Gamma)}  ,
   \label{BoundUGammaU}
\end{align}
and\footnote{
Here, we applied the general inequality $|{\rm Tr}(A)| \leq {\rm rank}(A) \, s_1(A) $, which is an immediate consequence of
von Neumann's trace inequality, with $A=M_{f}  \hat U' M_{\lambda} \hat U   \Gamma^{+} \hat U  $, where
${\rm rank}(A)\leq {\rm rank}( \Gamma^{+} )=R$.
}
\begin{align}
   \left| {\rm Tr}\left( M_{f}  \hat U' M_{\lambda} \hat U   \Gamma^{+} \hat U    \right)  \right|
       &\leq  R \, s_1(M_{f}  \hat U' M_{\lambda} \hat U   \Gamma^{+} \hat U )
      \nonumber \\
       &\leq  R \, [s_1(\hat U)]^3 \,  s_1(  \Gamma^{+} )   =   \frac{R \,   [s_1(\hat U)]^3 } {s_R(\Gamma)} ,
       \label{BoundTrGammaUUU}
\end{align}
we thus obtain
\begin{align}
      Q(\Gamma^\dagger) &\leq
     \left\| M_\lambda \hat U  M_f  \right\|_F^2 
       +  \left\| P_\lambda \hat U P_f  \right\|_F^2
     + \Op\left\{  \frac{[s_1(\hat U)]^4} {[s_R(\Gamma)]^2} + \frac{[s_1(\hat U)]^3}  {s_R(\Gamma)}   \right\} 
     \nonumber \\
     &=    \left\| M_\lambda \hat U  M_f  \right\|_F^2 +    \left\| P_\lambda \hat U P_f  \right\|_F^2 + \op\left[ (s_1(\hat U))^2  \right]   .
     \label{BoundQGammaTarget}
\end{align}
where in the last step we used \eqref{StrongFactorAssumption}.

 Next, for any $N \times T$ matrix $G$ we have
\begin{align*}
      Q(G) &=   \left\| G - \Gamma - \hat U   \right\|_F^2  
      \\
      &=  \left\| P_\lambda \left(G - \Gamma - \hat U \right) P_f  \right\|_F^2
           +  \left\| P_\lambda \left(G   - \hat U \right) M_f  \right\|_F^2
       \\ & \qquad \qquad
            +  \left\| M_\lambda \left(G   - \hat U \right) P_f  \right\|_F^2
            +  \left\| M_\lambda \left(G   - \hat U \right) M_f  \right\|_F^2 ,
\end{align*}
where we used that $M_\lambda \Gamma = 0$ and $\Gamma M_f  = 0$. In particular, we have
\begin{align*}
    Q( \widehat \Gamma_{\rm pre})    &=    \left\| P_\lambda \left(\widehat \Gamma_{\rm pre} - \Gamma - \hat U \right) P_f  \right\|_F^2
           +  \left\| P_\lambda \left(\widehat \Gamma_{\rm pre}   - \hat U \right) M_f  \right\|_F^2
           \\ & \qquad \qquad
            +  \left\| M_\lambda \left(\widehat \Gamma_{\rm pre}   - \hat U \right) P_f  \right\|_F^2
            +  \left\| M_\lambda \left(\widehat \Gamma_{\rm pre}   - \hat U \right) M_f  \right\|_F^2 .
\end{align*}
We have
\begin{align*}
     \left\| M_\lambda \left(\widehat \Gamma_{\rm pre}   - \hat U \right) M_f  \right\|_F^2
     &=    \left\| M_\lambda  \hat U M_f  - M_\lambda   \widehat \Gamma_{\rm pre} M_f  \right\|_F^2 
   \\
      &\geq \argmin_{\left\{ G \, : \, {\rm rank}(G) \leq  R \right\}} 
    \left\| M_\lambda  \hat U M_f  - G  \right\|_F^2
    \\
    &=  \left\| M_\lambda \hat U  M_f  \right\|_F^2 - \sum_{r=1}^R [s_r(M_\lambda \hat U  M_f )]^2
    \\
    &\geq   \left\| M_\lambda \hat U  M_f  \right\|_F^2 - R  [s_1(\hat U)]^2 ,
\end{align*}
where for the first inequality we used that $M_\lambda   \widehat \Gamma_{\rm pre} M_f  $ is a matrix of rank at most $R$, that is,
minimizing over all such matrices can only make the expression smaller, and in the next step we used that the solution to this least squares
minimization problem over $\left\{ G \, : \, {\rm rank}(G) \leq  R \right\}$ is given by the principal components
of $M_\lambda  \hat U M_f$. We thus obtain that
\begin{align}
      Q( \widehat \Gamma_{\rm pre})      &\geq  \left\| P_\lambda \left(\widehat \Gamma_{\rm pre} -  \Gamma  \right) P_f  - P_\lambda  \hat U P_f \right\|_F^2
           +  \left\| P_\lambda \left( \widehat \Gamma_{\rm pre}   - \hat U \right) M_f  \right\|_F^2
         \nonumber    \\ & \qquad  \qquad
            +  \left\| M_\lambda \left( \widehat \Gamma_{\rm pre}   - \hat U \right) P_f  \right\|_F^2
            + \left\| M_\lambda \hat U  M_f  \right\|_F^2 - R  [s_r(\hat U)]^2 .
            \label{LowerBoundHatGammaPreliminary}      
\end{align}

Since $ \widehat \Gamma_{\rm pre}  $ minimizes $Q(G)$ over all rank $\leq R$ matrices,
and   $ {\rm rank}\left( \Gamma^\dagger \right)  = R$, we know that 
$Q( \widehat \Gamma_{\rm pre} ) \leq Q(\Gamma^\dagger)$. Combining \eqref{BoundQGammaTarget}  
and \eqref{LowerBoundHatGammaPreliminary} we thus obtain that
\begin{align*}
    &  \left\| P_\lambda \left(\widehat \Gamma_{\rm pre} -  \Gamma  \right) P_f  - P_\lambda  \hat U P_f \right\|_F^2
               +  \left\| P_\lambda \left( \widehat \Gamma_{\rm pre}   - \hat U \right) M_f  \right\|_F^2
            +  \left\| M_\lambda \left( \widehat \Gamma_{\rm pre}   - \hat U \right) P_f  \right\|_F^2
          \\ & \qquad\qquad\qquad\qquad\qquad\qquad\qquad\qquad\qquad  
          \leq   R  [s_1(\hat U)]^2     + \op\left[ (s_1(\hat U))^2  \right]  +   \left\| P_\lambda \hat U P_f  \right\|_F^2 
            \\ & \qquad\qquad\qquad\qquad\qquad\qquad\qquad\qquad\qquad  
            = \Op(  [s_1(\hat U)]^2  ) .
\end{align*}
Since all three terms on the lhs are positive, this implies separately for each of them
\begin{align}
     \left\| P_\lambda \left(\widehat \Gamma_{\rm pre} -  \Gamma   \right) P_f  \right\|_F &\leq
           \left\| P_\lambda \left(\widehat \Gamma_{\rm pre} -  \Gamma  \right) P_f  - P_\lambda  \hat U P_f \right\|_F
           + \left\| P_\lambda  \hat U P_f  \right\|_F 
         \nonumber  \\
           &= \Op(s_1(\hat U)) ,
    \nonumber  \\
       \left\| P_\lambda  \widehat \Gamma_{\rm pre}   M_f  \right\|_F &=  \Op(s_1(\hat U))  ,
   \nonumber \\
       \left\| M_\lambda  \widehat \Gamma_{\rm pre}   P_f  \right\|_F &=  \Op(s_1(\hat U)) .
    \label{PreliminaryBoundsGammaHat}   
\end{align}
The first result in the last display implies
$$
    s_1\left( P_\lambda \widehat \Gamma_{\rm pre} P_f  -   \Gamma     \right) 
    \leq    \left\| P_\lambda \left(\widehat \Gamma_{\rm pre} -  \Gamma   \right) P_f  \right\|_F 
     =  \Op(s_1(\hat U)) ,
$$
and by Weyl inequality for singular values we thus find that
\begin{align}
    s_R\left( P_\lambda \widehat \Gamma_{\rm pre} P_f  \right)
     &\geq s_R\left(  \Gamma   \right) -  s_1\left( P_\lambda \widehat \Gamma_{\rm pre} P_f  -   \Gamma     \right) 
   \nonumber  \\
     & \geq s_R\left(  \Gamma   \right)  -   \Op(s_1(\hat U))
     \nonumber  \\
     & \geq 0 , \qquad {\rm wpa1,} 
    \label{PreliminaryBoundsGammaHatAgain}   
\end{align}
where in the last step we used  \eqref{StrongFactorAssumption} again.
We thus have, wpa1, that 
$P_\lambda  \widehat \Gamma_{\rm pre}   P_f$ has rank equal to $R$, and by applying Lemma~\ref{lemma:RankRmatrix}
we thus find
\begin{align*}
    M_\lambda  \widehat \Gamma_{\rm pre}   M_f 
    &=   (M_\lambda  \widehat \Gamma_{\rm pre}   P_f )  \left( P_\lambda  \widehat \Gamma_{\rm pre}   P_f  \right)^+ 
    (P_\lambda  \widehat \Gamma_{\rm pre}   M_f)   .
\end{align*}
Together with the bounds in \eqref{PreliminaryBoundsGammaHat} and \eqref{PreliminaryBoundsGammaHatAgain} and \eqref{StrongFactorAssumption} this implies that
\begin{align}
    \left\| M_\lambda  \widehat \Gamma_{\rm pre}   M_f \right\|_F
    & \leq R \,    s_1( M_\lambda  \widehat \Gamma_{\rm pre}   M_f )
  \nonumber  \\
    & \leq \frac{ R \,  s_1(M_\lambda  \widehat \Gamma_{\rm pre}   P_f ) \, s_1 (P_\lambda  \widehat \Gamma_{\rm pre}   M_f) }
     {s_R(P_\lambda  \widehat \Gamma_{\rm pre}   P_f )}
  \nonumber  \\
    &= \Op\left(  \frac{[s_1(\hat U)]^2}{s_R(\Gamma)}   \right).
    \label{PreliminaryBoundsGammaHat2}
\end{align}

We can now improve on the lower bound in $Q( \widehat \Gamma_{\rm pre})$ in  \eqref{LowerBoundHatGammaPreliminary}. 
As before, we have  
\begin{align*}
    Q( \widehat \Gamma_{\rm pre})    &=    \left\| P_\lambda \left(\widehat \Gamma_{\rm pre} - \Gamma - \hat U \right) P_f   \right\|_F^2
           +  \left\| P_\lambda \left(\widehat \Gamma_{\rm pre}   - \hat U \right) M_f  \right\|_F^2
           \\ & \qquad \qquad
            +  \left\| M_\lambda \left(\widehat \Gamma_{\rm pre}   - \hat U \right) P_f  \right\|_F^2
            +  \left\| M_\lambda \left(\widehat \Gamma_{\rm pre}   - \hat U \right) M_f  \right\|_F^2 .
\end{align*}
Similarly, using  the definition of $ \Gamma^\dagger$ we obtain
\begin{align*}
    \left\|         \widehat \Gamma_{\rm pre} -  \Gamma^\dagger - P_\lambda \hat U P_f  \right\|_F^2
    &=           \left\| P_\lambda \left(\widehat \Gamma_{\rm pre} - \Gamma - \hat U\right) P_f    \right\|_F^2
           +  \left\| P_\lambda \left(\widehat \Gamma_{\rm pre}   - \hat U \right) M_f  \right\|_F^2
           \\ & \quad
            +  \left\| M_\lambda \left(\widehat \Gamma_{\rm pre}   - \hat U \right) P_f  \right\|_F^2
            +  \left\| M_\lambda \left(\widehat \Gamma_{\rm pre}   - M_{\lambda} \hat U   \Gamma^{+} \hat U M_{f} \right) M_f  \right\|_F^2 .
\end{align*} 
Taking the difference of the equalities in the last two displays we obtain
\begin{align*}
     Q( \widehat \Gamma_{\rm pre})  -  &     \left\|         \widehat \Gamma_{\rm pre} -  \Gamma^\dagger - P_\lambda \hat U P_f  \right\|_F^2
   \\  
     &= \left\| M_\lambda \left(\widehat \Gamma_{\rm pre}   - \hat U \right) M_f  \right\|_F^2 
      -  \left\| M_\lambda \left(\widehat \Gamma_{\rm pre}   - M_{\lambda} \hat U   \Gamma^{\dagger} \hat U M_{f} \right) M_f  \right\|_F^2 
   \\
    &= \left\| M_\lambda \left(\widehat \Gamma_{\rm pre}   - \hat U \right) M_f  \right\|_F^2 
      +  \Op\left(  \frac{[s_1(\hat U)]^4}{[s_R(\Gamma)]^2}   \right) 
  \\
  &=     \left\| M_\lambda \hat U  M_f  \right\|_F^2 
    - 2 {\rm Tr}(    \hat U'  M_\lambda \widehat \Gamma_{\rm pre}  M_f  )   +  \Op\left(  \frac{[s_1(\hat U)]^4}{[s_R(\Gamma)]^2}   \right) 
 \\
  &\geq      \left\| M_\lambda \hat U  M_f  \right\|_F^2  - 2 R s_1(\hat U) s_1( M_\lambda \widehat \Gamma_{\rm pre}  M_f  ) +  \Op\left(  \frac{[s_1(\hat U)]^4}{[s_R(\Gamma)]^2}   \right) 
  \\
    &=      \left\| M_\lambda \hat U  M_f  \right\|_F^2  + \Op\left(  \frac{[s_1(\hat U)]^3}{s_R(\Gamma)}   \right) +  \Op\left(  \frac{[s_1(\hat U)]^4}{[s_R(\Gamma)]^2}   \right) 
    \\
     &=      \left\| M_\lambda \hat U  M_f  \right\|_F^2     + \op\left(  \left[s_1(\hat U)\right]^2  \right)     ,
\end{align*}
where we used \eqref{PreliminaryBoundsGammaHat2}  and \eqref{BoundUGammaU},
 the trace term was bounded analogously to \eqref{BoundTrGammaUUU},
 and in the final step we also used \eqref{StrongFactorAssumption}.
Again using that $Q( \widehat \Gamma_{\rm pre} ) \leq Q(\Gamma^\dagger)$ 
and \eqref{BoundQGammaTarget}
we now obtain
\begin{align}
  \left\|         \widehat \Gamma_{\rm pre} -  \Gamma^\dagger - P_\lambda \hat U P_f  \right\|^2_F 
    \leq  \op\left(  \left[s_1(\hat U)\right]^2  \right)  + \left\| P_\lambda \hat U P_f \right\|_F^2.
    \label{eq: A aux bound semi strong}
\end{align}

Our next step is to bound $ \Vert P_\lambda \hat U P_f  \Vert_F$. Note that, using Assumption~\ref{U_conditional_moment_bound_assump}, 
\begin{align*}
   \mathbb{E} \left[\left\| P_\lambda U P_f  \right\|_F^2 | \Gamma \right] 
   &= \mathbb{E}  \left[{\rm vec}(P_\lambda U P_f)'  {\rm vec}(P_\lambda U P_f) | \Gamma \right]    
   \\
   &= \mathbb{E} \left[{\rm Tr}\left[ {\rm vec}(P_\lambda U P_f)    {\rm vec}(P_\lambda U P_f)'   \right] | \Gamma \right]
   \\
   &= 
    {\rm Tr}\Big[ (P_f \otimes P_\lambda)
   \underbrace{\mathbb{E}\left[ {\rm vec}(U) {\rm vec}(U)' | \Gamma \right]
    }_{\leq \overline \sigma^2 \mathbb{I}_{NT}}
      (P_f \otimes P_\lambda) \Big]
 \\
   &\leq \overline \sigma^2 {\rm Tr}\left[ P_f \otimes P_\lambda \right]
  \\ 
   & =  \overline \sigma^2   {\rm Tr}(P_f)  {\rm Tr}(P_\lambda) 
  \\
   & = \overline \sigma^2 R^2
\end{align*}
for some $\overline \sigma^2 > 0$ for all $\Gamma$. Hence, we also have
\begin{align*}
  \mathbb{E} \left[\left\| P_\lambda U P_f  \right\|_F^2\right] \leq \overline \sigma^2 R^2,
\end{align*}
and, using Markov's inequality we conclude $ \Vert P_\lambda U P_f  \Vert_F^2 = \Op (1)$. Finally, using Assumption~\ref{ass: X decomposition} and Lemma~\ref{gamma_pre_rate_lemma}, 
we obtain
\begin{align*}
  \Vert P_\lambda \hat U P_f  \Vert_F \leq \Vert P_\lambda U P_f  \Vert_F + \Vert W \cdot (\hat \gamma_{\rm pre}  - \gamma) \Vert_F = \op({\maxSNT}).
\end{align*}
Combining this result with \eqref{eq: A aux bound semi strong}, we obtain
\begin{align*}
       \left\|         \widehat \Gamma_{\rm pre} -  \Gamma^\dagger \right\|_F
       &\leq     \left\|         \widehat \Gamma_{\rm pre} -  \Gamma^\dagger - P_\lambda \hat U P_f  \right\|_F  +\left\| P_\lambda \hat U P_f \right\|  
       \\
       &=    \op\left( s_1(U)  \right),
\end{align*}
where the last equality also uses \eqref{eq: A s_1 U hat aux}.

Since $\|A\|_* \leq \sqrt{{\rm rank}(A)} \, \|A\|_F$, this also implies that
\begin{align*}
       \left\|         \widehat \Gamma_{\rm pre} -  \Gamma^\dagger \right\|_*
      =   o_{\Theta, \mathcal P}\left( s_1(U)  \right)  = o_{\Theta, \mathcal P} (\maxSNT).
\end{align*}

We have thus shown that
\begin{align*}
    \left\|  \widetilde \Gamma  - M_{\lambda} \hat U   \Gamma^{+} \hat U M_{f}   \right\|_*  =  o_{\Theta, \mathcal P} (\maxSNT) ,
\end{align*}
which implies
\begin{align*}
    \left\|  \widetilde \Gamma   \right\|_*  &=
     \left\|  M_{\lambda} \hat U   \Gamma^{+} \hat U M_{f}   \right\|_*
     +
    o_{\Theta, \mathcal P} (\maxSNT) 
    \\
     &= \sqrt{R}
     \left\|  M_{\lambda} \hat U   \Gamma^{+} \hat U M_{f}   \right\|_F
     +
    o_{\Theta, \mathcal P} (\maxSNT) 
    \\
    &=  o_{\Theta, \mathcal P} (\maxSNT) ,
\end{align*}
where in the last step we used \eqref{BoundUGammaU}, \eqref{StrongFactorAssumption},
and \eqref{eq: A s_1 U hat aux} again. This completes the proof.

}

\subsection{Verification of Assumption~\ref{ass: factor model}\eqref{item: NC}}
\label{appendix:VerifyAss2i}

In this section, we verify that, in the absence of $Z$, Assumption~\ref{ass: X decomposition}\eqref{V_assump} and \eqref{H_assump} imply Assumption~\ref{ass: factor model}\eqref{item: NC}.
Thus, our goal in this section is to prove that
\begin{align}
\frac{1}{NT} \sum_{r = 2R+1}^{\min\{N,T\}} s_r^2(X) \geq \underline{s}^2 > 0,
   \label{GOAL_VerifyAss2i}
\end{align}
with probability approaching one. 

By the variational principle for eigenvalues, we have
$$
\sum_{r=1}^{2R} s_r^2(X) = \sum_{r=1}^{2R} \lambda_r(X'X) = \max_{Q \in \mathbb{R}^{T \times 2R}} \operatorname{Tr}(X'X \, P_Q),
$$
where $P_Q$ is the projection matrix onto the subspace spanned by the columns of $Q$. 
Using this and $\operatorname{Tr}(X'X) =\sum_{r=1}^{\min\{N,T\}} s_r^2(X)  $, we obtain
$$
\frac{1}{NT} \sum_{r=2R+1}^{\min\{N,T\}} s_r^2(X) = \frac{1}{NT} \operatorname{Tr}(X'X) - \frac{1}{NT} \max_{Q \in \mathbb{R}^{T \times 2R}} \operatorname{Tr}(X'X \, P_Q).
$$
Using $X = H + V$, where $H$ and $V$ are as defined in Assumption~\ref{ass: X decomposition}, and substituting into the above, we obtain:
\begin{align*}
\frac{1}{NT} \sum_{r=2R+1}^{\min\{N,T\}} s_r^2(X) 
  &\geq 
 \frac{1}{NT} \operatorname{Tr}(H'H) - \frac{1}{NT} \max_{Q \in \mathbb{R}^{T \times 2R}} \operatorname{Tr}(H \, P_Q \, H')
\\ & \quad
+  \frac{1}{NT} \operatorname{Tr}(V'V) - \frac{1}{NT} \max_{Q \in \mathbb{R}^{T \times 2R}} \operatorname{Tr}(V \, P_Q \, V') 
\\ & \quad
+ \frac{2 \operatorname{Tr}(V'H)}{NT} - \frac{2}{NT} \max_{Q \in \mathbb{R}^{T \times 2R}} \operatorname{Tr}(V \, P_Q \, H') .
\end{align*}
For the three lines on the right-hand side we have
\begin{align*}
\frac{1}{NT} \operatorname{Tr}(H'H) - \frac{1}{NT} \max_{Q \in \mathbb{R}^{T \times 2R}} \operatorname{Tr}(H \, P_Q \, H')
&=\frac{1}{NT} \sum_{r=2R+1}^{\min\{N,T\}} s_r^2(H) \geq 0 ,
\\
\frac{1}{NT} \operatorname{Tr}(V'V) - \frac{1}{NT} \max_{Q \in \mathbb{R}^{T \times 2R}} \operatorname{Tr}(V \, P_Q \, V') 
&\geq \frac{1}{NT} \|V\|_F^2 - \frac{2R \, [s_1(V)]^2}{NT},
\\
  \frac{2 \operatorname{Tr}(V'H)}{NT} - \frac{2}{NT} \max_{Q \in \mathbb{R}^{T \times 2R}} \operatorname{Tr}(V \, P_Q \, H') &\geq \frac{2 \operatorname{Tr}(V'H)}{NT} - \frac{4R \, s_1(H)s_1(V)}{NT} ,
\end{align*}
and therefore
\begin{align*}
\frac{1}{NT} \sum_{r=2R+1}^{\min\{N,T\}} s_r^2(X) 
&\geq \frac{1}{NT} \|V\|_F^2 + \frac{2 \operatorname{Tr}(V'H)}{NT} - \frac{4R \, s_1(H)s_1(V)}{NT} - \frac{2R \, [s_1(V)]^2}{NT}
 \\
 &=  \frac{1}{NT} \|V\|_F^2 + o_{\Theta,\mathcal{P}}(1)  
 \\
 &\asymp_{\Theta,\mathcal{P}} 1 ,
\end{align*}
we in the last two steps we used Assumption~\ref{ass: X decomposition}\eqref{V_assump} and \eqref{H_assump}, as well as $s_1(H) \leq \|H\|_F$.
This confirms that \eqref{GOAL_VerifyAss2i} holds.

More generally, if in addition, each of the other regressors $Z_k$, $k = 1,\ldots, K$, can be decomposed
as $V_k+H_k$ such that Assumption~\ref{ass: X decomposition}\eqref{V_assump} and \eqref{H_assump} holds equally for $V_k$
and $H_k$, and the $K+1$ vector $V_{{\rm vec},it}$ that combines $V_{it}$ and $V_{k,it}$ satisfies the
standard non-collinearity condition
$$
  \plim_{N,T \rightarrow \infty} \frac 1 {NT} \sum_{i=1}^N \sum_{t=1}^T V_{{\rm vec},it} V_{{\rm vec},it}'
   > 0, 
$$
then by the same arguments as for $K=0$ above, the Assumption~\ref{ass: factor model}\eqref{item: NC} holds for general $K$.

\section{Computational details}\label{computational_details_appendix}

The optimal weights $A^*_{b}$ given in Definition \ref{DefOptimalWeights} can be
computed directly using convex programming.  Alternatively, we can obtain these
weights from a nuclear norm regularized ``partialling out'' regression of $X$ on
$Z$ and a matrix of individual effects.  This
follows by applying a
result from \citet{armstrong_bias-aware_2020} to our setting, as we now describe.
We first consider the general case with
covariates (Section \ref{computation_general_case_sec}),
and then obtain a further simplification by specializing to the case with no
additional covariates $Z$.

\subsection{General case}\label{computation_general_case_sec}

The weights $A^*_b$ minimize $\left(\maxbias_{\hat C}(\hat\beta_A)\right)^2+\sigma^2\|A\|_F^2$ when $\hat C/\sigma=b$.
Equivalently, we can minimize $\sigma^2 \|A\|_F^2$ subject to a bound on
$\maxbias_{\hat C}(\hat\beta_A)$:
\begin{align}\label{A_optimization}
  \min_{A} \sigma^2 \|A\|_F^2
  \quad\text{s.t.}\quad
  \maxbias_{\hat C}(\hat\beta_A)\le B.
\end{align}
We can then vary the bound $B$ to optimize any increasing function of the
variance $\sigma^2\|A\|_F^2$ and worst-case bias $\maxbias_{\hat C}(\hat\beta_A)$.

Let $\Pi^*_\mu,\psi^*_\mu$ solve the nuclear norm regularized
regression
\begin{align}\label{Pi_optimization}
  \min_{\Pi,\psi} \|X-Z\cdot \psi -\Pi\|_F^2/2
  + \mu \|\Pi\|_*
\end{align}
where $\mu$ indexes the penalty on the nuclear norm.
Let
\begin{align}\label{resid_equation}
  \Omega^*_\mu=X-Z\cdot \psi^*_\mu -\Pi^*_\mu
\end{align}
denote the matrix of residuals from this regression.
Let
\begin{align}\label{beta_mu_equation}
  \hat\beta_{\tilde A^*_\mu } = \langle  \tilde A^*_\mu, \widetilde Y \rangle_F
  = \frac{\langle \Omega^*_\mu, \widetilde Y \rangle_F}{\langle \Omega^*_\mu, X \rangle_F}
  \quad\text{where}\quad
  \tilde A^*_{\mu}  = \frac{\Omega^*_\mu}{\langle \Omega^*_\mu, X \rangle_F}
\end{align}
and let
\begin{align}\label{B_mu_V_mu_equation}
  \overline B_\mu = \frac{1}{\| {\Pi^*_\mu} \|}_* \frac{\langle \Omega^*_\mu, \Pi^*_\mu \rangle_F}{\langle \Omega^*_\mu, X \rangle_F}
  \quad
  \text{and}
  \quad
  V_{\mu} = \sigma^2\frac{\|\Omega^*_\mu\|_F^2}{\langle \Omega^*_\mu, X \rangle_F^2}.
\end{align}

The following theorem follows immediately from applying Theorem 2.1 in
\citet{armstrong_bias-aware_2020} to our setup (in applying the formulas from
this paper, we
use the fact that $\langle \Omega^*_\mu,Z\cdot \psi^*_\mu \rangle_F=0$ by the
first order conditions for $\psi$, since $\psi$ is unconstrained).

\begin{theorem}\label{optimal_A_theorem}
  Let $\Pi^*_\mu,\psi^*_\mu$ be a solution to (\ref{Pi_optimization}) and let
  $\Omega^*_\mu$ be the matrix of residuals in
  (\ref{resid_equation}), and suppose
  $\|\Omega^*_\mu\|>0$.  Then $\tilde A^*_\mu$ and the
  corresponding estimator $\hat\beta_{\tilde A^*_\mu}$ given in (\ref{beta_mu_equation}) 
  solve (\ref{A_optimization}) for $B=\hat C\overline B_\mu$, with minimized
  value $V_\mu$, where $\overline B_\mu$ and $V_\mu$ are given in
  (\ref{B_mu_V_mu_equation}).
\end{theorem}

Thus, to compute the MSE optimizing weights $A^*_b$, it suffices to compute the
weights $\tilde A^*_{\mu}$ for each $\mu>0$, and then minimize $\hat C^2 \overline
B_\mu^2+V_\mu$ over the one-dimensional parameter $\mu$.  We can also minimize
other criteria, as in Remark \ref{other_criteria_remark} by choosing $\mu$ to
minimize other functions of worst-case bias $\hat C \overline B_\mu$ and variance
$V_\mu$.

\subsection{No additional covariates}\label{no_additional_covariates_computation_sec}

In the case where there are no additional covariates, the nuclear norm
regularized ``partialling out'' regression (\ref{Pi_optimization}) reduces to
\begin{align}\label{X_Pi_optimization_no_Z}
  \min_{\Pi} \|X-\Pi\|_F^2/2
  + \mu \|\Pi\|_*.
\end{align}
The solution $\Pi^*_\mu$ can then be computed using soft thresholding on the
singular values of $X$.  We describe the solution here, and refer to
\citet[][Lemma S.1]{moon2018nuclear} for a detailed derivation.

Let the singular value decomposition of $X$ be given by $X=V_XS_XW_X'$ where $V_X$ is an
$N\times N$ orthogonal matrix (i.e. $V_X'V_X=I_N$), $W_X$ is a $T\times T$ orthogonal
matrix (i.e. $W_X'W_X=I_T$) and $S_X$ is a $N\times T$ rectangular diagonal
matrix, with $j$-th diagonal element given by the $j$-th singular value $s_j(X)$
of $X$.
Let $\tilde S_X(\mu)$ be the $N\times T$ diagonal matrix with $j$-th diagonal element
given by $\max\{s_j(X)-\mu,0\}$ (i.e. we perform soft thresholding on the $j$-th
singular value).

Then the solution $\Pi^*_\mu$ to
(\ref{X_Pi_optimization_no_Z}) and residuals $\Omega^*_\mu=X-\Pi^*_\mu$ are
given by
\begin{align*}
  \Pi^*_\mu = V_X\tilde S_X(\mu)W_X',
  \quad
  \Omega^*_\mu = V_X(S_X-\tilde S_X(\mu))W_X',
\end{align*}
Note that $S_X-\tilde S_X(\mu)$ is a $N\times T$ diagonal matrix with $j$-th
diagonal element given by $\min\{s_j(X),\mu\}$.  Thus, the weights $\tilde
A^*_\mu=\Omega^*_\mu/\langle \Omega^*_\mu, X \rangle_F$ used in the estimator
$\hat\beta=\langle \tilde A^*_{\mu}, \tilde Y \rangle_F$ given in
(\ref{beta_mu_equation}) can be obtained by replacing the singular values
$s_j(X)$ that are larger than $\mu$ with the constant $\mu$, and then dividing by
the constant
  $\langle \Omega^*_\mu, X \rangle_F
  = \langle S_X-\tilde S_X(\mu), S_X \rangle_F
  =\sum_{j=1}^{\min\{N,T\}}\min\{s_j(X),\mu\}s_j(X)$.
\section{Additional Numerical Results}
\label{sec: MC app}
\subsection{Additional simulation results for Section \ref{ssec: MC}}
\label{ssec: additional simulation results}

In this section, we provide additional simulation results for the numerical experiment described in Section~\ref{ssec: MC}. For brevity, here we only report results for $R = 1$. Tables \ref{tab: app mc N50} and \ref{tab: app mc N300} report the same statistics as Table \ref{tab: N100 R1} in the main text but for smaller and bigger samples sizes $N = 50$ and $N = 300$, respectively. The results are similar to the ones presented in Section~\ref{ssec: MC}. We find that, when there is a weak factor, our method effectively reduces the bias and improves estimation precision in smaller sample sizes too, i.e., also for $N = 50$. The gains from using our method are even more considerable for $N = 300$, which is consistent with our estimator having a faster rate of convergence than the LS estimator.

\begin{table}[H]
\caption{Simulation results for the experiment in Section \ref{ssec: MC}, $N = 50$, $R = 1$}
\label{tab: app mc N50}
\begin{center}
\resizebox{\textwidth}{!}{
\begin{tabular}{c c c c c c c| c c c c c c}
&\multicolumn{6}{c}{LS}&\multicolumn{6}{c}{Debiased} \\
\cmidrule(lr){2-7}  \cmidrule(lr){8-13}
{$\kappa$}&{bias}&{std}&{rmse}&{size}&{length}&{length*}&{bias}&{std}&{rmse}&{size}&{length}&{length*}\\
\midrule
\multicolumn{13}{c}{$T = 20$}\\
\midrule
0.00&-0.0006&0.0242&0.0242&7.4&0.086&0.348&-0.0007&0.0300&0.0300&0.0&0.364&0.184\\
0.05&0.0233&0.0249&0.0340&22.5&0.087&0.349&0.0088&0.0302&0.0314&0.0&0.364&0.184 \\
0.10&0.0466&0.0268&0.0538&56.5&0.087&0.351&0.0177&0.0310&0.0357&0.0&0.365&0.185 \\
0.15&0.0683&0.0309&0.0750&78.5&0.089&0.357&0.0252&0.0327&0.0413&0.0&0.367&0.186 \\
0.20&0.0847&0.0401&0.0937&83.8&0.091&0.368&0.0293&0.0361&0.0465&0.0&0.370&0.186 \\
0.25&0.0879&0.0555&0.1040&76.0&0.097&0.390&0.0281&0.0406&0.0494&0.0&0.372&0.187 \\
0.50&0.0115&0.0398&0.0414&12.1&0.122&0.493&0.0025&0.0359&0.0360&0.0&0.380&0.189 \\
1.00&0.0004&0.0330&0.0330&6.1&0.124&0.499&-0.0006&0.0346&0.0346&0.0&0.381&0.189 \\
\midrule
\multicolumn{13}{c}{$T = 50$}\\
\midrule
0.00&0.0003&0.0147&0.0147&6.0&0.055&0.279&-0.0001&0.0211&0.0211&0.0&0.230&0.118\\
0.05&0.0247&0.0152&0.0290&44.0&0.055&0.280&0.0070&0.0212&0.0223&0.0&0.230&0.118\\
0.10&0.0487&0.0169&0.0515&87.9&0.055&0.282&0.0134&0.0218&0.0256&0.0&0.231&0.118\\
0.15&0.0703&0.0214&0.0734&95.9&0.056&0.287&0.0173&0.0238&0.0294&0.0&0.232&0.119\\
0.20&0.0784&0.0368&0.0866&86.3&0.060&0.306&0.0158&0.0267&0.0310&0.0&0.234&0.119\\
0.25&0.0583&0.0510&0.0774&59.2&0.068&0.344&0.0097&0.0273&0.0290&0.0&0.236&0.119\\
0.50&0.0035&0.0207&0.0210&6.5&0.078&0.398&0.0003&0.0236&0.0236&0.0&0.237&0.120 \\
1.00&0.0003&0.0201&0.0201&5.0&0.078&0.399&-0.0003&0.0234&0.0234&0.0&0.237&0.120\\
\midrule
\multicolumn{13}{c}{$T = 100$}\\
\midrule
0.00&0.0002&0.0105&0.0105&6.0&0.039&0.229&-0.0001&0.0137&0.0137&0.0&0.173&0.080\\
0.05&0.0247&0.0109&0.0270&68.4&0.039&0.229&0.0064&0.0138&0.0152&0.0&0.173&0.080\\
0.10&0.0487&0.0124&0.0502&98.1&0.039&0.231&0.0120&0.0143&0.0187&0.0&0.174&0.080\\
0.15&0.0684&0.0188&0.0709&97.3&0.040&0.238&0.0134&0.0165&0.0212&0.0&0.175&0.080\\
0.20&0.0577&0.0390&0.0697&72.5&0.046&0.271&0.0082&0.0179&0.0197&0.0&0.177&0.081\\
0.25&0.0225&0.0300&0.0375&33.9&0.053&0.310&0.0031&0.0164&0.0167&0.0&0.177&0.081\\
0.50&0.0016&0.0145&0.0146&5.6&0.055&0.325&0.0001&0.0153&0.0153&0.0&0.177&0.081 \\
1.00&0.0001&0.0143&0.0143&5.1&0.055&0.326&-0.0001&0.0152&0.0152&0.0&0.177&0.081\\
\midrule
\multicolumn{13}{c}{$T = 300$}\\
\midrule
0.00&-0.0001&0.0059&0.0059&5.8&0.023&0.170&-0.0002&0.0077&0.0077&0.0&0.127&0.049\\
0.05&0.0245&0.0066&0.0254&96.7&0.023&0.171&0.0060&0.0078&0.0098&0.0&0.127&0.049 \\
0.10&0.0481&0.0088&0.0489&99.8&0.023&0.173&0.0100&0.0089&0.0134&0.0&0.128&0.049 \\
0.15&0.0482&0.0271&0.0553&83.6&0.026&0.196&0.0059&0.0103&0.0119&0.0&0.129&0.049 \\
0.20&0.0117&0.0143&0.0185&33.3&0.031&0.234&0.0015&0.0090&0.0091&0.0&0.129&0.049 \\
0.25&0.0047&0.0093&0.0104&12.6&0.032&0.239&0.0006&0.0087&0.0087&0.0&0.129&0.049 \\
0.50&0.0004&0.0084&0.0085&5.6&0.032&0.241&-0.0001&0.0086&0.0086&0.0&0.129&0.049 \\
1.00&-0.0001&0.0084&0.0084&5.5&0.032&0.241&-0.0002&0.0086&0.0086&0.0&0.129&0.049\\
\bottomrule
\bottomrule
\multicolumn{13}{p{1.2\textwidth}}{${\rm Lind}(A) \in \{0.0109,0.0049,0.0028,0.0011\}$ for $T \in \{20,50,100,300\}$. The results are based on 5,000 simulations.} 
\end{tabular}
}
\end{center}
\end{table}

\begin{table}[H]
\caption{Simulation results for the experiment in Section \ref{ssec: MC}, $N = 300$, $R = 1$}
\label{tab: app mc N300}
\begin{center}
\resizebox{\textwidth}{!}{
\begin{tabular}{c c c c c c c| c c c c c c}
&\multicolumn{6}{c}{LS}&\multicolumn{6}{c}{Debiased} \\
\cmidrule(lr){2-7}  \cmidrule(lr){8-13}
{$\kappa$}&{bias}&{std}&{rmse}&{size}&{length}&{length*}&{bias}&{std}&{rmse}&{size}&{length}&{length*}\\
\midrule
\multicolumn{13}{c}{$T = 20$}\\
\midrule
0.00&0.0001&0.0096&0.0096&6.4&0.036&0.219&0.0001&0.0115&0.0115&0.0&0.246&0.088 \\
0.05&0.0242&0.0106&0.0264&72.7&0.036&0.220&0.0091&0.0117&0.0148&0.0&0.246&0.088\\
0.10&0.0474&0.0136&0.0493&96.6&0.036&0.223&0.0169&0.0127&0.0211&0.0&0.247&0.088\\
0.15&0.0633&0.0235&0.0675&93.0&0.038&0.233&0.0192&0.0159&0.0249&0.0&0.249&0.088\\
0.20&0.0475&0.0382&0.0610&65.4&0.044&0.267&0.0120&0.0182&0.0218&0.0&0.250&0.088\\
0.25&0.0192&0.0276&0.0336&31.4&0.049&0.297&0.0047&0.0155&0.0162&0.0&0.251&0.088\\
0.50&0.0015&0.0134&0.0135&6.3&0.051&0.310&0.0004&0.0134&0.0134&0.0&0.252&0.089 \\
1.00&0.0002&0.0132&0.0132&5.6&0.051&0.310&0.0001&0.0133&0.0133&0.0&0.252&0.089 \\
\midrule
\multicolumn{13}{c}{$T = 50$}\\
\midrule
0.00&-0.0001&0.0060&0.0060&5.8&0.023&0.173&-0.0002&0.0078&0.0078&0.0&0.127&0.048\\
0.05&0.0246&0.0067&0.0254&96.8&0.023&0.173&0.0060&0.0079&0.0099&0.0&0.127&0.048 \\
0.10&0.0482&0.0090&0.0490&99.8&0.023&0.175&0.0100&0.0090&0.0134&0.0&0.128&0.049 \\
0.15&0.0478&0.0272&0.0550&83.6&0.026&0.199&0.0058&0.0103&0.0118&0.0&0.129&0.049 \\
0.20&0.0117&0.0144&0.0186&32.5&0.031&0.237&0.0014&0.0090&0.0091&0.0&0.129&0.049 \\
0.25&0.0047&0.0093&0.0105&12.8&0.032&0.243&0.0005&0.0088&0.0088&0.0&0.129&0.049 \\
0.50&0.0004&0.0085&0.0085&5.7&0.032&0.245&-0.0001&0.0087&0.0087&0.0&0.129&0.049 \\
1.00&-0.0001&0.0084&0.0084&5.6&0.032&0.245&-0.0002&0.0086&0.0086&0.0&0.129&0.049\\
\midrule
\multicolumn{13}{c}{$T = 100$}\\
\midrule
0.00&0.0001&0.0041&0.0041&5.0&0.016&0.123&0.0001&0.0055&0.0055&0.0&0.080&0.033  \\
0.05&0.0248&0.0046&0.0253&100.0&0.016&0.123&0.0047&0.0056&0.0073&0.0&0.080&0.033\\
0.10&0.0482&0.0071&0.0488&99.9&0.016&0.125&0.0056&0.0067&0.0088&0.0&0.080&0.033 \\
0.15&0.0178&0.0172&0.0248&61.1&0.021&0.163&0.0015&0.0063&0.0065&0.0&0.081&0.033 \\
0.20&0.0048&0.0064&0.0080&16.3&0.022&0.172&0.0006&0.0060&0.0061&0.0&0.081&0.033 \\
0.25&0.0024&0.0059&0.0064&7.6&0.023&0.173&0.0003&0.0060&0.0060&0.0&0.081&0.033  \\
0.50&0.0004&0.0057&0.0057&4.8&0.023&0.174&0.0001&0.0060&0.0060&0.0&0.081&0.033  \\
1.00&0.0001&0.0057&0.0057&4.9&0.023&0.174&0.0001&0.0060&0.0060&0.0&0.081&0.033  \\
\midrule
\multicolumn{13}{c}{$T = 300$}\\
\midrule
0.00&-0.0000&0.0024&0.0024&5.5&0.009&0.105&-0.0001&0.0036&0.0036&0.0&0.043&0.018\\
0.05&0.0249&0.0028&0.0250&100.0&0.009&0.106&0.0030&0.0037&0.0048&0.0&0.043&0.018\\
0.10&0.0310&0.0169&0.0352&95.9&0.011&0.123&0.0011&0.0040&0.0042&0.0&0.043&0.018 \\
0.15&0.0036&0.0037&0.0051&21.8&0.013&0.148&0.0002&0.0039&0.0039&0.0&0.044&0.018 \\
0.20&0.0014&0.0034&0.0037&7.5&0.013&0.149&0.0001&0.0039&0.0039&0.0&0.044&0.018  \\
0.25&0.0007&0.0033&0.0034&5.2&0.013&0.149&-0.0000&0.0039&0.0039&0.0&0.044&0.018 \\
0.50&0.0000&0.0033&0.0033&4.5&0.013&0.149&-0.0000&0.0039&0.0039&0.0&0.044&0.018 \\
1.00&-0.0000&0.0033&0.0033&4.5&0.013&0.149&-0.0001&0.0039&0.0039&0.0&0.044&0.018\\
\bottomrule
\bottomrule
\multicolumn{13}{p{1.2\textwidth}}{${\rm Lind}(A) \in \{0.0025,0.0011,0.0006,0.0002\}$ for $T \in \{20,50,100,300\}$. The results are based on 5,000 simulations.} 
\end{tabular}
}
\end{center}
\end{table}

\subsection{A design with an additional covariate and heteroskedastic serially correlated errors}
\label{ssec: MC extended}
Similarly to Section \ref{ssec: MC}, we consider
\begin{align*}
  Y_{it} &= X_{it} \beta + Z_{it} \delta + \kappa \lambda_i f_t + U_{it},\\
  X_{it} &= \lambda_i f_t + V_{it}^X, \\
  Z_{it} &= \lambda_i f_t + V_{it}^Z,
\end{align*}
where $\lambda_i$, $f_t$, and $(V_{it}^X,V_{it}^Z)$ are all mutually independent across $i$, $t$, and $(i,t)$, and
\begin{align*}
  \lambda_i \sim N(0,1) \perp f_t \sim N(0,1) \perp  \begin{pmatrix} V_{it}^X \\ V_{it}^Z \end{pmatrix} \sim N \left(\begin{pmatrix} 0 \\ 0 \end{pmatrix}, \begin{pmatrix} \sigma_V^2 & \rho \sigma_{V}^2 \\ \rho \sigma_V^2 & \sigma_V^2 \end{pmatrix}\right).
\end{align*}
Conditional on $\lambda_i$, $f_t$, and $(V_{it}^X,V_{it}^Z)$ for $i \in \{1, \ldots, N\}$ and $t \in \{0. \ldots, T\}$, we construct serially correlated errors $U_{it}$ as
\begin{align*}
  U_{it} = \varepsilon_{it} + \theta_{\varepsilon} \varepsilon_{it-1},
\end{align*}
where $\varepsilon_{it} = \sigma_{\varepsilon}(X_i,Z_i,\lambda_i,f_t) e_{it}$ and $e_{it}$ are independently drawn from a scaled Student's t-distribution with 5 degrees of freedom and variance normalized to 1. The conditional variance of $\varepsilon_{it}$ is given by
\begin{align*}
  \sigma_{\varepsilon}^2(X_i,Z_i,\lambda_i,f_t) = \frac{\sigma_U^2}{1 + \theta_{\varepsilon}^2} \left(\frac{1}{2} + \Lambda \left(\frac{X_{it} + Z_{it} + \lambda_i f_t}{3} \right) \right),
\end{align*}
where $\Lambda (\cdot)$ stands for the logistic CDF.

In this experiment, we fix $(\beta,\delta,\sigma_U^2,\sigma_V^2,\rho_V,\theta_{\varepsilon}) = (0,1,1,1,1/\sqrt{2},1/\sqrt{2})$. As in Section \ref{sec: numerical}, we compare the performance of the standard LS based approach and our approach for various values of $\kappa$. Both approaches are implemented with $R = 1$ and $R = 2$, i.e., when the number of factors is correctly specified and when it is overspecified. To adjust for serial correlation of $U_{it}$, we compute HAC standard errors for the LS estimator accounting for up to two lags of serial correlation, and we compute clustered standard errors for our approach allowing for an arbitrary form of serial correlation.

For brevity, we only report results for estimation and inference on $\beta$ and for $(N,T) = (100,50)$ in Table \ref{tab: extended MC} below. The results are qualitatively similar to the ones presented in Section \ref{sec: numerical}. Our findings suggest that the standard LS based approach might perform poorly when there is a weak factor, and our approach reduces the weak factors bias and improves the quality of point estimation in more complicated settings with additional covariates and non-Gaussian, heteroskedastic, and serially correlated errors. Importantly, when the number of factors is overspecified ($R = 2$), the bias of the LS estimator does not vanish, so one cannot protect themselves from the weak factors bias by simply conservatively overspecifying $R$. At the same time, overspecifying the number of factors does not result in loss of efficiency for our estimator. Nonetheless, overspecifying $R$ leads to wider confidence intervals. While this might seem as a disadvantage of our approach, the necessity of having wider confidence intervals when the true number of factors is unknown has been established in \citet{zhu2019well}. Specifically, \citet{zhu2019well} shows that the uncertainty in the number of factors necessarily results in a dramatic loss of inference efficiency when robustness to weak factors is required.

\begin{table}[H]
\caption{Simulation results for the experiment in Section \ref{ssec: MC extended}}
\label{tab: extended MC}
\begin{center}
\resizebox{\textwidth}{!}{
\begin{tabular}{c c c c c c c| c c c c c c}
&\multicolumn{6}{c}{LS}&\multicolumn{6}{c}{Debiased} \\
\cmidrule(lr){2-7}  \cmidrule(lr){8-13}
{$\kappa$}&{bias}&{std}&{rmse}&{size}&{length}&{length*}&{bias}&{std}&{rmse}&{size}&{length}&{length*}\\
\midrule
\multicolumn{13}{c}{$R=1$}\\
\midrule
0.00&-0.0001&0.0195&0.0195&6.2&0.073&0.189&-0.0003&0.0208&0.0209&0.0&0.295&0.145\\
0.05&0.0132&0.0195&0.0236&12.3&0.073&0.189&0.0087&0.0209&0.0226&0.0&0.295&0.145 \\
0.10&0.0264&0.0197&0.0329&29.9&0.073&0.190&0.0176&0.0210&0.0274&0.0&0.295&0.146 \\
0.15&0.0393&0.0201&0.0442&55.1&0.073&0.191&0.0260&0.0212&0.0336&0.0&0.296&0.147 \\
0.20&0.0512&0.0213&0.0555&76.0&0.074&0.192&0.0332&0.0220&0.0398&0.0&0.298&0.147 \\
0.25&0.0543&0.0286&0.0614&75.9&0.075&0.194&0.0337&0.0255&0.0423&0.0&0.301&0.148 \\
0.50&0.0018&0.0218&0.0219&6.2&0.078&0.203&0.0015&0.0224&0.0224&0.0&0.309&0.149  \\
1.00&-0.0000&0.0204&0.0204&5.6&0.078&0.203&-0.0002&0.0218&0.0218&0.0&0.309&0.149\\
\midrule
\multicolumn{13}{c}{$R=2$}\\
\midrule
0.00&-0.0001&0.0196&0.0196&6.9&0.071&0.184&-0.0003&0.0210&0.0210&0.0&0.490&0.140\\
0.05&0.0131&0.0197&0.0237&13.7&0.071&0.185&0.0087&0.0210&0.0227&0.0&0.491&0.140 \\
0.10&0.0263&0.0199&0.0330&32.0&0.071&0.185&0.0174&0.0212&0.0274&0.0&0.491&0.140 \\
0.15&0.0389&0.0204&0.0439&55.8&0.072&0.186&0.0255&0.0215&0.0333&0.0&0.493&0.141 \\
0.20&0.0494&0.0221&0.0542&73.9&0.072&0.187&0.0314&0.0224&0.0386&0.0&0.495&0.142 \\
0.25&0.0476&0.0304&0.0565&67.4&0.073&0.190&0.0288&0.0262&0.0389&0.0&0.499&0.142 \\
0.50&0.0015&0.0207&0.0208&6.9&0.076&0.197&0.0011&0.0219&0.0220&0.0&0.506&0.143  \\
1.00&0.0000&0.0206&0.0206&6.4&0.076&0.197&-0.0002&0.0219&0.0219&0.0&0.507&0.143 \\
\bottomrule
\bottomrule
\multicolumn{13}{p{1.2\textwidth}}{The results are based on 5,000 simulations.} 
\end{tabular}
}
\end{center}
\end{table}

\subsection{Coverage of non-robust CIs based on the debiased estimator}
\label{ssec: MC non-robust coverage}
In this section, we present additional simulation results for the numerical experiment considered in Section~\ref{ssec: MC}.

Specifically, in Table~\ref{tab: non-robust coverage} below, we report the size of the t-test (with nominal size 5\%) based on the non-robust CI provided in \eqref{eq:non_bias_aware_ci} (with $95\%$ nominal coverage) and its average length. As before, we also report the same statistics for the LS CI.

We find that the non-bias aware CIs based on the debiased estimator are comparable in terms of the length to the LS CIs and that their coverage is close to 95\% in the absence of weak factors. This is in line with our asymptotic analysis provided in Section~\ref{ssec: semi strong}. At the same time, even if there is a weak factor, our non-robust CIs enjoy much better coverage than the LS ones. For example, for $(N,T) = (300,300)$ and $\kappa =0.10$, our non-robust CI has coverage 91.4\% whereas the coverage of the LS CI is only 4.1\%.

\begin{table}[H]
\caption{Simulation results for the non-robust CI \eqref{eq:non_bias_aware_ci}}
\label{tab: non-robust coverage}
\begin{small}
\begin{center}
\begin{tabular}{c c c c c| c c c c| c c c c}
&\multicolumn{4}{c|}{$(N,T) = (100,100)$}&\multicolumn{4}{c|}{$(N,T) = (300,100)$}&\multicolumn{4}{c}{$(N,T) = (300,300)$}\\
&\multicolumn{2}{c}{LS}&\multicolumn{2}{c|}{Debiased}&\multicolumn{2}{c}{LS}&\multicolumn{2}{c|}{Debiased}&\multicolumn{2}{c}{LS}&\multicolumn{2}{c}{Debiased} \\
{$\kappa$}&{size}&{length}&{size}&{length}&{size}&{length}&{size}&{length}&{size}&{length}&{size}&{length}\\
\midrule
\midrule
0.00&6.1&0.028&5.9&0.041&5.0&0.016&5.8&0.021&5.5&0.009&5.2&0.014\\
0.05&91.0&0.028&8.9&0.041&100.0&0.016&13.7&0.021&100.0&0.009&13.1&0.014\\
0.10&99.9&0.028&17.0&0.041&99.9&0.016&23.3&0.021&95.9&0.011&8.6&0.014\\
0.15&92.9&0.030&17.5&0.041&61.1&0.021&10.0&0.021&21.8&0.013&6.8&0.014\\
0.20&47.4&0.037&10.1&0.041&16.3&0.022&7.8&0.021&7.5&0.013&6.8&0.014\\
0.25&17.9&0.039&8.5&0.041&7.6&0.023&7.7&0.021&5.2&0.013&6.8&0.014\\
0.50&5.9&0.039&8.0&0.041&4.8&0.023&7.5&0.021&4.5&0.013&6.8&0.014\\
1.00&5.4&0.039&8.0&0.041&4.9&0.023&7.6&0.022&4.5&0.013&6.8&0.014\\
\bottomrule
\bottomrule
\multicolumn{13}{l}{The results are based on 5,000 simulations.} 
\end{tabular}
\end{center}
\end{small}
\end{table}

\section{Additional Results for Empirical Illustration}
\label{sec: empirical app}
In his section, we revisit the empirical application considered in Section~\ref{ssec: empirical} by considering an alternative specification with dynamic treatment effects. Specifically, in the spirit of \citet{Wolfers2006}, \citet{KimOka2014} and \citet{MoonWeidner2015}, we consider
\begin{align*}
  Y_{it} = \sum_{k=1}^4 X_{k,it} \beta_k + \alpha_i + \zeta_i t + \nu_i t^2 + \phi_t + \sum_{r=1}^R \lambda_{ir} f_{tr} + U_{it},
\end{align*}
where $X_{k,it}$ are the treatment dummies defined as
\begin{align*}
  X_{k,it} &= \mathbf 1 \{D_i + 4(k-1) \leq t \leq D_i + 4k - 1 \} \quad \text{for} \quad k \in \{1, \ldots, 3\},\\
  X_{4,it} &= \mathbf 1 \{D_i + 12 \leq t\}, 
\end{align*}
where $D_i$ denotes the year in which state $i$ adopted a unilateral divorce law. Here, instead of introducing bi-annual dummies as in \citet{Wolfers2006}, we consider a coarser dynamics of treatment effects to ensure that our regressors $X_{k,it}$ have sufficient variation necessary for debiasing. %

As before, we estimate and construct 95\% CIs for $\beta_k$ using the LS and our approaches. The results are provided in Table \ref{tab: emp res app} below. They are qualitatively similar to the results reported in Section \ref{ssec: empirical}. While it is also possible to obtain shorter confidence intervals and establish significance of certain dynamic effects using our approach in the absence of weak factors, we again find that the potential presence of one weak factors is sufficient to render the estimated effects insignificant.

\begin{table}[H]
  \caption{LS and debiased estimates and 95\% CIs for dynamic effects of divorce law reform}
  \label{tab: emp res app}
\begin{footnotesize}
\begin{center}
\begin{tabular}{l c c c c c c c c c}
& $R=1$& $R=2$& $R=3$& $R=4$& $R=5$& $R=6$ \\ 
\midrule
\multicolumn{7}{c}{LS}\\ 
\midrule
years 1-4  & $0.033$& $0.084$& $0.093$& $0.040$& $0.012$& $0.085$ \\ 
& $[-0.07,0.14]$& $[-0.04,0.21]$& $[-0.03,0.21]$& $[-0.08,0.16]$& $[-0.11,0.14]$& $[-0.04,0.21]$ \\ 
years 5-8& $-0.081$& $-0.026$& $0.025$& $-0.028$& $-0.078$& $0.088$ \\ 
& $[-0.23,0.07]$& $[-0.20,0.15]$& $[-0.15,0.20]$& $[-0.19,0.14]$& $[-0.24,0.09]$& $[-0.06,0.24]$ \\ 
years 9-12& $-0.253$& $-0.247$& $-0.186$& $-0.244$& $-0.297$& $-0.065$ \\ 
& $[-0.45,-0.05]$& $[-0.49,-0.00]$& $[-0.41,0.04]$& $[-0.46,-0.03]$& $[-0.50,-0.09]$& $[-0.26,0.13]$ \\ 
years 13+& $-0.198$& $-0.255$& $-0.234$& $-0.313$& $-0.365$& $-0.110$ \\ 
& $[-0.46,0.06]$& $[-0.54,0.03]$& $[-0.50,0.03]$& $[-0.57,-0.06]$& $[-0.61,-0.12]$& $[-0.35,0.13]$ \\ 
\midrule
\multicolumn{7}{c}{Debiased}\\ 
\midrule
years 1-4  & $0.081$& $0.147$& $0.137$& $0.098$& $0.079$& $0.118$ \\ 
$R_w = 0$& $[-0.03,0.19]$& $[0.05,0.25]$& $[0.05,0.22]$& $[0.02,0.18]$& $[0.00,0.16]$& $[0.05,0.19]$ \\ 
$R_w = 1$& $[-0.80,0.96]$& $[-0.58,0.88]$& $[-0.45,0.72]$& $[-0.39,0.59]$& $[-0.33,0.49]$& $[-0.23,0.47]$ \\ 
$R_w = R$& $[-0.80,0.96]$& $[-1.21,1.51]$& $[-1.45,1.72]$& $[-1.63,1.83]$& $[-1.64,1.80]$& $[-1.63,1.86]$ \\ 
years 5-8& $-0.008$& $0.054$& $0.099$& $0.056$& $0.031$& $0.125$ \\ 
$R_w = 0$& $[-0.17,0.15]$& $[-0.08,0.19]$& $[-0.01,0.21]$& $[-0.05,0.16]$& $[-0.07,0.13]$& $[0.04,0.22]$ \\ 
$R_w = 1$& $[-1.34,1.33]$& $[-1.04,1.14]$& $[-0.77,0.97]$& $[-0.67,0.79]$& $[-0.57,0.63]$& $[-0.39,0.64]$ \\ 
$R_w = R$& $[-1.34,1.33]$& $[-2.00,2.10]$& $[-2.30,2.50]$& $[-2.56,2.67]$& $[-2.57,2.63]$& $[-2.51,2.76]$ \\ 
years 9-12& $-0.147$& $-0.139$& $-0.098$& $-0.126$& $-0.157$& $0.003$ \\ 
$R_w = 0$& $[-0.36,0.06]$& $[-0.33,0.05]$& $[-0.25,0.05]$& $[-0.27,0.02]$& $[-0.29,-0.02]$& $[-0.12,0.13]$ \\ 
$R_w = 1$& $[-2.04,1.75]$& $[-1.70,1.43]$& $[-1.34,1.15]$& $[-1.17,0.92]$& $[-1.01,0.70]$& $[-0.73,0.74]$ \\ 
$R_w = R$& $[-2.04,1.75]$& $[-3.08,2.81]$& $[-3.54,3.34]$& $[-3.88,3.63]$& $[-3.89,3.58]$& $[-3.78,3.79]$ \\ 
years 13+& $-0.178$& $-0.228$& $-0.196$& $-0.225$& $-0.262$& $-0.071$ \\ 
$R_w = 0$& $[-0.45,0.09]$& $[-0.46,0.01]$& $[-0.38,-0.01]$& $[-0.40,-0.05]$& $[-0.44,-0.09]$& $[-0.23,0.09]$ \\ 
$R_w = 1$& $[-2.71,2.35]$& $[-2.31,1.86]$& $[-1.85,1.46]$& $[-1.62,1.17]$& $[-1.40,0.88]$& $[-1.05,0.90]$ \\ 
$R_w = R$& $[-2.71,2.35]$& $[-4.16,3.71]$& $[-4.80,4.40]$& $[-5.26,4.80]$& $[-5.26,4.74]$& $[-5.14,4.99]$ \\ 
\bottomrule
\bottomrule
\end{tabular}
\end{center}
\end{footnotesize}
\end{table}

\end{document}